\documentclass[10pt,a4paper]{amsart}

\numberwithin{equation}{section}

\usepackage{amsmath}
\usepackage{amsthm}
\usepackage{amssymb}
\usepackage{amsfonts}
\usepackage{graphicx}

\usepackage{hyperref}
\usepackage{float}
\usepackage{color}
\def\f12{\frac 1 2}

\def\z{\zeta}
\def\a{\alpha}
\def\b{\beta}

\def\hh{\mathcal{H}^{+}}

\def\f12{\frac 1 2}

\newcommand{\nabb}{\mbox{$\nabla \mkern-13mu /$\,}}
\newcommand{\lapp}{\mbox{$\triangle \mkern-13mu /$\,}}
\newtheorem{remark}{Remark}[section]
\newtheorem{lemma}{Lemma}[subsection]
\newtheorem{theorem}{Theorem}[section]
\newtheorem{proposition}{Proposition}[subsection]

\newtheorem{corollary}{Corollary}[subsection]

\newtheorem{mytheo}{Theorem}
\begin{document}

\title[Linear Stability and Instability Of Extreme Reissner-Nordstr\"{o}m II]{Stability and Instability of Extreme Reissner-Nordstr\"{o}m Black Hole Spacetimes for Linear Scalar Perturbations II}
% Remove command to get current date 
\author[Stefanos Aretakis]{Stefanos Aretakis$^*$}\thanks{$^*$University of Cambridge,
Department of Pure Mathematics and Mathematical Statistics,
Wilberforce Road, Cambridge, CB3 0WB, United Kingdom}
\date{February 20, 2011}

\maketitle

\begin{abstract}

This paper contains the second part of a two-part series on the stability and instability of extreme Reissner-Nordstr\"{o}m spacetimes  for  linear scalar perturbations. We  continue our study of solutions to the linear wave equation $\Box_{g}\psi=0$ on a suitable globally hyperbolic subset of such a  spacetime, arising from regular initial data prescribed on a Cauchy hypersurface $\Sigma_{0}$ crossing the future event horizon $\mathcal{H}^{+}$.  We here obtain definitive energy and pointwise decay,  non-decay  and blow-up results. Our estimates hold up to and including the horizon $\mathcal{H}^{+}$. A hierarchy of conservations laws on degenerate horizons is also derived.
\end{abstract}

\section{Introduction}
\label{sec:Introduction}

In this paper we shall attempt to  provide a complete picture of the  stability and instability of extreme Reissner-Nordstr\"{o}m backgrounds under linear scalar perturbations, extending the results of \cite{aretakis1} on the wave equation
\begin{equation}
\Box_{g}\psi=0.
\label{we}
\end{equation} Here we derive definitive energy and pointwise decay, non-decay and blow-up results for solutions $\psi$ to the wave equation and their derivatives in the domain of outer communications up to and including the event horizon $\mathcal{H}^{+}$. Note that the picture we obtain is in sharp contrast  with the non-extreme case where decay holds for all higher order derivatives of $\psi$ along $\hh$.

\subsection{Preliminaries}
\label{sec:Preliminaries}
The extreme Reissner-Nordstr\"{o}m metric in coordinates $(v,r,\theta,\phi)\in\mathbb{R}\times\mathbb{R}^{+}\times\mathbb{S}^{2}$ takes the form
\begin{equation*}
g=-Ddv^{2}+2dvdr+r^{2}g_{\mathbb{S}^{2}},
\end{equation*}
where $M$ is a positive constant, $D=\left(1-\frac{M}{r}\right)^{2}$ and $g_{\mathbb{S}^{2}}$ is the standard metric on $\mathbb{S}^{2}$. The event horizon $\mathcal{H}^{+}$ corresponds to $r=M$, the black hole region to $r\leq M$ and the domain of outer communications to $r>M$. The photon sphere is located at $r=2M$.

For the reader's convenience we recall that  the main results of \cite{aretakis1} include: 1) non-degenerate energy and pointwise uniform boundedness of solutions, up to and including  $\mathcal{H}^{+}$, 2) local integrated decay of energy, up to and including the event horizon $\mathcal{H}^{+}$, 3)  sharp second order $L^{2}$ estimates, up to and including $\mathcal{H}^{+}$, 4) non-decay along $\mathcal{H}^{+}$  of higher order translation invariant quantities for spherically symmetric solutions.

Recall also that  $L^{2}$ spacetime estimates which do not degenerate at the photon sphere require commutation with the Killing vector field $T=\partial_{v}$. This is the so-called trapping effect at the photon sphere.  Note that another characteristic feature of degenerate horizons which was exposed in \cite{aretakis1} is that obtaining $L^{2}$ spacetime estimates which do not degenerate at the horizon $\mathcal{H}^{+}$ requires commutation with the (non-Killing) vector field $\partial_{r}$ and, therefore,  loss of derivatives characteristic of  trapping  takes place on degenerate horizons in analogy to the photon sphere.

\subsection{Overview of Results and Techniques}
\label{sec:OverviewOfResultsAndTechniques}

In the present paper, we combine the previous results of \cite{aretakis1} with certain  new techniques to obtain definitive decay, non-decay and blow-up results. In particular, we present a method based on an adaptation of \cite{new} to derive degenerate and non-degenerate energy decay. This adaptation  requires the introduction of yet another vector field $P$ and is necessary in view of the degeneracy of the  surface gravity on the horizon. We also introduce a new method for obtaining sharp  pointwise decay results. The instability properties of $\psi$ (non-decay and blow-up for derivatives of $\psi$) rest upon a hierarchy of conservation laws on a specific class of degenerate horizons (which includes the extreme Reissner-Nordstr\"{o}m) presented here for the first time. As we shall see, these laws are of great analytical importance.

\subsubsection{Conservation Laws on $\mathcal{H}^{+}$}
\label{sec:ConservationLawsOnMathcalH}

Recall that in \cite{aretakis1} we derived a conservation law  for the spherical mean $\psi_{0}$ of solutions $\psi$ to the wave equation based on the degeneracy of the redshift along $\hh$. However, as we shall see, on top of the degeneracy of the redshift, the event horizon satisfies an additional property which allows us to obtain a hierarchy of such laws. Specifically, we show that \textit{a conservation law holds for every projection   $\psi_{l}$  of $\psi$ (viewed as an $L^{2}$ function on the spheres of symmetry) on the eigenspace $E^{l}$ of the spherical Laplacian  $\lapp$ (for all spherical harmonic numbers  $l\in\mathbb{N}$)}.

According to these laws a linear combination of the transversal derivatives of $\psi_{l}$ of order at most $l+1$ is conserved along the null geodesics of $\mathcal{H}^{+}$ (see Theorem \ref{t3} of Section \ref{sec:TheMainTheorems}). As we shall see, these conserved quantities  allow us to infer the instability properties of extreme black holes described in Section \ref{sec:HigherOrderPointwiseEstimates1}, and thus, understanding their structure is crucial and essential.  Of course, no such conserved quantities exist in the subextreme case.

As an aside, based on these  laws, we  also explicitly show that the Schwarzschild boundedness argument of Kay and Wald \cite{wa1} cannot be applied in the extreme case, i.e.~we show that for generic $\psi$, there does not exist a Cauchy hypersurface $\Sigma$ crossing $\mathcal{H}^{+}$ and a solution $\tilde{\psi}$ such that 
\begin{equation*}
T\tilde{\psi}=\psi
\end{equation*}
in the  causal future of $\Sigma$ (where $T=\partial_{v}$). The existence of such $\tilde{\psi}$ was key for the argument of \cite{wa1}. 

\subsubsection{Sharp Higher Order $L^{2}$ Estimates}
\label{sec:SharpHigherOrderL2Estimates}

\noindent We next establish higher order $L^{2}$ estimates of the derivatives of $\psi$ by commuting repeatedly with the vector field $\partial_{r}$; see Theorem \ref{theorem3} of Section \ref{sec:TheMainTheorems}. In view of the conservation laws one expects to derive $k'$th order ($k\geq 1)$ $L^{2}$ estimates close to $\mathcal{H}^{+}$ only if $\psi_{l}=0$ for all $l\leq k$. In fact in Section \ref{sec:HigherOrderEstimates} we show that if the above restriction on the frequency range is not satisfied then no such estimate can be derived.  By using  appropriate modifications and Hardy inequalities we obtain the sharpest possible result. See Section \ref{sharpl2}. Note that the spacetime term of such estimates degenerates with respect to the transversal derivative to $\mathcal{H}^{+}$. In order to retrieve this derivative one needs to commute once again with $\partial_{r}$,  use Hardy inequalities and thus assume that an initial quantity of even higher order is bounded. This reflects \textit{the higher order trapping effect present on $\mathcal{H}^{+}$} (recall that the case $k=1$ was treated in \cite{aretakis1}). The difficulty in deriving such $L^{2}$ estimates comes from the fact that the trapping effect is coupled with the low-frequency obstruction described in Section \ref{sec:ConservationLawsOnMathcalH}.

\subsubsection{Energy and Pointwise Decay}
\label{sec:EnergyAndPointwiseDecay}

\noindent Using an adaptation of the methods developed in the recent \cite{new}, we  obtain energy and pointwise decay for $\psi$. See Theorems \ref{t4} and \ref{t6} of Section \ref{sec:TheMainTheorems}. 

 Recall that in \cite{new}, a general framework is provided for obtaining energy decay. The ingredients necessary for applying the framework are: 1) good asymptotics of the metric towards null infinity, 2) uniform boundedness of energy and 3) integrated local energy decay (where the spacetime integral of energy should be controlled by the energy of $\psi$ and, in view of the trapping effect at the photon sphere, of  $T\psi$ too). We first verify that extreme  Reissner-Nordstr\"{o}m satisfies the first hypothesis. However, in view of the trapping and the conservation laws on the event horizon $\mathcal{H}^{+}$, it turns out that the method described in \cite{new} can not be directly used to yield decay results in the extreme case. Indeed, the third hypothesis of \cite{new} is not satisfied in extreme Reissner-Nordstr\"{o}m.  For this reason,  we introduce a new causal vector field $P$ which allows us to obtain \textit{several  hierarchies of estimates in an appropriate neighbourhood of $\mathcal{H}^{+}$}. These estimates avoid multipliers or commutators with weights in $t$, following the philosophy of \cite{new}. Our method applies to   black hole spacetimes where trapping is present on $\mathcal{H}^{+}$ (including, in particular, a wide class of extreme black holes).

Pointwise decay for $l\geq 2$ then  follows by commuting with the generators of so(3) and Sobolev estimates. Regarding the cases $l=0,1$, we present a new method which is based on the interpolation of previous estimates which hold close to $\hh$ and away from $\hh$.  Note that the low angular frequencies decay more slowly than the higher ones. See Section \ref{sec:PointwiseEstimates}.

 We finally mention that Blue and Soffer have previously proved in \cite{blu1} that a weighted $L^{6}$ norm in space decays like $t^{-\frac{1}{3}}$. However, this weight degenerates on the horizon and the initial data have to be supported away from $\hh$.

\subsubsection{Higher Order  Estimates: Energy and Pointwise Decay, Non-decay and Blow-up}
\label{sec:HigherOrderPointwiseEstimates1}

 In order to provide a complete picture of the behaviour of solutions $\psi$, it remains to derive pointwise estimates for all derivatives of $\psi$.  Let $\psi_{l}$ denote the projection of $\psi$ on the eigenspace $E^{l}$ of the spherical Laplacian  $\lapp$, as above. Then the derivatives transversal to $\mathcal{H}^{+}$  of $\psi_{l}$ decay if the order of the differentiation is at most $l$. If the order is $l+1$, then  for generic initial data this derivative converges along $\mathcal{H}^{+}$ to a non-zero number and thus \textbf{does not decay}. By generic initial data we mean data for which certain quantities do not vanish on $\mathcal{H}^{+}$. If, moreover, the order is at least $l+2$, then for generic initial data these derivatives \textbf{blow up} asymptotically along $\mathcal{H}^{+}$.  Note that these differential operators are translation invariant and do not depend on the choice of a coordinate system. \textit{The blow-up of these geometric quantities suggests that extreme black holes are dynamically unstable.}
 
If, on the other hand, we consider the wave $T^{m}\psi_{l}$ then  one needs to differentiate at least  $l+2+m$ times in the transversal direction to obtain a quantity which blows up. See Theorems \ref{theo8} and \ref{theo9} of Section \ref{sec:TheMainTheorems}. Therefore, the $T$ derivatives\footnote{It is also shown that $T\psi$ decays faster than $\psi$.}  counteract the action of the derivatives transversal to $\mathcal{H}^{+}$. 
 
 We conclude this paper by deriving similar decay and blow-up results for the higher order non-degenerate energy. In particular, we show that 
 although (an appropriate modification of) the redshift current can be used as a multiplier for all angular frequencies, the redshift vector field $N$ can only be used as a commutator for $\psi$ supported on the frequencies $l\geq 1$ and, more generally, \textit{one can commute with the redshift vector field at most $l$ times for   $\psi$ supported on the angular frequency $l$}.  See Section \ref{sec:HigherOrderEstimates}.

\subsection{Open Problems}
\label{sec:FutureWork}

An important problem is that of understanding the solutions of the wave equation on the extreme Kerr spacetime. This spacetime is not spherically symmetric and there is no globally causal Killing field in the domain of outer communications (in particular, $T$ becomes spacelike close to the event horizon). Recent results \cite{megalaa} overcome these difficulties for the whole subextreme range of Kerr. The extreme case remains open.

Another related problem is that of the wave equation coupled with the Einstein-Maxwell equations. Then decay for the scalar field was proven in the deep work of Dafermos and Rodnianski  \cite{price}. Again these results hold for non-extreme black holes. For the  extreme case even boundedness of the scalar field for this system remains open.

\section{The Main Theorems}
\label{sec:TheMainTheorems}

We consider the Cauchy problem for the wave equation on the domain of outer communications of extreme Reissner-Nordstr\"{o}m spacetimes (including $\hh$)  with initial data 
\begin{equation}
\left.\psi\right|_{\Sigma_{0}}=\psi_{0}\in H^{k}_{\operatorname{loc}}\left(\Sigma_{0}\right), \left.n_{\Sigma_{0}}\psi\right|_{\Sigma_{0}}=\psi_{1}\in H^{k-1}_{\operatorname{loc}}\left(\Sigma_{0}\right),
\label{cd}
\end{equation}
where the hypersurface $\Sigma_{0}$  crosses $\hh$ and terminates either at spacelike infinity $i^{0}$ or at null infinity $\mathcal{I}^{+}$ and $n_{\Sigma_{0}}$ denotes the future unit normal of $\Sigma_{0}$. We assume that  $k\geq 2$ and that 
\begin{equation}
\lim_{x\rightarrow i^{0}}r\psi^{2}(x)=0.
\label{condition}
\end{equation}
 For simplicity, from now on, when we say ``for all solutions $\psi$ of the wave equation" we will assume that $\psi$ satisfies the above conditions. Note that for obtaining sharp decay results we will have to consider even higher regularity for $\psi$.
 
\subsection{Notation}
\label{sec:Notation}

For the definition of the relevant notions and notation used throughout the paper we refer  to \cite{aretakis1}. For the convenience of the reader, we briefly recall the notation (and conventions) necessary for understanding the statement of the main theorems. Let $\psi_{l}$ denote  the projection of $\psi$ on the eigenspace $E^{l}$ (with corresponding eigenvalue $-l(l+1),l\in\mathbb{N}$) of the spherical Laplacian  $\lapp$. We will say that  $\psi$ is supported on the angular frequencies  $l\geq L$ if $\psi_{i}=0,i=0,...,L-1$ initially (and thus everywhere). Similarly, we will also say that $\psi$ is supported on the angular frequency  $l=L$ if $\psi\in E^{L}$. 

 Let  $N$ be a $\varphi_{\tau}^{T}-$invariant timelike vector field which coincides with $T$ away from $\hh$ (as  defined in Section 10 of \cite{aretakis1}).
The coordinate vector field $\partial_{r}$ corresponds to the system $(v,r)$ and is transversal to $\hh$. Let $T$ denote the globally causal and Killing vector field $\partial_{v}$. Let $\varphi_{\tau}^{T}$ denote the flow of $T$. We define the foliation  $\Sigma_{\tau}=\varphi_{\tau}^{T}(\Sigma_{0})$ and the region $\mathcal{R}(0,\tau)=\cup_{0\leq \tilde{\tau}\leq \tau}\Sigma_{\tilde{\tau}}$. 

Note that the energy currents $J_{\mu}^{V}[\psi],K^{V}[\psi]$ associated to the vector field $V$ are defined in Section 5 of \cite{aretakis1}. For reference, we mention that close to  $\mathcal{H}^{+}$ we have 
\begin{equation*}
J_{\mu}^{T}[\psi]n^{\mu}_{\Sigma_{\tau}}\sim \, (T\psi)^{2}+D(\partial_{r}\psi)^{2}+\left|\nabb\psi\right|^{2},
\end{equation*}
which degenerates on $\mathcal{H}^{+}$ (since $D=\left(1-\frac{M}{r}\right)^{2}$), whereas
\begin{equation*}
J_{\mu}^{N}[\psi]n^{\mu}_{\Sigma_{\tau}}\sim \, (T\psi)^{2}+(\partial_{r}\psi)^{2}+\left|\nabb\psi\right|^{2},
\end{equation*}
which does not degenerate on $\mathcal{H}^{+}$. 

For obtaining energy decay we shall make use of the $\tilde{\Sigma}_{\tau}$ foliation defined as follows: 
We fix $R_{0}>2M$ and consider the hypersurface $\tilde{\Sigma}_{0}$ which is spacelike for $M\leq r\leq R_{0}$ and crosses $\mathcal{H}^{+}$ and for $r\geq R_{0}$ is given by $u=u(p_{0})$, where the coordinate $u$ corresponds to the null system $(u,v)$ with respect to which the metric is $g=-Ddudv+r^{2}g_{\scriptsize\mathbb{S}^{2}}$, and the point  $p_{0}\in\tilde{\Sigma}_{0}$ is such that $r(p_{0})=R_{0}$.
 \begin{figure}[H]
	\centering
		\includegraphics[scale=0.1]{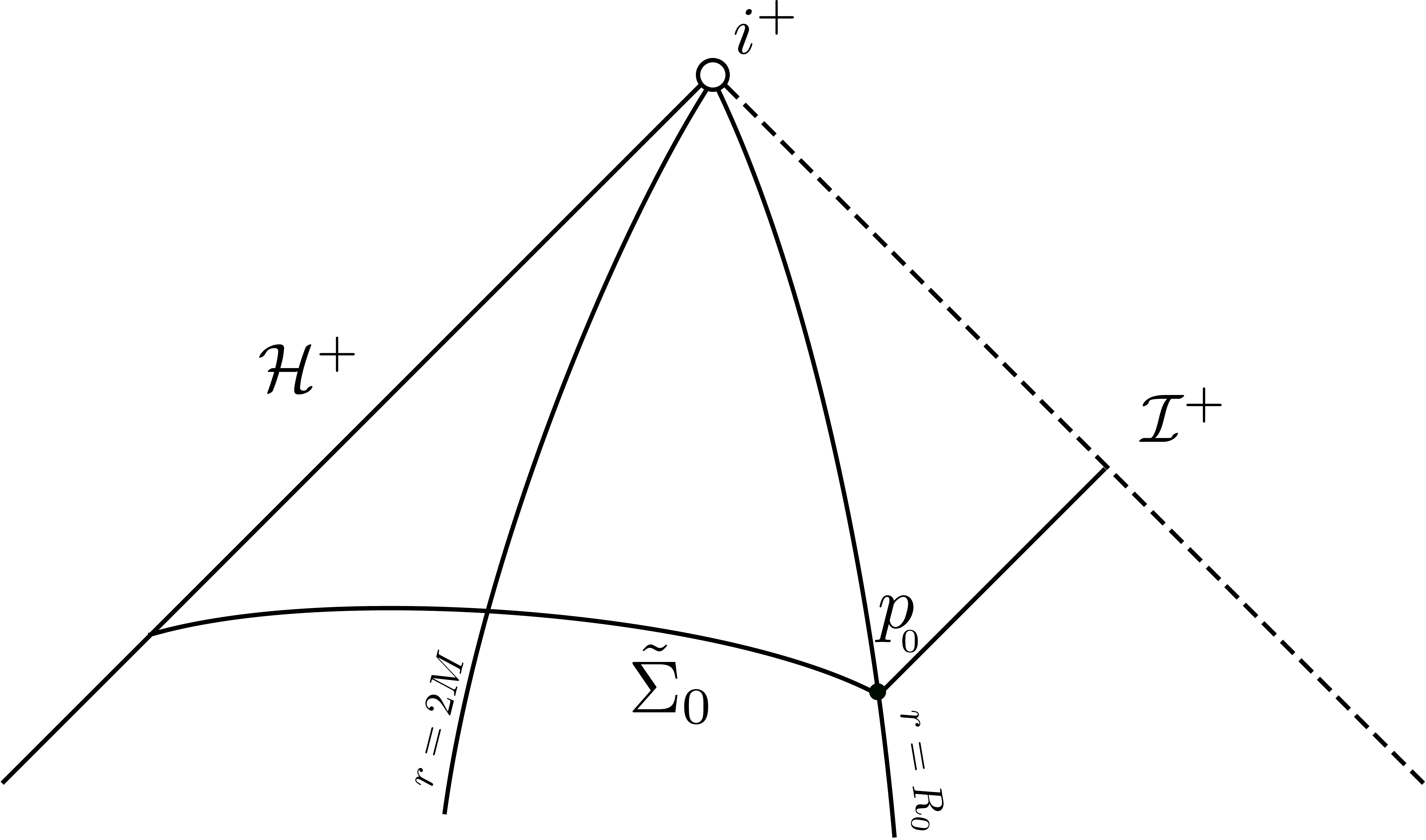}
	\label{fig:pic2ern0}
\end{figure}
We now define  $\tilde{\Sigma}_{\tau}=\varphi_{\tau}^{T}(\tilde{\Sigma}_{0})$. Then for $r$ sufficiently large we have
\begin{equation*}
J_{\mu}^{T}[\psi]n^{\mu}_{\tilde{\Sigma}_{\tau}}\sim (\partial_{v}\psi)^{2}+\left|\nabb\psi\right|^{2},
\end{equation*}
where $\partial_{v}$ corresponds here to the null coordinate system $(u,v)$. 

\subsection{Summary of Results of Part I}
\label{sec:SummaryOfResultsOfPartI}

It would be helpful to summarize several of the results of \cite{aretakis1} at this point. We have\\
 
(1)  Uniform boundedness of non-degenerate energy:
	\begin{equation*}
	\displaystyle\int_{\Sigma_{\tau}}{J_{\mu}^{N}[\psi]n^{\mu}_{\Sigma_{\tau}}}\leq C\displaystyle\int_{\Sigma_{0}}{J_{\mu}^{N}[\psi]n^{\mu}_{\Sigma_{0}}}.
\end{equation*}\\

(2)  Local integrated decay of energy:
	\begin{equation*}
	\begin{split}
\displaystyle\int_{\mathcal{R}(0,\tau)}\!\!\!{\left(\frac{(r-2M)^{2}\cdot \sqrt{D}}{r^{3+\delta}}\left((T\psi)^{2}+\left|\nabb\psi\right|^{2}+D^{2}(\partial_{r}\psi)^{2}+\frac{1}{r^{2}}\psi^{2}\right)\right)}\leq C_{\delta}\displaystyle\int_{\Sigma_{0}}{J_{\mu}^{T}[\psi]n^{\mu}_{\Sigma_{0}}}.
	\end{split}
	\end{equation*}

	Note that the above estimate degenerates on the photon sphere (where $r=2M$) and the event horizon. In order to remove the degeneracy on the photon sphere we need to commute with the vector field $T$. This is related to the so-called trapping effect present on the photon sphere. Note that the same phenomenon takes place on the `photon spheres' of a wide class of black hole spacetimes. Furthermore, as is shown in \cite{aretakis1}, the degeneracy of the above estimate  on $\mathcal{H}^{+}$ may only be removed after commuting with the (non-Killing) vector field $\partial_{r}$. This revealed that degenerate horizons exhibit phenomena characteristic of trapping. This will be of fundamental importance for the analysis of the present paper. \\
	
	(3) Sharp Second Order $L^{2}$ Estimates: There exists $r_{0}$ such that $M<r_{0}<2M$ and if $\mathcal{A}=\mathcal{R}(0,\tau)\cap\left\{M\leq r\leq r_{0}\right\}$ then  for all $\psi$   supported on the angular frequencies $l\geq 1$, the following holds:

	\begin{equation*}
\begin{split}
&\int_{\Sigma_{\tau}\cap\mathcal{A}}{\left(T\partial_{r}\psi\right)^{2}+\left(\partial_{r}\partial_{r}\psi\right)^{2}+\left|\nabb\partial_{r}\psi\right|^{2}}+\int_{\mathcal{H}^{+}}{\left(T\partial_{r}\psi\right)^{2}+\chi_{1}\left|\nabb\partial_{r}\psi\right|^{2}}\\&+\int_{\mathcal{A}}{\left(T\partial_{r}\psi\right)^{2}+\sqrt{D}\left(\partial_{r}\partial_{r}\psi\right)^{2}+\left|\nabb\partial_{r}\psi\right|^{2}}\\\leq &\, C\int_{\Sigma_{0}}{J_{\mu}^{N}[\psi]n^{\mu}_{\Sigma_{0}}}+C\int_{\Sigma_{0}}{J_{\mu}^{N}[T\psi]n^{\mu}_{\Sigma_{0}}}+C\int_{\Sigma_{0}\cap\mathcal{A}}{J_{\mu}^{N}[\partial_{r}\psi]n^{\mu}_{\Sigma_{0}}},
\end{split}
\end{equation*}
where $\chi_{1}=0$ if $\psi$ is supported on $l=1$ and $\chi_{1}=1$ if $\psi$ is supported on $l\geq 2$. \\
	
(4)  Non-decay (for generic initial data) of the higher order quantity
	\begin{equation*}
	\psi^{2}+(\partial_{r}\psi)^{2}
	\end{equation*}
along $\mathcal{H}^{+}$. \\

\subsection{The Statements of the Main Theorems}
\label{sec:TheMainTheorems0}

The main results of the present paper are:

\begin{mytheo}(\textbf{Conservation Laws along $\mathcal{H}^{+}$})
For all $l\in\mathbb{N}$ there exist constants $\b_{i},i=0,1,...,l$, which depend on $M$ and $l$ such that for all solutions $\psi$ which are supported on the (fixed) angular frequency $l$ the quantity
\begin{equation*}
H_{l}[\psi]=\partial_{r}^{l+1}\psi+\sum_{i=0}^{l}{\b_{i}\partial_{r}^{i}\psi}
\end{equation*}
is conserved along the null geodesics of $\mathcal{H}^{+}$.
\label{t3}
\end{mytheo}

\begin{mytheo}(\textbf{Higher Order $L^{2}$ Estimates: Trapping on $\mathcal{H}^{+}$})
\begin{enumerate}
	\item 

There exists $r_{0}$ such that $M<r_{0}<2M$ and a  constant $C>0$ which depends on $M$, $l$ and $\Sigma_{0}$ such that if $\mathcal{A}=\left\{M\leq r\leq r_{0}\right\}\cap\mathcal{R}(0,\tau)$ and $k\leq l$ then   for all solutions $\psi$ of the wave equation which are supported on frequencies greater or equal to $l$, the following holds
\begin{equation*}
\begin{split}
&\displaystyle\int_{\Sigma_{\tau}\cap\mathcal{A}}{\left(T\partial_{r}^{k}\psi\right)^{2}+\left(\partial_{r}^{k+1}\psi\right)^{2}+\left|\nabb\partial_{r}^{k}\psi\right|^{2}}+\displaystyle\int_{\mathcal{H}^{+}}{\left(T\partial_{r}^{k}\psi\right)^{2}+\chi_{\left\{k=l\right\}}\left|\nabb\partial_{r}^{k}\psi\right|^{2}}\\
+&\displaystyle\int_{\mathcal{A}}{\left(T\partial_{r}^{k}\psi\right)^{2}+\left(1-\frac{M}{r}\right)\left(\partial_{r}^{k+1}\psi\right)^{2}+\left|\nabb\partial_{r}^{k}\psi\right|^{2}}\\
\leq  &C\sum_{i=0}^{k}\displaystyle\int_{\Sigma_{0}}{J_{\mu}^{N}\left[T^{i}\psi\right]n^{\mu}_{\Sigma_{0}}}+C\sum_{i=1}^{k}\int_{\Sigma_{0}\cap\mathcal{A}}{J_{\mu}^{N}\left[\partial^{i}_{r}\psi\right]n^{\mu}_{\Sigma_{0}}},
\end{split}
\end{equation*}
where $\chi_{\left\{k=l\right\}}=0$ if $k=l$ and    $\chi_{\left\{k=l\right\}}=1$ otherwise. 
\item If $\psi$ is replaced with $T^{m}\psi,m\geq 1,$ then  similar $L^{2}$ estimates hold for all $k\leq l+m$. Also in this case,  for all $k\leq l+m$ we do not need the factor $\chi_{\left\{k=l\right\}}$.

\end{enumerate}
\label{theorem3}
\end{mytheo}

\begin{remark} In view of the $\left(1-\frac{M}{r}\right)$ factor, the spacetime term of the above estimate degenerates on $\mathcal{H}^{+}$. To remove this degeneracy one needs to lose even more derivatives by commuting with the vector field $\partial_{r}$ one more time. This reflects the higher order trapping effect of $\mathcal{H}^{+}$. The commutation with $T$ is related to the trapping on the photon sphere.
\end{remark}

\begin{mytheo}(\textbf{Energy Decay})
Consider the foliation $\tilde{\Sigma}_{\tau}$ as defined above. Let
\begin{equation*}
\begin{split}
I^{T}_{\tilde{\Sigma}_{\tau}}[\psi]=&\int_{\tilde{\Sigma}_{\tau}}{J^{N}_{\mu}[\psi]n^{\mu}_{\tilde{\Sigma}_{\tau}}}+
\int_{\tilde{\Sigma}_{\tau}}{J^{T}_{\mu}[T\psi]n^{\mu}_{\tilde{\Sigma}_{\tau}}}+\int_{\tilde{\Sigma}_{\tau}}{r^{-1}\left(\partial_{v}(r\psi)\right)^{2}}
\end{split}
\end{equation*}
and
\begin{equation*}
\begin{split}
I^{N}_{\tilde{\Sigma}_{\tau}}[\psi]=&\int_{\tilde{\Sigma}_{\tau}}{J^{N}_{\mu}[\psi]n^{\mu}_{\tilde{\Sigma}_{\tau}}}+
\int_{\tilde{\Sigma}_{\tau}}{J^{N}_{\mu}[T\psi]n^{\mu}_{\tilde{\Sigma}_{\tau}}}+\int_{\mathcal{A}\cap\tilde{\Sigma}_{\tau}}{J^{N}_{\mu}[\partial_{r}\psi]n^{\mu}_{\tilde{\Sigma}_{\tau}}}+\int_{\tilde{\Sigma}_{\tau}}{r^{-1}\left(\partial_{v}(r\psi)\right)^{2}},
\end{split}
\end{equation*}
where $\mathcal{A}$ is as defined in Theorem \ref{theorem3}.  Here $\partial_{v}$ corresponds to the null system $(u,v)$ (whereas $\partial_{r}$ still corresponds to the Eddington-Finkelstein coordinate system $(v,r)$).
There exists a constant $C$ that depends  on the mass $M$  and $\tilde{\Sigma}_{0}$  such that:
\begin{itemize}
	\item For all solutions $\psi$ of the wave equation we have 
\begin{equation*}
\displaystyle\int_{\tilde{\Sigma}_{\tau}}J^{T}_{\mu}[\psi]n_{\tilde{\Sigma}_{\tau}}^{\mu}\leq CE_{1}[\psi]\frac{1}{\tau^{2}},
\end{equation*}
where
\begin{equation*}
E_{1}[\psi]=I^{T}_{\tilde{\Sigma}_{0}}[T\psi]+\displaystyle\int_{\tilde{\Sigma}_{0}}{J^{N}_{\mu}[\psi]n^{\mu}_{\tilde{\Sigma}_{0}}}+\displaystyle\int_{\tilde{\Sigma}_{0}}{\left(\partial_{v}(r\psi)\right)^{2}}.
\end{equation*}
\end{itemize}
\begin{itemize}
	\item For all solutions $\psi$ to the wave equation which are supported on the frequencies $l\geq 1$ we have
\begin{equation*}
\begin{split}
\int_{\tilde{\Sigma}_{\tau}}{J^{N}_{\mu}[\psi]n^{\mu}_{\tilde{\Sigma}_{\tau}}}\,\leq\, CE_{2}[\psi]\frac{1}{\tau},
\end{split}
\end{equation*}
where
\begin{equation*}
\begin{split}
E_{2}[\psi]=I^{N}_{\tilde{\Sigma}_{0}}[\psi].
\end{split}
\end{equation*}
\end{itemize}
\begin{itemize}
	\item  For all solutions $\psi$ to the wave equation which are supported on the frequencies $l\geq 2$ we have
\begin{equation*}
\begin{split}
\displaystyle\int_{\tilde{\Sigma}_{\tau}}{J^{N}_{\mu}[\psi]n^{\mu}_{\tilde{\Sigma}_{\tau}}}\,\leq\, CE_{3}[\psi]\frac{1}{\tau^{2}},
\end{split}
\end{equation*}
where
\begin{equation*}
\begin{split}
E_{3}[\psi]=I^{N}_{\tilde{\Sigma}_{0}}[\psi]+I^{N}_{\tilde{\Sigma}_{0}}[T\psi]+\int_{\mathcal{A}\cap\tilde{\Sigma}_{0}}{J^{N}_{\mu}[\partial_{r}\partial_{r}\psi]n^{\mu}_{\tilde{\Sigma}_{0}}}+\int_{\tilde{\Sigma}_{0}}{(\partial_{v}(r\psi))^{2}}.
\end{split}
\end{equation*}
\end{itemize}
\label{t4}
\end{mytheo}

\begin{remark}
In view of the trapping effect on $\mathcal{H}^{+}$, to obtain decay of the non-degenerate energy we need to `lose' the $\partial_{r}$ derivative (which appears in $I^{N}_{\tilde{\Sigma}_{0}}[\psi]$). Note that  the full decay requires `losing' the higher order derivative $\partial_{r}\partial_{r}$ (which appears in $E_{3}[\psi]$).

\end{remark}

\begin{mytheo}(\textbf{Pointwise Decay})
Fix $R_{1}$ such that $M<R_{1}$ and let $\tau\geq 1$. Let $E_{1}, E_{2}, E_{3}$ be the quantities as defined in Theorem \ref{t4}. Then, there exists a constant $C$ that depends  on the mass $M$, $R_{1}$  and $\tilde{\Sigma}_{0}$  such that:
\begin{itemize}
	\item  For all solutions $\psi$ to the wave equation  we have
	\begin{equation*}
\left|\psi\right|\leq C \sqrt{E_{5}}\frac{1}{\sqrt{r}\cdot \tau},\ \  \left|\psi\right|\leq C\sqrt{E_{5}}\frac{1}{r\cdot\sqrt{\tau}}
\end{equation*}
in $\left\{R_{1}\leq r\right\}$, where
\begin{equation*}
E_{5}=\sum_{\left|k\right|\leq 2}{E_{1}\left[\Omega^{k}\psi\right]},
\end{equation*}
with $\Omega\in\left\{\Omega_{1},\Omega_{2},\Omega_{3}\right\}$ and $\Omega_{i},i=1,2,3$ are the angular momentum operators.
\end{itemize}
\begin{itemize}
	\item For all  solutions $\psi$ of the wave equation we have 
\begin{equation*}
\begin{split}
\left|\psi\right|\leq C\sqrt{E_{6}}\frac{1}{\tau^{\frac{3}{5}}}
\end{split}
\end{equation*}
in $\left\{M\leq r\leq R_{1}\right\}$, where 
\begin{equation*}
E_{6}=E_{1}+E_{4}[\psi]+E_{4}[T\psi]+E_{5}+\!\sum_{\left|k\right|\leq 2}\!\!\left(E_{2}{\left[\Omega^{k}\psi\right]}\!+\!E_{3}{\left[\Omega^{k}\psi\right]}\right)+\left\|\partial_{r}\psi\right\|_{L^{\infty}\left(\tilde{\Sigma}_{0}\right)}^{2}\!. 
\end{equation*}
\end{itemize}
\begin{itemize}
	\item For all solutions $\psi$ to the wave equation which are supported on the frequencies $l\geq 1$ we have
\begin{equation*}
\begin{split}
\left|\psi\right|\leq C\sqrt{E_{7}}\frac{1}{\tau^{\frac{3}{4}}}
\end{split}
\end{equation*}
in $\left\{M\leq r\leq R_{1}\right\}$, where 
\begin{equation*}
E_{7}=E_{5}+\sum_{\left|k\right|\leq 2}E_{2}{\left[\Omega^{k}\psi\right]}+\sum_{\left|k\right|\leq 2}E_{3}{\left[\Omega^{k}\psi\right]}.
\end{equation*}
\end{itemize}
\begin{itemize}
\item For all solutions $\psi$ to the wave equation which are supported on the frequencies $l\geq 2$ we have
\begin{equation*}
\begin{split}
\left|\psi\right|\leq C\sqrt{E_{8}}\frac{1}{\tau},
\end{split}
\end{equation*}
in $\left\{M\leq r\leq R_{1}\right\}$, where 
\begin{equation*}
\begin{split}
E_{8}=\sum_{\left|k\right|\leq 2}E_{3}{\left[\Omega^{k}\psi\right]}.
\end{split}
\end{equation*}
\end{itemize}
\label{t6}
\end{mytheo}

\begin{mytheo}(\textbf{Higher Order Energy and Pointwise Estimates I: Decay Results})
Fix $R_{1}$ such that $R_{1}>M$ and let $\tau\geq 1$. Let also $k,l,m\in\mathbb{N}$. Then, there exist  constants $C$ which depend on $M,l,k, R_{1}$ and $\tilde{\Sigma}_{0}$ such that the following holds: For all solutions $\psi$ of the wave equation which are supported on the (fixed) angular frequency $l$, there exist norms $\tilde{E}_{k,l,m}$, $E_{k,l,m}$ of the initial data of $\psi$ such that
\begin{enumerate}
	\item $\displaystyle\int_{\tilde{\Sigma}_{\tau}\cap\left\{M\leq r\leq R_{1}\right\}}{J_{\mu}^{N}[\partial_{r}^{k}T^{m}\psi]n_{\tilde{\Sigma}_{\tau}}^{\mu}}\leq C\tilde{E}_{k,l,m}^{2}\frac{1}{\tau^{2}}$ for all $k\leq l+m-2$,
	\item $\displaystyle\int_{\tilde{\Sigma}_{\tau}\cap\left\{M\leq r\leq R_{1}\right\}}{J_{\mu}^{N}[\partial_{r}^{l+m-1}T^{m}\psi]n_{\tilde{\Sigma}_{\tau}}^{\mu}}\leq C\tilde{E}_{l+m-1,l,m}^{2}\frac{1}{\tau}$.
	\item $\displaystyle\left|\partial_{r}^{k}T^{m}\psi\right|\leq CE_{k,l,m}\displaystyle\frac{1}{\tau}$ in $\left\{M\leq r\leq R_{1}\right\}$ for all $k\leq l-2+m$,
	\item $\displaystyle\left|\partial_{r}^{l+m-1}T^{m}\psi\right|\leq CE_{l+m-1,l,m}\displaystyle\frac{1}{\tau^{\frac{3}{4}}}$ in $\left\{M\leq r\leq R_{1}\right\}$,
		\item $\displaystyle\left|\partial_{r}^{l+m}T^{m}\psi\right|\leq CE_{l+m,l,m}\displaystyle\frac{1}{\tau^{\frac{1}{4}}}$ in $\left\{M\leq r\leq R_{1}\right\}$,
\end{enumerate} 
\label{theo8}
\end{mytheo}

\begin{mytheo}(\textbf{Higher Order Energy and Pointwise Estimates II: Non-Decay and Blow-up Results})
Fix $R_{1}$ such that $R_{1}>M$. Let $k,l,m\in\mathbb{N}$ and $H_{l}[\psi]$ be the functions as defined in Theorem \ref{t3}.  Then there exist non zero constants $C,c$ (in fact $c>0$) which depend on $M,l,k, R_{1}$ such that for generic  solutions $\psi$ to the wave equation which are supported on the (fixed) angular frequency $l$ we have
\begin{enumerate}
\item
\begin{equation*}
\partial_{r}^{l+m+1}T^{m}\psi(\tau,\theta,\phi)\rightarrow CH_{l}[\psi](\theta,\phi)
\end{equation*}
as $\tau\rightarrow +\infty$ along $\hh$ and generically $H_{l}[\psi]\neq 0$ almost everywhere on $\mathbb{S}^{2}_{0}=\tilde{\Sigma}_{0}\cap\hh$ (and $C=1$ for $m=0$).
\item

\begin{equation*}
\displaystyle\left|\partial_{r}^{l+m+k}T^{m}\psi\right|(\tau,\theta,\phi)\geq c \left|H_{l}[\psi](\theta,\phi)\right|\tau^{k-1}
\end{equation*}
asymptotically on $\mathcal{H}^{+}$ for all $k\geq 2$.
\end{enumerate}
Finally, for generic solutions $\psi$ to the wave equation we have 
\begin{equation*}
\displaystyle\int_{\tilde{\Sigma}_{\tau}\cap\left\{M\leq r\leq R_{1}\right\}}{J_{\mu}^{N}[\partial_{r}^{k}T^{m}\psi]n_{\tilde{\Sigma}_{\tau}}^{\mu}}\longrightarrow +\infty 
\end{equation*}
as $\tau\rightarrow +\infty$ for all $k\geq m+1$. 
\label{theo9}
\end{mytheo}

\section{Conservation Laws on Degenerate Event Horizons}
\label{sec:ConservationLawsOnDegenerateEventHorizons}

We will prove that the lack of redshift gives rise to conservation laws along $\mathcal{H}^{+}$ for translation invariant derivatives. These laws govern the evolution of the low angular frequencies and play a fundamental role in understanding the evolution of generic solutions to the wave equation. 

We  use the regular coordinate system $(v,r)$. Let $T=\partial_{v}$, where $\partial_{v}$ denotes the coordinate vector field corresponding to the system $(v,r)$. As we shall see our results can be applied to a  general class of (spherically symmetric) degenerate black hole spacetimes. We conclude this section by showing that the argument of Kay and Wald (see \cite{wa1}) could not have been applied in our case even for obtaining uniform boundedness of the solutions to the wave equation.

Let us first consider spherically symmetric solutions.
\begin{proposition}
For all  spherically symmetric solutions $\psi$ to the wave equation the quantity 
\begin{equation}
H_{0}[\psi]=\partial_{r}\psi +\frac{1}{M}\psi
\label{0l}
\end{equation}
is conserved along $\mathcal{H}^{+}$. 
\label{ndl0}
\end{proposition}
\begin{proof}
Since $\psi$ solves $\Box_{g}\psi =0$ and since $\lapp\psi =0$ we have $T\partial_{r}\psi+\frac{1}{M}T\psi=0$ and, since $T$ is tangential to $\mathcal{H}^{+}$, this implies that $\partial_{r}\psi +\frac{1}{M}\psi$ remains constant along $\mathcal{H}^{+}$.
\end{proof}

\begin{proposition}
For all  solutions $\psi$ to the wave equation  that are supported on the angular frequency $l=1$ the quantity
\begin{equation}
H_{1}[\psi]=\partial_{r}\partial_{r}\psi+\frac{3}{M}\partial_{r}\psi+\frac{1}{M^{2}}\psi
\label{02}
\end{equation} 
is conserved along the null geodesics of $\mathcal{H}^{+}$.
\label{ndl1}
\end{proposition}
\begin{proof}
Since $\lapp\psi=-\frac{2}{r^{2}}\psi$, the wave equation on $\mathcal{H}^{+}$ gives us
\begin{equation}
2T\partial_{r}\psi+\frac{2}{M}T\psi=\frac{2}{M^{2}}\psi.
\label{psi1l}
\end{equation}
Moreover if $R=D'+\frac{2D}{r}$ then
\begin{equation*}
\begin{split}
\partial_{r}\left(\Box_{g}\psi\right)=&D\partial_{r}\partial_{r}\partial_{r}\psi+2T\partial_{r}\partial_{r}\psi+\frac{2}{r}\partial_{r}T\psi+R\partial_{r}\partial_{r}\psi+\partial_{r}\lapp\psi+D'\partial_{r}\partial_{r}\psi-\frac{2}{r^{2}}T\psi+R'\partial_{r}\psi\\
\end{split}
\end{equation*}
and thus by restricting this identity on $\mathcal{H}^{+}$ we take
\begin{equation}
\begin{split}
2\partial_{r}^{2}T\psi+\frac{2}{M}\partial_{r}T\psi-\frac{2}{M^{2}}T\psi+\frac{4}{M^{3}}\psi+\left(-\frac{2}{M^{2}}+R'(M)\right)\partial_{r}\psi=0.
\end{split}
\label{rpsi1}
\end{equation}
However, $R'(M)=\frac{2}{M^{2}}$ and in view of \eqref{psi1l} we have  
\begin{equation*}
\begin{split}
2\partial_{r}^{2}T\psi+\frac{2}{M}\partial_{r}T\psi-\frac{2}{M^{2}}T\psi+\frac{2}{M}\left(2\partial_{r}T\psi+\frac{2}{M}T\psi\right)=0
\end{split}
\end{equation*}
which means that   \eqref{02} is constant along the  integral curves of $T$ on $\mathcal{H}^{+}$.
\end{proof}

\begin{proposition}
There exist constants $\a_{i}^{j},j=0,1,...,l-1,\, i=0,1,...,j+1$, which depend on $M$ and $l$ such that for all  solutions $\psi$ of the wave equation which are supported on the  (fixed) frequency $l$ we have
\begin{equation*}
\partial_{r}^{j}\psi=\sum_{i=0}^{j+1}{\a_{i}^{j}T\partial_{r}^{i}\psi},
\end{equation*}
on $\mathcal{H}^{+}$. 
\label{corl2}
\end{proposition}
\begin{proof}
For $j=0,1$ we just have to revisit the proof  of Proposition \ref{ndl1} and use the fact that for all $l\geq 2$ we have $-\frac{l(l+1)}{M^{2}}+R'(M)\neq 0$. We next proceed by induction on $j$ for fixed $l$.  We suppose that the result holds for $j=0,1,...,k-1$ and we will prove that it holds for $j=k$ provided $k\leq l-1$. Clearly, 
\begin{equation}
\begin{split}
\partial_{r}^{k}\left(\Box_{g}\psi\right)=&D\left(\partial_{r}^{k+2}\psi\right)+2\partial_{r}^{k+1}T\psi+\frac{2}{r}\partial_{r}^{k}T\psi+R\partial_{r}^{k+1}\psi+\partial_{r}^{k}\lapp\psi+\\
&+\sum_{i=1}^{k}{
\binom{k}{i}\partial_{r}^{i}D\cdot \partial_{r}^{k-i+2}\psi}+\sum_{i=1}^{k}{\binom{k}{i}\partial_{r}^{i}\frac{2}{r}\cdot\partial_{r}^{k-i}T\psi}+\sum_{i=1}^{k}{\binom{k}{i}\partial_{r}^{i}R\cdot\partial_{r}^{k-i+1}\psi}.
\end{split}
\label{kcom}
\end{equation}
We observe that the coef{}ficients of $\partial_{r}^{k+2}\psi$ and $\partial_{r}^{k+1}\psi$ vanish on $\mathcal{H}^{+}$. Since $\lapp\psi=-\frac{l(l+1)}{r^{2}}\psi$, the coef{}ficient of $\partial_{r}^{k}\psi$ on $\mathcal{H}^{+}$ is equal to 
\begin{equation}
\begin{split}
\binom{k}{2}D''+\binom{k}{1}R'-\frac{l(l+1)}{M^{2}}=\frac{k(k+1)}{2}\frac{2}{M^{2}}-\frac{l(l+1)}{M^{2}},
\end{split}
\label{idconservation}
\end{equation}
which is non-zero if and only if $l\neq k$. Therefore, for all $k\leq l-1$ we can solve with respect to $\partial_{r}^{k}\psi$ and use the inductive hypothesis completing thus the proof of the proposition.
\end{proof}

\begin{proof}[Proof of Theorem \ref{t3} of Section \ref{sec:TheMainTheorems}]
We apply \eqref{kcom} for $k=l$. Then, according to our previous calculation, the coef{}ficients of  $\partial_{r}^{l+2}\psi$,  $\partial_{r}^{l+1}\psi$ and $\partial_{r}^{l}\psi$ vanish on $\mathcal{H}^{+}$. Therefore, we end up with the terms $\partial_{r}^{k}T\psi, k=0,1,...,l+1$ and $\partial_{r}^{j}\psi,j=0,1,...,l-1$. Thus, from Proposition \ref{corl2} there exist 
 constants $\b_{i},i=0,1,...,l$ which depend on $M$ and $l$ such that
\begin{equation*}
T\partial_{r}^{l+1}\psi+\sum_{i=0}^{l}{\b_{i}T\partial_{r}^{i}\psi}=0
\end{equation*}
on $\mathcal{H}^{+}$, which implies that the quantity
\begin{equation*}
H_{l}[\psi]=\partial_{r}^{l+1}\psi+\sum_{i=0}^{l}{\b_{i}\partial_{r}^{i}\psi}
\end{equation*}
is conserved along the integral curves of $T$ on $\mathcal{H}^{+}$. 
\end{proof}
Note that the above theorem holds for more general extreme black hole spacetimes. Indeed, let the metric with respect to the coordinate system $(v,r,\theta,\phi)$  take the form
\begin{equation*}
g=-Ddv^{2}+2dvdr+r^{2}g_{\mathbb{S}^{2}},
\end{equation*}
for a general $D=D(r)$. If this spacetime admits a black hole whose event horizon is located at $r=r_{\mathcal{H}^{+}}$ where $D(r_{\mathcal{H}^{+}})=0$, then the above theorem (and proof) still holds if
\begin{equation}
D'(r_{\mathcal{H}^{+}})=0,
\label{ex1}
\end{equation}
\begin{equation}
D''(r_{\mathcal{H}^{+}})=\frac{2}{r_{\mathcal{H}^{+}}^{2}}.
\label{ex2}
\end{equation}
The equation \eqref{ex1} expresses the extremality of the black hole whereas the additional \eqref{ex2} is necessary so \eqref{idconservation} holds. Note here that \eqref{idconservation} trivially holds for the frequency $l=0$ and thus \eqref{ex2} is not required for spherically symmetric solutions of the wave equation. In \cite{aretakis3}, we provide even more general assumptions under which we have conservation laws for spherically symmetric self-gravitating scalar fields on extreme black holes. 

\subsection{Applications}
\label{sec:Applications}
Note that although we show in Section \ref{sec:HigherOrderEstimates} that $H_{l}[\psi]$ is non-zero almost everywhere on $\mathcal{H}^{+}\cap\Sigma_{0}$ for generic initial data, we have $H_{l}[T^{m}\psi]=0$ for all $\psi$ and $m\geq 1$. For the waves of the form $T^{m}\psi$ we have the following
\begin{proposition}
There exist  constants $C\neq 0$ and $\lambda_{ij}$ which depend on $M,l,m$ and such that for all solutions of the wave equation which are supported on the  frequency $l$ we have
\begin{equation*}
\begin{split}
\partial_{r}^{l+m+1}T^{m}\psi+\sum_{j=0}^{m}\sum_{i=0}^{l}\lambda_{ij}\partial_{r}^{i}T^{j}\psi=C\cdot H_{l}[\psi]
\end{split}
\end{equation*}
on $\hh$.
\label{tmpsi}
\end{proposition}
\begin{proof}
Consider \eqref{kcom} for $k=l+1$. Then
\begin{equation*}
\begin{split}
0=&\ 2\partial_{r}^{l+2}T\psi+\frac{2}{M}\partial_{r}^{l+1}T\psi+\partial_{r}^{l+1}\left(-\frac{l(l+1)}{r^{2}}\psi\right)+\\
&+\sum_{i=2}^{l+1}{
\binom{l+1}{i}\partial_{r}^{i}D\cdot \partial_{r}^{l-i+3}\psi}+\sum_{i=1}^{l+1}{\binom{l+1}{i}\partial_{r}^{i}\frac{2}{r}\cdot\partial_{r}^{l+1-i}T\psi}+\sum_{i=1}^{l+1}{\binom{l+1}{i}\partial_{r}^{i}R\cdot\partial_{r}^{l-i+2}\psi}.
\end{split}
\end{equation*}
Since $H_{l}[T\psi]=0$, the term $\partial_{r}^{l+1}T\psi$ can be expressed in terms of $T\psi, \partial_{r}T\psi,..., \partial_{r}^{l}T\psi$. Note also  $\partial_{r}^{l+1}\psi$, whose coefficient on the right hand side is non-zero, can be replaced by a linear expression of $H_{l}[\psi], \psi, \partial_{r}\psi,...,\partial_{r}^{l}\psi$. This proves the proposition for $m=1$. The general case can be proved inductively by using \eqref{kcom} for $k=l+m$ and $\psi$ replaced with $T^{m-1}\psi$. Indeed, we obtain that $\partial_{r}^{l+m+1}T^{m}\psi$ can be expressed on $\hh$ as a linear combination of the terms $\partial_{r}^{l+m}T^{m}\psi,\partial_{r}^{l+m}T^{m-1}\psi$ and $\partial_{r}^{k}T^{m}\psi$ and $\partial_{r}^{k}T^{m-1}\psi$ for $k\leq l+m-1$. For the terms $\partial_{r}^{l+m}T^{m}\psi,\partial_{r}^{k}T^{m}\psi,\partial_{r}^{k}T^{m-1}\psi$ with $k\leq l+m-1$ we use the inductive hypothesis and that $H_{l}[T^{i}\psi]=0$ for all $i\geq 1$. Note finally that the coefficient of $\partial_{r}^{l+m}T^{m-1}\psi$ is non-zero and, therefore, this term can be replaced by a linear combination of $H_{l}[\psi]$ and $\partial_{r}^{i}T^{j}\psi$ for $i\leq l, j\leq m$.

\end{proof}

We conclude this section with the following important application of  Theorem \ref{t3}.
\begin{proposition}
For generic initial data  there is no Cauchy hypersurface $\Sigma$ crossing $\mathcal{H}^{+}$ and a solution $\tilde{\psi}$ of the wave equation such that 
\begin{equation*}
T\tilde{\psi}=\psi
\end{equation*}
in the future of $\Sigma$.
\label{wald}
\end{proposition}
\begin{proof}
Suppose that there exists a wave $\tilde{\psi}$ such that $T\tilde{\psi}=\psi$. Then we can decompose $\tilde{\psi}=\tilde{\psi}_{0}+\tilde{\psi}_{\geq 1}$
 and take
\begin{equation*}
\begin{split}
T\tilde{\psi}=&T\tilde{\psi}_{0}+T\tilde{\psi}_{\geq 1}=(T\tilde{\psi})_{0}+(T\tilde{\psi})_{\geq 1}
\end{split}
\end{equation*}
since $T$ is an endomorphism of the eigenspaces of $\lapp$. But $T\tilde{\psi}=\psi$ and thus $\psi_{0}=(T\tilde{\psi})_{0}=T\tilde{\psi}_{0}$. Since $\tilde{\psi}_{0}$ is a spherically symmetric wave we have $\partial_{r}T\tilde{\psi}_{0}+\frac{1}{M}T\tilde{\psi}_{0}=0$ which yields $\partial_{r}\psi_{0}+\frac{1}{M}\psi_{0}=0.$
However,the  quantity $\partial_{r}\psi_{0}+\frac{1}{M}\psi_{0}$ is completely determined by the initial data and thus  for generic initial data it is non-zero.
\end{proof}

This shows that we can not adapt the argument of Kay and Wald for the extreme case (see \cite{wald} and \cite{wa1}) even for proving the uniform boundedness of solutions to the wave equation. Indeed, using this argument one could prove that in Schwarzschild that  for any solution of the wave equation $\psi$ there is another solution $\tilde{\psi}$ such that $T\tilde{\psi}=\psi$ in the future of a Cauchy hypersurface $\Sigma$ of the region $J^{+}(\Sigma)\cap \mathcal{D}$, where $\mathcal{D}$ denotes the domain of outer communications of Schwarzschild.

\section{Sharp Higher Order $L^{2}$ Estimates}
\label{sharpl2}

We commute the wave equation with $\partial_{r}^{k}$ where $k\in\mathbb{N}$ and $k\geq 2$ aiming at controlling all higher derivatives of $\psi$ (on the spacelike hypersurfaces and the spacetime region up to and including the horizon $\mathcal{H}^{+}$). In view of Theorem \ref{t3} of Section \ref{sec:TheMainTheorems} the weakest condition on $\psi$ would be such that it is supported on the frequencies $l\geq k$.

\subsection{The Commutator $\left[\Box_{g}, \partial_{r}^{k} \right]$}
\label{sec:TheCommutatorLeftBoxGPsiPartialRKRight}

First note that if $R=D'+\frac{D}{2r}, D'=\frac{dD}{dr}$ then
\begin{equation*}
\begin{split}
\partial_{r}^{k}\left(\Box_{g}\psi\right)=&D\left(\partial_{r}^{k+2}\psi\right)+2T\partial_{r}^{k+1}\psi+\frac{2}{r}T\partial_{r}^{k}\psi+R\partial_{r}^{k+1}\psi+\partial_{r}^{k}\lapp\psi+\\
&+\sum_{i=1}^{k}{
\binom{k}{i}\partial_{r}^{i}D\cdot \partial_{r}^{k-i+2}\psi}+\sum_{i=1}^{k}{\binom{k}{i}\partial_{r}^{i}\frac{2}{r}\cdot T\partial_{r}^{k-i}\psi}+\sum_{i=1}^{k}{\binom{k}{i}\partial_{r}^{i}R\cdot\partial_{r}^{k-i+1}\psi}.
\end{split}
\end{equation*}
Let us compute the commutator $\left[\lapp,\partial_{r}^{k}\right]$. If we denote $\lapp_{1}$ the Laplacian on the unit sphere then
\begin{equation*}
\begin{split}
\partial_{r}^{k}\lapp\psi &=\partial_{r}^{k}\left(\frac{1}{r^{2}}\lapp_{1}\right)=\sum_{i=0}^{k}{\binom{k}{i}\partial_{r}^{i}\frac{1}{r^{2}}\cdot\partial_{r}^{k-i}\lapp_{1}}=\sum_{i=0}^{k}{\binom{k}{i}r^{2}\partial_{r}^{i}\frac{1}{r^{2}}\cdot\lapp\partial_{r}^{k-i}\psi}\\&=\lapp\partial_{r}^{k}\psi+\sum_{i=1}^{k}{\binom{k}{i}r^{2}\partial_{r}^{i}r^{-2}\cdot\lapp\partial_{r}^{k-i}\psi}.\\
\end{split}
\end{equation*}
Therefore,
\begin{equation}
\left[\lapp,\partial_{r}^{k}\right]\psi=-\sum_{i=1}^{k}{\binom{k}{i}r^{2}\partial_{r}^{i}r^{-2}\cdot\lapp\partial_{r}^{k-i}\psi}
\label{sh1com1ttat}
\end{equation}
and so 
\begin{equation}
\begin{split}
\left[\Box_{g},\partial_{r}^{k}\right]\psi=&-\sum_{i=1}^{k}{
\binom{k}{i}\partial_{r}^{i}D\cdot \partial_{r}^{k-i+2}\psi}-\sum_{i=1}^{k}{\binom{k}{i}\partial_{r}^{i}\frac{2}{r}\cdot T\partial_{r}^{k-i}\psi}\\
&-\sum_{i=1}^{k}{\binom{k}{i}\partial_{r}^{i}R\cdot\partial_{r}^{k-i+1}\psi}-\sum_{i=1}^{k}{\binom{k}{i}r^{2}\partial_{r}^{i}r^{-2}\cdot\lapp\partial_{r}^{k-i}\psi}.
\label{comhigherorder}
\end{split}
\end{equation}

\subsection{Induction on $l$}
\label{sec:TheMultiplierLKAndTheEnergyIdentity}

For any solution $\psi$ of the wave equation we  control the  higher order derivatives in the spacetime region away from  $\mathcal{H}^{+}$ and the photon sphere:
\begin{equation}
\left\|\partial_{\a}\psi\right\|^{2}_{L^{2}\left(\mathcal{R}\left(0,\tau\right)\cap\left\{M<r_{0}\leq r\leq r_{1}<2M\right\}\right)} \leq C\int_{\Sigma_{0}}{\left(\sum_{i=0}^{k-1}{J^{T}_{\mu}\left[T^{i}\psi\right]n^{\mu}_{\Sigma_{0}}}\right)},
\label{hospacaaH}
\end{equation}
where $C$ depends on $M$, $r_{0}$, $r_{1}$ and $\Sigma_{0}$ and $\left|a\right|=k$. This can be proved by commuting the wave equation  with $T^{i},i=1,...,k-1$ and using the degenerate $X$ estimate of Theorem 1 of \cite{aretakis1} and local elliptic estimates. We next derive estimates controlling the $k^{\text{th}}$-order derivatives of $\psi$ in a neighbourhood of $\mathcal{H}^{+}$. 
\begin{proof}[Proof of Theorem \ref{theorem3} of Section \ref{sec:TheMainTheorems}]
 For simplicity we write $\mathcal{R}$ instead of $\mathcal{R}(0,\tau)$ and  $\hh$ instead of $\hh\cap\mathcal{R}(0,\tau)$. We follow an inductive process.  We suppose that the theorem holds for all $m$ such that $0\leq m\leq l-1,$ i.e.~ that there exists a neighbourhood $\mathcal{A}$ of $\hh$ and a uniform positive constant $C$ that depends on $M$, $\Sigma_{0}$ and $k$ such that 
\begin{equation}
\begin{split}
&\int_{\Sigma_{\tau}\cap\mathcal{A}}{\left(T\partial_{r}^{k}\psi\right)^{2}+\left(\partial_{r}^{k+1}\psi\right)^{2}+\left|\nabb\partial_{r}^{k}\psi\right|^{2}}+\int_{\mathcal{H}^{+}}{\left(T\partial_{r}^{k}\psi\right)^{2}+\chi_{\left\{k=m\right\}}\left|\nabb\partial_{r}^{k}\psi\right|^{2}}\\
+&\int_{\mathcal{A}}{\left(T\partial_{r}^{k}\psi\right)^{2}+\sqrt{D}\left(\partial_{r}^{k+1}\psi\right)^{2}+\left|\nabb\partial_{r}^{k}\psi\right|^{2}}\\
\leq  &C\sum_{i=0}^{k}\int_{\Sigma_{0}}{J_{\mu}^{N}\left[T^{i}\psi\right]n^{\mu}_{\Sigma_{0}}}+C\sum_{i=1}^{k}\int_{\Sigma_{0}\cap\mathcal{A}}{J_{\mu}^{N}\left[\partial^{i}_{r}\psi\right]n^{\mu}_{\Sigma_{0}}}
\label{hok-1}
\end{split}
\end{equation}
for all $k\leq m$ and  solutions supported on frequencies greater or equal to $m$.  Note that for $m=0$ the above estimate is the result of Section 11 and for $m=1$ of Section 12 of \cite{aretakis1}.

In order to show the above estimate for $m=l$ we will construct an appropriate multiplier $L_{k}=L^{v}_{k}T+L^{r}_{k}\partial_{r}$ which is a future directed causal $\varphi_{\tau}^{T}$-invariant vector field such that  $L_{k}=\textbf{0}$ in $r\geq r_{1}$ and timelike in the region $M<r_{1}$. The precise form of $L_{k}^{v}, L_{k}^{r}$ will be chosen  only in the end.  Again, we will be interested in the region $\left\{M\leq r\leq r_{0}<r_{1}\right\}$, where $r_{0},r_{1}$ are to be determined later. The control for the higher order derivatives will be derived from the energy identity of the current $J_{\mu}^{L_{k}}\left[\partial_{r}^{k}\psi\right]$
\begin{equation}
\int_{\Sigma_{\tau}}{J^{L_{k}}_{\mu}\!\!\left[\partial_{r}^{k}\psi\right]n_{\Sigma_{\tau}}^{\mu}}+\int_{\mathcal{R}}{\nabla^{\mu}J^{L_{k}}_{\mu}\!\!\left[\partial_{r}^{k}\psi\right]}+\int_{\mathcal{H}^{+}}{J^{L_{k}}_{\mu}\!\!\left[\partial_{r}^{k}\psi\right]n_{\mathcal{H}^{+}}^{\mu}}=\int_{\Sigma_{0}}{J^{L_{k}}_{\mu}\!\!\left[\partial_{r}^{k}\psi\right]n_{\Sigma_{0}}^{\mu}}.
\label{eideho}
\end{equation}
By assumption, the right hand side is bounded. Also, since $L_{k}$ is timelike in a spatially compact region, we have 
\begin{equation}
J_{\mu}^{L_{k}}\!\!\left[\partial_{r}^{k}\psi\right]n^{\mu}_{\Sigma_{\tau}}\sim \left(T\partial_{r}^{k}\psi\right)^{2}+\left(\partial_{r}^{k+1}\psi\right)^{2}+\left|\nabb\partial_{r}^{k}\psi\right|^{2}
\label{hoLst}
\end{equation}
and on the horizon
\begin{equation}
J_{\mu}^{L_{k}}\!\!\left[\partial_{r}^{k}\psi\right]n^{\mu}_{\mathcal{H}^{+}}= L^{v}_{k}(M)\left(T\partial_{r}^{k}\psi\right)^{2}-\frac{L^{r}_{k}(M)}{2}\left|\nabb\partial_{r}^{k}\psi\right|^{2}.
\label{hoLH}
\end{equation}
Note that $L_{k}^{r}<0$ on the horizon. 
It suffices to estimate the bulk integral. We have
\begin{equation*}
\begin{split}
\nabla^{\mu}J_{\mu}^{L_{k}}\!\!\left[\partial_{r}^{k}\psi\right]&=K^{L_{k}}\!\left[\partial_{r}^{k}\psi\right]
+\mathcal{E}^{L_{k}}\!\left[\partial_{r}^{k}\psi\right].\\
\end{split}
\end{equation*}
Recall that 
\begin{equation*}
K^{L_{k}}\!\left[\partial_{r}^{k}\psi\right]=F_{vv}(T\partial^{k}\psi)^{2}+F_{rr}(\partial^{k+1}\psi)^{2}+F_{\scriptsize\nabb}\left|\nabb\partial^{k}\psi\right|^{2}+F_{vr}(T\partial^{k}\psi)(\partial^{k+1}\psi),
\end{equation*}
where the coefficients $F_{ab}$ are given by
\begin{equation*}
\begin{split}
F_{vv}=\left(\partial_{r}L_{k}^{v}\right),  F_{rr}=D\left[\frac{\left(\partial_{r}L_{k}^{r}\right)}{2}-\frac{L_{k}^{r}}{r}\right]-\frac{L_{k}^{r}D'}{2}, F_{\scriptsize\nabb}=-\frac{1}{2}\left(\partial_{r}L_{k}^{r}\right), F_{vr}=D\left(\partial_{r}L_{k}^{v}\right)-\frac{2L_{k}^{r}}{r}.
\label{list00}
\end{split}
\end{equation*}
Equation \eqref{comhigherorder} gives us
\begin{equation*}
\begin{split}
&\mathcal{E}^{L_{k}}\!\left[\partial_{r}^{k}\psi\right]=\left(\Box_{g}\partial_{r}^{k}\psi\right)L_{k}\left[\partial_{r}^{k}\psi\right]\\
=&\left[-\sum_{i=1}^{k}{
\binom{k}{i}\partial_{r}^{i}D\cdot \partial_{r}^{k-i+2}\psi}-\sum_{i=1}^{k}{\binom{k}{i}\partial_{r}^{i}\frac{2}{r}\cdot T\partial_{r}^{k-i}\psi}\right.\\
&\left.-\sum_{i=1}^{k}{\binom{k}{i}\partial_{r}^{i}R\cdot\partial_{r}^{k-i+1}\psi}-\sum_{i=1}^{k}{\binom{k}{i}r^{2}\partial_{r}^{i}r^{-2}\cdot\lapp\partial_{r}^{k-i}\psi}\right]L_{k}\!\left[\partial_{r}^{k}\psi\right]\\
=&-\sum_{i=1}^{k}{\binom{k}{i}L^{v}_{k}\cdot\partial_{r}^{i}D\left(\partial_{r}^{k-i+2}\psi\right)\left(T\partial_{r}^{k}\psi\right)}-\sum_{i=1}^{k}{\binom{k}{i}L^{v}_{k}\cdot\partial_{r}^{i}\frac{2}{r}\left(T\partial_{r}^{k-i}\psi\right)\left(T\partial_{r}^{k}\psi\right)}\\
&-\sum_{i=1}^{k}{\binom{k}{i}L^{v}_{k}\cdot\partial_{r}^{i}R\left(\partial_{r}^{k-i+1}\psi\right)\left(T\partial_{r}^{k}\psi\right)}-\sum_{i=1}^{k}{\binom{k}{i}L^{v}_{k}\cdot r^{2}\partial_{r}^{i}r^{-2}\left(\lapp\partial_{r}^{k-i}\psi\right)\left(T\partial_{r}^{k}\psi\right)}\\
&-\sum_{i=1}^{k}{\binom{k}{i}L^{r}_{k}\cdot\partial_{r}^{i}D\left(\partial_{r}^{k-i+2}\psi\right)\left(\partial_{r}^{k+1}\psi\right)}-\sum_{i=1}^{k}{\binom{k}{i}L^{r}_{k}\cdot\partial_{r}^{i}\frac{2}{r}\left(T\partial_{r}^{k-i}\psi\right)\left(\partial_{r}^{k+1}\psi\right)}\\
&-\sum_{i=1}^{k}{\binom{k}{i}L^{r}_{k}\cdot\partial_{r}^{i}R\left(\partial_{r}^{k-i+1}\psi\right)\left(\partial_{r}^{k+1}\psi\right)}-\sum_{i=1}^{k}{\binom{k}{i}L^{r}_{k}\cdot r^{2}\partial_{r}^{i}r^{-2}\left(\lapp\partial_{r}^{k-i}\psi\right)\left(\partial_{r}^{k+1}\psi\right)}.
\end{split}
\end{equation*}
In order to obtain the sharp result we  need the following lemma
\begin{lemma}
Suppose $\psi$ is a solution to the wave equation which is supported on the (fixed) frequency $l$. Then for all $0\leq i\leq l-1$ and any positive number $\epsilon$ we have
\begin{equation*}
\begin{split}
\left|\int_{\mathcal{H}^{+}}{(\partial_{r}^{i}\psi)(\partial_{r}^{l}\psi)}\right|\leq & C_{\epsilon}\sum_{i=0}^{l-1}\int_{\Sigma_{0}}{J_{\mu}^{N}[T^{i}\psi]n^{\mu}_{\Sigma_{0}}}+C_{\epsilon}\sum_{i=0}^{l}\int_{\Sigma_{0}}{J_{\mu}^{N}[\partial_{r}^{i}\psi]n^{\mu}_{\Sigma_{0}}}\\ &+\epsilon\int_{\Sigma_{\tau}\cap\mathcal{A}}{J_{\mu}^{L_{l}}[\partial_{r}^{l}\psi]n^{\mu}_{\Sigma_{\tau}}}+\epsilon\int_{\mathcal{H}^{+}}{(T\partial_{r}^{l}\psi)^{2}},
\end{split}
\end{equation*}
where $C_{\epsilon}$ depends on $M$, $l$ and $\Sigma_{0}$.
\label{lemk}
\end{lemma}
\begin{proof}
Using Theorem \ref{t3} of Section \ref{sec:TheMainTheorems} it suf{}fices to estimate the integrals
\begin{equation*}
\int_{\mathcal{H}^{+}}{(T\partial_{r}^{j}\psi)(\partial_{r}^{l}\psi)},
\end{equation*}
where $0\leq j\leq l$. For $0\leq j\leq l-1$ we have
\begin{equation*}
\begin{split}
\int_{\mathcal{H}^{+}}{(T\partial_{r}^{j}\psi)(\partial_{r}^{l}\psi)}=&\int_{\mathcal{H}^{+}\cap\Sigma_{\tau}}{(\partial_{r}^{j}\psi)(\partial_{r}^{l}\psi)}-\int_{\mathcal{H}^{+}\cap\Sigma_{0}}{(\partial_{r}^{j}\psi)(\partial_{r}^{l}\psi)}-\int_{\mathcal{H}^{+}}{(\partial_{r}^{j}\psi)(T\partial_{r}^{l}\psi)}.
\end{split}
\end{equation*}
The two boundary integrals can be estimated using the second Hardy inequality of Section 6 of \cite{aretakis1}. Regarding the last integral on the right hand side, the Cauchy-Schwarz and Poincar\'e inequality imply
\begin{equation*}
\begin{split}
\int_{\mathcal{H}^{+}}{(\partial_{r}^{j}\psi)(T\partial_{r}^{l}\psi)}\leq \int_{\mathcal{H}^{+}}{\frac{1}{\epsilon}\left|\nabb\partial_{r}^{j}\psi\right|^{2}+\epsilon(T\partial_{r}^{l}\psi)^{2}}.
\end{split}
\end{equation*}
For $j=l$, we use that $(T\partial_{r}^{l}\psi)(\partial_{r}^{l}\psi)=\frac{1}{2}T(\partial_{r}^{l}\psi)^{2}$ and the second Hardy inequality.
\end{proof}
We are now in position to estimate the bulk integrals. We decompose $\psi=\psi_{l}+\psi_{\geq l+1}$. This is needed in view of the factor $\chi_{\left\{k=l\right\}}$ in the statement of the theorem. Note that \textbf{everytime we say that an integral can be estimated we mean that can be estimated by terms appearing on the right hand side of \eqref{hok-1} and $\epsilon$'s of the spacetime terms appearing on the left.}
\begin{center}
\large{\textbf{Estimate for} $\displaystyle\int_{\mathcal{R}}{H_{i}^{1}\left(\partial_{r}^{k-i+1}\psi\right)\left(\partial_{r}^{k+1}\psi\right)}, i\geq 0$}\\
\end{center}
For $i=0$ we have $H_{0}^{1}=-kL^{r}_{k}D'\geq 0$
and so this coefficient has the ``right'' sign in  \eqref{eideho}.

For $i\geq 1$  we use integration by parts 
\begin{equation*}
\begin{split}
&\int_{\mathcal{R}}{H_{i}^{1}\left(\partial_{r}^{k-i+1}\psi\right)\left(\partial_{r}^{k+1}\psi\right)}+\int_{\mathcal{R}}{\left(\partial_{r}H_{i}^{1}+\frac{2}{r}H_{i}^{1}\right)\left(\partial_{r}^{k-i+1}\psi\right)\left(\partial_{r}^{k}\psi\right)}+\int_{\mathcal{R}}{H_{i}^{1}\left(\partial_{r}^{k-i+2}\psi\right)\left(\partial_{r}^{k}\psi\right)}\\
=&\int_{\Sigma_{0}}{H_{i}^{1}\left(\partial_{r}^{k-i+1}\psi\right)\left(\partial_{r}^{k}\psi\right)\partial_{r}\cdot n_{\Sigma_{0}}}-\int_{\Sigma_{\tau}}{H_{i}^{1}\left(\partial_{r}^{k-i+1}\psi\right)\left(\partial_{r}^{k}\psi\right)\partial_{r}\cdot n_{\Sigma_{\tau}}}-\int_{\mathcal{H}^{+}}{H_{i}^{1}\left(\partial_{r}^{k-i+1}\psi\right)\left(\partial_{r}^{k}\psi\right)\partial_{r}\cdot n_{\mathcal{H}^{+}}}.
\end{split}
\end{equation*}

From Proposition 12.4.1. of \cite{aretakis1} applied (and generalised) for $\partial_{r}^{k}\psi$ and the inductive hypothesis all the  integrals over $\mathcal{A}$ and $\Sigma_{\tau}$ can be estimated. It only remains to estimate the integral over $\mathcal{H}^{+}$ when $i=1$. In this case, if we apply the Poincar\'e inequality for the $\psi_{l}$ component we notice that we need to absorb a good term  in the divergence identity for $L_{k}$ whenever $k=l$. Indeed,
\begin{equation}
\begin{split}
\frac{1}{2}\frac{M^{2}}{l\left(l+1\right)}H_{1}^{1}\left(M\right)&= -\frac{L^{r}_{k}\left(M\right)}{2}\Leftrightarrow  k= l
\end{split}
\label{crucial}
\end{equation}
This implies that we cannot use the Poincar\'e inequality on $\mathcal{H}^{+}$ in order to estimate $(\partial_{r}^{l}\psi_{l})^{2}$ anymore. That is why we proved Lemma \ref{lemk} which we will use for the following integrals. Clearly, for the component $\psi_{\geq l+1}$ we need a fraction of $\left|\nabb\partial_{r}^{l}\psi\right|^{2}$ along $\mathcal{H}^{+}$ and thus we can take small (epsilon) portions of this term later on. 

\begin{center}
\large{\textbf{Estimate for} $\displaystyle\int_{\mathcal{R}}{H_{i}^{2}\left(T\partial_{r}^{k-i+1}\psi\right)\left(\partial_{r}^{k+1}\psi\right)}, i \geq 1$}\\
\end{center}
For $i\geq 2$ we use Stokes' theorem
\begin{equation*}
\begin{split}
&\int_{\mathcal{R}}{H_{i}^{2}\left(T\partial_{r}^{k-i+1}\psi\right)\left(\partial_{r}^{k+1}\psi\right)}+\int_{\mathcal{R}}{\left(\partial_{r}H_{i}^{2}+\frac{2}{r}H_{i}^{2}\right)\left(T\partial_{r}^{k-i+1}\psi\right)\left(\partial_{r}^{k}\psi\right)}+\int_{\mathcal{R}}{H_{i}^{2}\left(T\partial_{r}^{k-i+2}\psi\right)\left(\partial_{r}^{k}\psi\right)}\\
=&\int_{\Sigma_{0}}{H_{i}^{2}\left(T\partial_{r}^{k-i+1}\psi\right)\left(\partial_{r}^{k}\psi\right)\partial_{r}\cdot n_{\Sigma_{0}}}-\int_{\Sigma_{\tau}}{H_{i}^{2}\left(T\partial_{r}^{k-i+1}\psi\right)\left(\partial_{r}^{k}\psi\right)\partial_{r}\cdot n_{\Sigma_{\tau}}}-\int_{\mathcal{H}^{+}}{H_{i}^{2}\left(T\partial_{r}^{k-i+1}\psi\right)\left(\partial_{r}^{k}\psi\right)\partial_{r}\cdot n_{\mathcal{H}^{+}}}.
\end{split}
\end{equation*}
Again from Proposition 12.4.1.~of \cite{aretakis1} and the inductive hypothesis we can estimate all the above integrals. Note that in order to estimate the integrals along $\mathcal{H}^{+}$ for the component $\psi_{l}$ we follow the same argument as in the proof of Lemma \ref{lemk}.
For $i=1$ we use the wave equation and thus
\begin{equation*}
\begin{split}
H_{1}^{2}\left(T\partial_{r}^{k}\psi\right)\!\left(\partial_{r}^{k+1}\psi\right)=&H_{1}^{2}\left[\partial_{r}^{k-1}\left(-D\partial_{r}^{2}\psi-\frac{2}{r}T\psi-R\partial_{r}\psi-\lapp\psi\right)\right]\!\left(\partial_{r}^{k+1}\psi\right)\\
=&-\sum_{j=0}^{k-1}{\binom{k-1}{j}H_{1}^{2}\partial_{r}^{j}D\left(\partial_{r}^{k-j+1}\psi\right)\left(\partial_{r}^{k+1}\psi\right)}\\
&-\sum_{j=0}^{k-1}{\binom{k-1}{j}H_{1}^{2}\partial_{r}^{j}\frac{2}{r}\left(T\partial_{r}^{k-j-1}\psi\right)\left(\partial_{r}^{k+1}\psi\right)}\\
&-\sum_{j=0}^{k-1}{\binom{k-1}{j}H_{1}^{2}\partial_{r}^{j}R\left(\partial_{r}^{k-j}\psi\right)\left(\partial_{r}^{k+1}\psi\right)}\\
&-H_{1}^{2}\left(\partial_{r}^{k-1}\lapp\psi\right)\left(\partial_{r}^{k+1}\psi\right).\\
\end{split}
\end{equation*}
The integrals of the first sum can be estimated for $j=0,1$ since their coef{}ficients vanish on the horizon and the case $j\geq 2$ was investigated above. 

The integrals of the second sum were also estimated before. 

For $j=0$ the integral of the third sum can be estimated since its coefficient vanishes on $\mathcal{H}^{+}$. If $j\geq 1$ then again these integrals have been estimated.
It remains to estimate the integral of the last term. Integration by parts gives 
\begin{equation*}
\begin{split}
&\int_{\mathcal{R}}{H_{1}^{2}\left(\partial_{r}^{k-1}\lapp\psi\right)\left(\partial_{r}^{k+1}\psi\right)}+\int_{\mathcal{R}}{\left(\partial_{r}H_{1}^{2}+\frac{2}{r}H_{1}^{2}\right)\left(\partial_{r}^{k-1}\lapp\psi\right)\left(\partial_{r}^{k}\psi\right)}+\int_{\mathcal{R}}{H_{1}^{2}\left(\partial_{r}^{k}\lapp\psi\right)\left(\partial_{r}^{k}\psi\right)}\\
=&\int_{\Sigma_{0}}{H_{1}^{2}\left(\partial_{r}^{k-1}\lapp\psi\right)\left(\partial_{r}^{k}\psi\right)\partial_{r}\cdot n_{\Sigma_{0}}}-\int_{\Sigma_{\tau}}{H_{1}^{2}\left(\partial_{r}^{k-1}\lapp\psi\right)\left(\partial_{r}^{k}\psi\right)\partial_{r}\cdot n_{\Sigma_{\tau}}}-\int_{\mathcal{H}^{+}}{H_{1}^{2}\left(\partial_{r}^{k-1}\lapp\psi\right)\left(\partial_{r}^{k}\psi\right)\partial_{r}\cdot n_{\mathcal{H}^{+}}}.
\end{split}
\end{equation*}
If we set $Q=\partial_{r}H_{1}^{2}+\frac{2}{r}H_{1}^{2}$ then 
\begin{equation*}
\begin{split}
&\int_{\mathcal{R}}{Q\left(\partial_{r}^{k-1}\lapp\psi\right)\left(\partial_{r}^{k}\psi\right)}=\int_{\mathcal{R}}{Q\left(\lapp\partial_{r}^{k-1}\psi-\left[\lapp,\partial_{r}^{k-1}\right]\psi\right)\left(\partial_{r}^{k}\psi\right)}\\
=&\int_{\mathcal{R}}{Q\left(\lapp\partial_{r}^{k-1}\psi+\sum_{n=1}^{k-1}{\binom{k-1}{n}r^{2}\partial_{r}^{n}r^{-2}\cdot \lapp\partial_{r}^{k-1-n}\psi}\right)\left(\partial_{r}^{k}\psi\right)}.\\
\end{split}
\end{equation*}
We estimate this integral by applying Stokes' theorem on $\mathbb{S}^{2}$ and Cauchy-Schwarz and using the inductive hypothesis. Regarding the last bulk integral we have
\begin{equation*}
\begin{split}
&\int_{\mathcal{R}}{H_{1}^{2}\left(\partial_{r}^{k}\lapp\psi\right)\left(\partial_{r}^{k}\psi\right)}=\int_{\mathcal{R}}{H_{1}^{2}\left(\lapp\partial_{r}^{k}\psi-\left[\lapp,\partial_{r}^{k}\right]\psi\right)\left(\partial_{r}^{k}\psi\right)}\\=&\int_{\mathcal{R}}{H_{1}^{2}\left(\lapp\partial_{r}^{k}\psi+\sum_{n=1}^{k}{\binom{k}{n}r^{2}\partial_{r}^{n}r^{-2}\lapp\partial_{r}^{k-n}\psi}\right)\left(\partial_{r}^{k}\psi\right)}.\\
\end{split}
\end{equation*}
Therefore, by applying Stokes' theorem on $\mathbb{S}^{2}$ we see that this integral can be estimated provided we have $
H_{1}^{2}<-\frac{1}{2}\partial_{r}L^{r}_{k}\left(M\right).$ Note that $H_{1}^{2}$ does not depend on $\partial_{r}L^{r}_{k}$.  Furthermore, the boundary integral over $\mathcal{H}^{+}$ can be estimated as follows: For the component $\psi_{l}$ we have $\lapp\psi_{l}=-\frac{l(l+1)}{r^{2}}\psi_{l}$  and thus $\partial_{r}^{k-1}\lapp\psi_{l}$ depends on $\partial_{r}^{i}\psi_{l}$ for $0\leq i\leq k-1$ and thus we  use Lemma \ref{lemk}. For the component $\psi_{\geq l+1}$ we commute $\partial_{r}^{k-1}$ and $\lapp$, we use Stokes' theorem on $\mathbb{S}^{2}$ and  Cauchy-Schwarz as above.

\begin{center}
\large{\textbf{Estimate for} $\displaystyle\int_{\mathcal{R}}{H_{i}^{3}\left(\partial_{r}^{k-i+2}\psi\right)\left(T\partial_{r}^{k}\psi\right)}, i\geq 1$}\\
\end{center}
The case $i=1$ was investigated above. 

For $i\geq 2$ we use Proposition 12.4.1. of \cite{aretakis1} and Cauchy-Schwarz.

\begin{center}
\large{\textbf{Estimate for} $\displaystyle\int_{\mathcal{R}}{H_{i}^{4}\left(T\partial_{r}^{k-i}\psi\right)\left(T\partial_{r}^{k}\psi\right)}, i\geq 1$}\\
\end{center}
We use Cauchy-Schwarz and the inductive hypothesis.

\begin{center}
\large{\textbf{Estimate for} $\displaystyle\int_{\mathcal{R}}{H_{i}^{5}\left(\lapp\partial_{r}^{i}\psi\right)\left(T\partial_{r}^{k}\psi\right)}, i\leq k-1$}\\
\end{center}
For $i=0$ we solve with respect to $\lapp\psi$ in the wave equation and then use Cauchy-Schwarz and the inductive hypothesis. We proceed by induction on $i$. Assuming that $(\lapp\partial_{r}^{j}\psi)(T\partial_{r}^{k}\psi)$ is estimated for all $0\leq j\leq i-1$ we will prove that the integral $(\lapp\partial_{r}^{i}\psi)(T\partial_{r}^{k}\psi)$  can also be estimated. We have
\begin{equation*}
\begin{split}
\lapp\partial_{r}^{i}\psi\left(T\partial_{r}^{k}\psi\right)=\partial_{r}^{i}\lapp\psi\left(T\partial_{r}^{k}\psi\right)+[\lapp,\partial_{r}^{i}]\psi\left(T\partial_{r}^{k}\psi\right)
\end{split}
\end{equation*}
The first term on the right hand side can be estimated by solving with respect to $\lapp\psi$ in the wave equation and using Cauchy-Schwarz. The second term can be estimated by our inductive hypothesis.

\begin{center}
\large{\textbf{Estimate for} $\displaystyle\int_{\mathcal{R}}{H_{i}^{6}\left(\lapp\partial_{r}^{k-i}\psi\right)\left(\partial_{r}^{k+1}\psi\right)}, i\geq 1$}\\
\end{center}
Integration by parts yields
\begin{equation*}
\begin{split}
&\int_{\mathcal{R}}{H_{i}^{6}\left(\lapp\partial_{r}^{k-i}\psi\right)\left(\partial_{r}^{k+1}\psi\right)}+\int_{\mathcal{R}}{\left(\partial_{r}H_{i}^{6}+\frac{2}{r}H_{i}^{6}\right)\left(\lapp\partial_{r}^{k-i}\psi\right)\left(\partial_{r}^{k}\psi\right)}+\int_{\mathcal{R}}{H_{i}^{6}\left(\partial_{r}\lapp\partial_{r}^{k-i}\psi\right)\left(\partial_{r}^{k}\psi\right)}\\
=&\int_{\Sigma_{0}}{H_{i}^{6}\left(\lapp\partial_{r}^{k-i}\psi\right)\left(\partial_{r}^{k}\psi\right)\partial_{r}\cdot n_{\Sigma_{0}}}-\int_{\Sigma_{\tau}}{H_{i}^{6}\left(\lapp\partial_{r}^{k-i}\psi\right)\left(\partial_{r}^{k}\psi\right)\partial_{r}\cdot n_{\Sigma_{\tau}}}\\&-\int_{\mathcal{H}^{+}}{H_{i}^{6}\left(\lapp\partial_{r}^{k-i}\psi\right)\left(\partial_{r}^{k}\psi\right)\partial_{r}\cdot n_{\mathcal{H}^{+}}}.\\
\end{split}
\end{equation*}
The second bulk integral and the boundary integrals over $\Sigma$ are estimated  using Stokes' theorem on $\mathbb{S}^{2}$ and the inductive hypothesis. The last bulk integral is estimated by commuting the spherical Laplacian with $\partial_{r}$ and applying again Stokes' theorem on $\mathbb{S}^{2}$ and Cauchy-Schwarz. Thus for $i\geq 2$ we use the inductive hypothesis and for $i=1$ it suf{}fices to have $H_{1}^{6}<-\frac{1}{2}\partial_{r}L^{r}_{k}\left(M\right).$
Regarding the integral over $\mathcal{H}^{+}$ we have the following: For the component $\psi_{l}$ we have $\lapp\psi_{l}=-\frac{l(l+1)}{r^{2}}\psi_{l}$  and thus $\lapp\partial_{r}^{k-i}\psi_{l}=-\frac{l(l+1)}{r^{2}}\partial_{r}^{k-i}\psi_{l}$  and so we apply Lemma \ref{lemk} to estimate it.  For the component $\psi_{\geq l+1}$ we use Stokes' theorem on $\mathbb{S}^{2}$, Cauchy-Schwarz and the inductive hypothesis.

The construction of $L_{k}$ is now clear for all $k\in\mathbb{N}$. It suf{}fices to take $L^{r}_{k}\left(M\right)<0$, $L^{v}_{k}\left(M\right)>0$ and $-\partial_{r}L_{k}^{r}$ and $\partial_{r}L_{k}^{v}\left(M\right)$ suf{}ficiently large.  

Note that  no commutation with the generators of the Lie algebra so(3) is required.

\end{proof}

\subsection{Improved $L^{2}$ Estimates for $T^{m}\psi,m\geq 1$}
\label{sec:Applications}

We next show how to apply the results of Section \ref{sec:ConservationLawsOnDegenerateEventHorizons} to obtain improved results for $T^{m}\psi,m\geq 1$. We show the following
\begin{proof}[Proof of Theorem \ref{theorem3}, statement (2)]

Recall that the only reason we had to restrict $k$ to be such that $k\leq l$ is for estimating the integral $\int_{\mathcal{R}}{H^{1}_{1}(\partial_{r}^{k}\psi)(\partial_{r}^{k+1}\psi)}$. Specifically, we saw that this integral can be controlled by other ``good'' terms in the energy indentity \eqref{eideho} only if $k\leq l$ since in this case we can estimate the integral
\begin{equation*} I_{k}[\psi]=\int_{\mathcal{H}^{+}}{(\partial_{r}^{k}\psi)^{2}}
\end{equation*}
 using the Poincar\'e inequality. Clearly, if we could show that the integral $I$ is bounded, then no use of the Poincar\'e inequality would be required and thus by working as in Section \ref{sec:TheMultiplierLKAndTheEnergyIdentity} we could derive  $L^{2}$ estimates for even higher derivatives of $\psi$. We first prove the following

\begin{lemma}
There exist constants $C$ which depend on $M,k,l,m$ such that for all solutions $\psi$ of the wave equation which are supported on the frequency $l$ we have
\begin{enumerate}
	\item 
	\begin{equation*}
	I_{k}[T\psi]\leq C\sum_{i=0}^{k-1}\int_{\Sigma_{0}}J_{\mu}^{N}[T^{i}\psi]n^{\mu}_{\Sigma_{0}}+C\sum_{i=0}^{k-1}\int_{\Sigma_{0}\cap\mathcal{A}}J_{\mu}^{N}[\partial_{r}^{i}\psi]n^{\mu}_{\Sigma_{0}},
	\end{equation*}for all  $1\leq k\leq l$.

	\item 	\begin{equation*}
	I_{l+1}[T\psi]\leq C\sum_{i=0}^{l-1}\int_{\Sigma_{0}}J_{\mu}^{N}[T^{i}\psi]n^{\mu}_{\Sigma_{0}}+C\sum_{i=0}^{l-1}\int_{\Sigma_{0}\cap\mathcal{A}}J_{\mu}^{N}[\partial_{r}^{i}\psi]n^{\mu}_{\Sigma_{0}}.
	\end{equation*}
	
		\item 	\begin{equation*}
	I_{l+m}[T^{m}\psi]\leq C\sum_{i=0}^{l+m-2}\int_{\Sigma_{0}}J_{\mu}^{N}[T^{i}\psi]n^{\mu}_{\Sigma_{0}}+C\sum_{i=0}^{l-1}\sum_{j=0}^{m-1}\int_{\Sigma_{0}\cap\mathcal{A}}J_{\mu}^{N}[\partial_{r}^{i}T^{j}\psi]n^{\mu}_{\Sigma_{0}}.
	\end{equation*}
	
\end{enumerate}

\label{lemmat}
\end{lemma}
\begin{proof}
First note that $I_{0}[T\psi]\leq \int_{\Sigma_{0}}J_{\mu}^{N}[\psi]n^{\mu}_{\Sigma_{0}}$. For $k=1$ (and so $l\geq 1$) we use $\Box_{g}\psi=0$ and that the zeroth order term can be bounded using the Poincar\'e inequality to obtain $I_{1}[T\psi]\leq C\int_{\Sigma_{0}}J_{\mu}^{N}[\psi]n^{\mu}_{\Sigma_{0}}$. For $k=2$ (and so $l\geq 2$) using $\partial_{r}(\Box_{g}\psi)=0$ we have that $\partial_{r}^{2}T\psi$ can be expressed as a linear combination of $\partial_{r}T\psi, \partial_{r}\psi, T\psi,\psi$ on $\hh$. Note that (the integral of) $\partial_{r}T\psi,  T\psi,\psi$ can be estimated using previous results. The term $\partial_{r}\psi$ can be estimated using the first statement of Theorem \ref{theorem3}.   Inductively, using $\partial_{r}^{k-1}(\Box_{g}\psi)=0$ (see \eqref{kcom}) we obtain
\begin{equation*}
I_{k}[T\psi]\leq C\sum_{i=0}^{k-1}\int_{\hh}(\partial_{r}^{i}\psi)^{2}+C\sum_{i=0}^{k-1}I_{i}[T\psi].
	\end{equation*}
The first part follows from the inductive hypothesis and the first statement of Theorem \ref{theorem3} (as long as $k\leq l$).

For the critical case of the second part note that for  $l=0$ we have   $\partial_{r}T\psi=-\frac{1}{M}T\psi$ on $\hh$ and so  $I_{1}[T\psi]\leq \int_{\Sigma_{0}}J_{\mu}^{N}[\psi]n^{\mu}_{\Sigma_{0}}$. On the other hand, for $l\geq 1$, since  $H_{l}[T\psi]=0$, we have that $\partial_{r}^{l+1}T\psi$ can be written as a linear combination of $\partial_{r}^{i}T\psi,i=0,1,...,l$ on $\hh$. Therefore, the first part and the Cauchy-Schwarz  inequality finish the proof of the second part.

The third part can be proved by induction on $m$. For $m=1$ the result has been proved. By considering now $\eqref{kcom}$ for $k=l+m-1$ and $\psi$ replaced with $T^{m-1}\psi$ we take that $\partial_{r}^{l+m}T^{m}\psi$ can be written (on $\hh$) as a linear combination of $\partial_{r}^{l+m-1}T^{m}\psi$, $\partial_{r}^{l+m-1}T^{m-1}\psi$ and other lower order terms. The integral of $\partial_{r}^{l+m-1}T^{m-1}\psi$ can be estimated inductively. Finally, for the term $\partial_{r}^{l+m-1}T^{m}\psi$ we observe that $I_{l+m-1}[T^{m}\psi]=I_{l+m-1}[T^{m-1}(T\psi)]$, which can also be estimated inductively. Note that in view of Proposition \ref{tmpsi}  no such estimate holds for $I_{k}[T^{m}\psi]$ for $k\geq l+m+1$.
\end{proof}
We can now show that $I_{k}[T^{m}\psi]$ is bounded whenever $m\geq 1$ and $k\leq l+m$. Indeed, if $k\leq l$ then $I_{k}[T^{m}\psi]=I_{k}[T(T^{m-1}\psi)]$ and the result follows from the first part of Lemma \ref{lemmat}. If $l\leq k\leq l+m$ then 
\begin{equation*}
I_{k}[T^{m}\psi]=I_{l+m-i}[T^{m-i}(T^{i}\psi)]=I_{l+m'}[T^{m'}(T^{i}\psi)], 
	\end{equation*}
which can also be estimated using Lemma \ref{lemmat} since $m'=m-i\geq 1$.

Finally, we need to show that $I_{k}[T^{m}\psi],k\leq l+m$, is bounded for all $\psi_{\geq l}$ which are supported on frequencies greater or equal to $l$. Indeed, if $k\leq l$ then we simply use the first statement of Theorem $\ref{theorem3}$. If $k>l$ then 
\begin{equation*}
I_{k}[T^{m}\psi_{\geq l}]=I_{k}[T^{m}\psi_{l}]+I_{k}[T^{m}\psi_{l+1}]+\cdots+I_{k}[T^{m}\psi_{k}]+I_{k}[T^{m}\psi_{\geq k+1}]
\end{equation*}
The last term on the right hand side can be bounded using Theorem \ref{theorem3} again. The remaining terms can be estimated using the above results. Therefore, no use of the flux of $L_{k}$ along $\mathcal{H}^{+}$ is needed whenever $m\geq 1$.

\end{proof}

\section{Energy Decay}
\label{sec:EnergyDecay}

In this section we derive the decay for the non-degenerate energy flux of $N$ through an appropriate foliation. The first step is to obtain non degenerate estimates on regions which connect $\mathcal{H}^{+}$ and $\mathcal{I}^{+}$ (without containing $i^{0}$; this has to do with the fact that energy is radiated away through null infinity). Such estimates were first derived in the recent \cite{new} along with a new robust method for obtaining decay results. Here we establish several estimates which will allow us to adapt the methods of \cite{new} in the extreme case. These new estimates are closely related with the trapping properties of $\mathcal{H}^{+}$.

Recall the  $\tilde{\Sigma}_{\tau}$ foliation defined in Section \ref{sec:TheMainTheorems}. For arbitrary $\tau_{1}<\tau_{2}$ we define 
\begin{equation*}
\begin{split}
&\tilde{\mathcal{R}}_{\tau_{1}}^{\tau_{2}}=\cup_{\tau\in\left[\tau_{1},\tau_{2}\right]}{\tilde{\Sigma}_{\tau}},\  \tilde{\mathcal{D}}_{\tau_{1}}^{\tau_{2}}=\tilde{\mathcal{R}}_{\tau_{1}}^{\tau_{2}}\cap\left\{r\geq R_{0}\right\}, \  \tilde{N}_{\tau}=\tilde{\Sigma}_{\tau}\cap\left\{r\geq R_{0}\right\},\  \Delta_{\tau_{1}}^{\tau_{2}}=\tilde{\mathcal{R}}_{\tau_{1}}^{\tau_{2}}\cap\left\{r= R_{0}\right\}.
\end{split}
\end{equation*}
 \begin{figure}[H]
	\centering
		\includegraphics[scale=0.10]{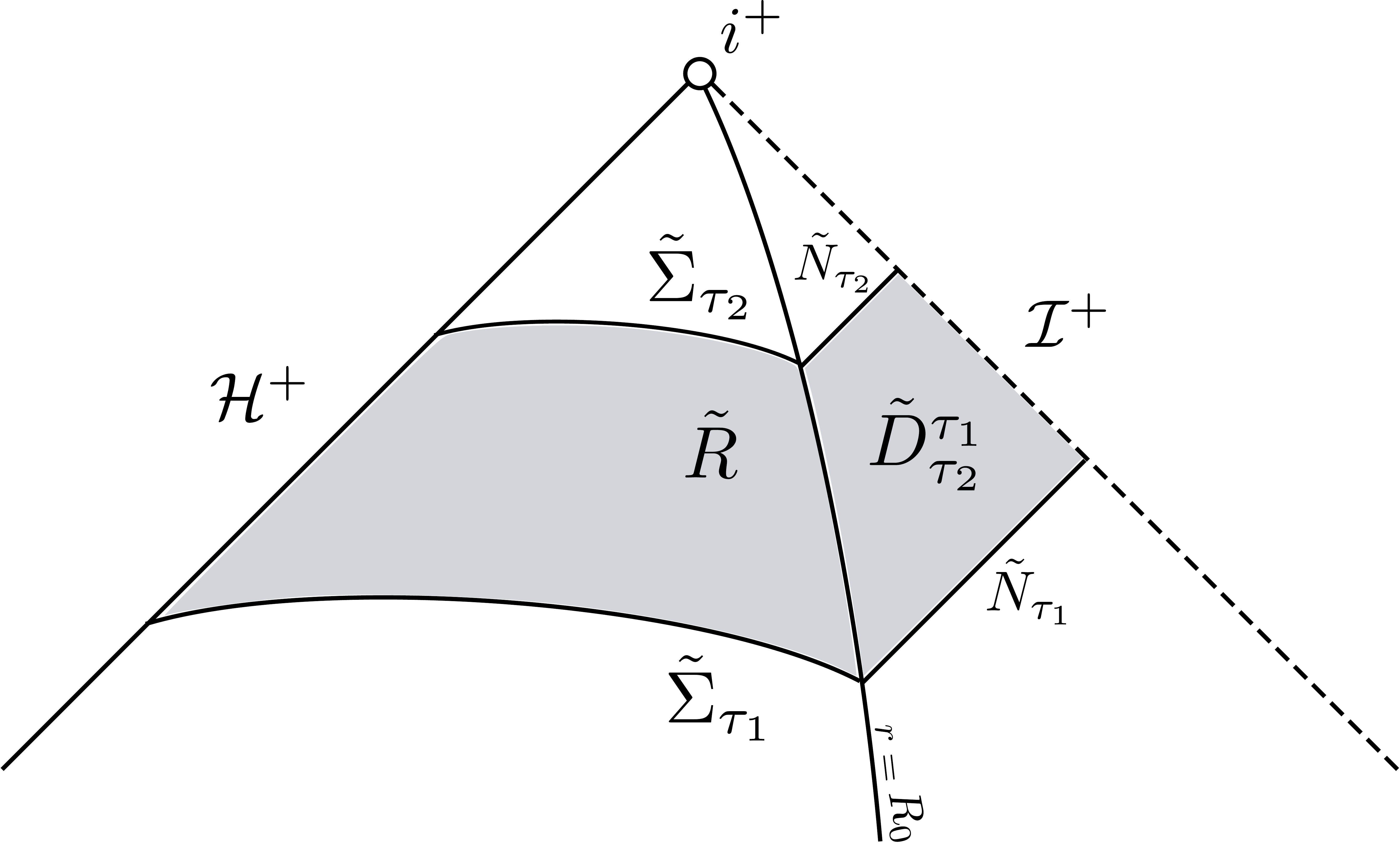}
	\label{fig:pic1ern}
\end{figure}

\subsection{$r$-Weighted Energy Estimates  in a Neighbourhood of $\mathcal{I}^{+}$}
\label{sec:RWeightedEnergyEstimates}
The main idea is to derive a non-degenerate $L^{2}$ spacetime estimate and then derive similar estimates for its boundary terms.
From now on we work with the null $(u,v)$ coordinates  unless otherwise stated.  
\begin{proposition}
Suppose $p<3$. There exists a constant $C$ that depends  on $M$ and $\tilde{\Sigma}_{0}$ such that if $\psi$ satisfies the wave equation and $\phi=r\psi$ then
\begin{equation}
\begin{split}
&\int_{\tilde{N}_{\tau_{2}}}{r^{p}\frac{\left(\partial_{v}\phi\right)^{2}}{r^{2}}}+\int_{\tilde{\mathcal{D}}_{\tau_{1}}^{\tau_{2}}}{r^{p-1}\left(p+2\right)\frac{\left(\partial_{v}\phi\right)^{2}}{r^{2}}}+\int_{\tilde{\mathcal{D}}_{\tau_{1}}^{\tau_{2}}}{\frac{r^{p-1}}{4}\left(-pD-rD'\right)\left|\nabb\psi\right|^{2}}\\
&\ \ \ \ \ \ \ \ \ \leq C\int_{\tilde{\Sigma}_{\tau_{1}}}{J_{\mu}^{T}\left[\psi\right]n^{\mu}_{\tilde{\Sigma}_{\tau_{1}}}}+\int_{\tilde{N}_{\tau_{1}}}{r^{p}\frac{\left(\partial_{v}\phi\right)^{2}}{r^{2}}}
\end{split}
\label{rwest}
\end{equation}
\label{rweigestprop}
\end{proposition}

\begin{proof}
We first consider the cut-off function $\left.\zeta:\left[R_{0},\!\!\right.\left.+\infty\right)\right.\rightarrow [0,1]$ such that 
\begin{equation*}
\begin{split}
&\zeta(r)=0\text{   for all } r\in\left[R_{0},R_{0}+1/2\right],\ \zeta(r)=1\text{   for all } r\in\left[\left.\!R_{0}+1,+\infty \right)\right. .\\
\end{split}
\end{equation*}
Let $q=p-2$. We consider the vector field
\begin{equation*}
V=r^{q}\partial_{v}
\end{equation*}
which we apply as multiplier acting on the function $\zeta\phi$ in the region $\tilde{\mathcal{D}}_{\tau_{1}}^{\tau_{2}}$. Then
\begin{equation*}
\int_{\tilde{\mathcal{D}}_{\tau_{1}}^{\tau_{2}}}{K^{V}[\zeta\phi]+\mathcal{E}^{V}[\zeta\phi]}=\int_{\partial\tilde{\mathcal{D}}_{\tau_{1}}^{\tau_{2}}}{J_{\mu}^{V}[\zeta\phi]n^{\mu}}.
\end{equation*}
Note that for $r\geq R_{0}+1$ we have $K^{V}[\zeta\phi]=K^{V}[\phi]$ and $\mathcal{E}^{V}[\zeta\phi]=\mathcal{E}^{V}[\phi]$. Then,
\begin{equation*}
\begin{split}
K^{V}(\phi)&= \textbf{T}_{\mu\nu}\left(\nabla^{\mu}\left(r^{q}\partial_{v}\right)\right)^{\nu}=\textbf{T}_{\mu\nu}\left(\left(\nabla^{\mu}r^{q}\right)\partial_{v}\right)^{\nu}+\textbf{T}_{\mu\nu}r^{q}\left(\nabla^{\mu}\partial_{v}\right)^{\nu}\\
&=2r^{q-1}(\partial_{u}\phi)(\partial_{v}\phi)+qr^{q-1}(\partial_{v}\phi)^{2}+\frac{r^{q-1}}{4}\left(-qD-rD'\right)\left|\nabb\phi\right|^{2}.
\end{split}
\end{equation*}
Note that since $\psi$ solves the wave equation $\phi$ satisfies $\frac{4}{D}\partial_{u}\partial_{v}\phi+\frac{D'}{r}\phi-\lapp\phi=0$
and so 
\begin{equation*}
\begin{split}
\Box_{g}\phi=-\frac{2}{r}(\partial_{u}\phi-\partial_{v}\phi)+\frac{D'}{r}\phi,
\end{split}
\end{equation*}
which, as expected, depends only on the 1-jet of $\phi$. Therefore,
\begin{equation*}
\begin{split}
\mathcal{E}^{V}[\phi]&=r^{q}(\partial_{v}\phi)(\Box_{g}\phi)=-2r^{q-1}(\partial_{u}\phi)(\partial_{v}\phi)+2r^{q-1}(\partial_{v}\phi)^{2}+D'r^{q-1}\phi(\partial_{v}\phi).
\end{split}
\end{equation*}
Thus
\begin{equation*}
\begin{split}
K^{V}[\phi]+\mathcal{E}^{V}[\phi]=&(q+2)r^{q-1}(\partial_{v}\phi)^{2}+\frac{r^{q-1}}{4}\left(-qD-rD'\right)\left|\nabb\phi\right|^{2}+D'r^{q-1}\phi(\partial_{v}\phi).
\end{split}
\end{equation*}
However,
\begin{equation*}
\begin{split}
\int_{\tilde{\mathcal{D}}_{\tau_{1}}^{\tau_{2}}}{D'r^{q-1}\z\phi\left(\partial_{v}\z\phi\right)}=& \int_{\tilde{\mathcal{D}}_{\tau_{1}}^{\tau_{2}}}{r^{q-4}\frac{M}{2}D\left[\sqrt{D}(1-q)-\frac{3M}{r}\right]\left(\z\phi\right)^{2}}\\&-\int_{\Delta_{\tau_{1}}^{\tau_{2}}}{\frac{r^{q-1}}{4}D'\sqrt{D}\left(\z\phi\right)^{2}}+\int_{\mathcal{I}^{+}}{\frac{D'D}{4}r^{q-1}\left(\z\phi\right)^{2}}.
\end{split}
\end{equation*}
Note that in Minkowski spacetime we would have no zeroth order term in the wave equation. In our case we do have, in such a way, however, such that the terms on the right hand side of the above identity have the right sign for $p<3$ and sufficiently large\footnote{Clearly we need to take $R_{0}>2M$.} $R_{0}$. 

In view of the cut-off function $\zeta$ all the integrals over $\Delta_{\tau_{1}}^{\tau_{2}}$ vanish. Clearly, all error terms that arise in the region\footnote{The weights in $r$ play no role in this region.} $\mathcal{W}=\overline{\operatorname{supp} (\z-1)}=\left\{R_{0}\leq r\leq R_{0}+1\right\}$  are quadratic forms of the 1-jet of $\psi$ and, therefore, these integrals are bounded by $\displaystyle\int_{\tilde{\Sigma}_{\tau_{1}}}{\!\! J_{\mu}^{T}\left[\psi\right]n^{\mu}_{\tilde{\Sigma}_{\tau_{1}}}}$. Note that only a degenerate Morawetz (see \cite{aretakis1}) is needed for the considerations near null unfinity.  Also
\begin{equation*}
\begin{split}
\int_{\partial\tilde{\mathcal{D}}_{\tau_{1}}^{\tau_{2}}}{J_{\mu}^{V}[\zeta\phi]n^{\mu}}=\int_{\tilde{N}_{\tau_{1}}}{r^{q}\left(\partial_{v}\zeta\phi\right)^{2}}-\int_{\tilde{N}_{\tau_{2}}}{r^{q}\left(\partial_{v}\zeta\phi\right)^{2}}-\int_{\mathcal{I}^{+}}{\frac{D}{4}\left|\nabb\phi\right|^{2}}.
\end{split}
\end{equation*}
The last two integrals on the right hand side appear with the right sign. Finally, in view of the first Hardy inequality of Section 6 of \cite{aretakis1}, the error terms produced by the cut-off $\zeta$ in the region $\mathcal{W}$ are controlled by  the flux of $T$ through $\tilde{\Sigma}_{\tau_{1}}$. 

\end{proof}
The reason we introduced the function $\phi$ is because the weight $r$ that it contains makes it non-degenerate ($\psi=0$ on $\mathcal{I}^{+}$ but $\phi$ does not vanish there in general). The reason we have divided by $r^{2}$ in \eqref{rwest} is because we want to emphasise the weight that corresponds to $\psi$ and not to $\phi$.

A first application of the above $r$-weighted energy estimate is the following
\begin{proposition}
There exists a constant $C$ that depends  on $M$ and $\tilde{\Sigma}_{0}$  such that if $\psi$ satisfies the wave equation and $\tilde{\mathcal{D}}_{\tau_{1}}^{\tau_{2}}$  as defined above with $R_{0}$ sufficiently large, then 
\begin{equation*}
\int_{\tau_{1}}^{\tau_{2}}{\left(\int_{\tilde{N}_{\tau}}{J^{T}_{\mu}[\psi]n^{\mu}_{\tilde{N}_{\tau}}}\right)d\tau} \,\leq\, C\int_{\tilde{\Sigma}_{\tau_{1}}}{J_{\mu}^{T}\left[\psi\right]n^{\mu}_{\tilde{\Sigma}_{\tau_{1}}}}+C\int_{\tilde{N}_{\tau_{1}}}{r^{-1}\left(\partial_{v}\phi\right)^{2}}.
\end{equation*}
\label{nondegx}
\end{proposition}
\begin{proof}
Applying Proposition \ref{rweigestprop} for $p=1$ and using the fact that for $r\geq R_{0}$ and $R_{0}$ large enough $D-rD'>\frac{1}{2},$
we have that there exists a constant $C$ that depends  on $M$ and $\tilde{\Sigma}_{0}$  such that
\begin{equation}
\int_{\tilde{\mathcal{D}}_{\tau_{1}}^{\tau_{2}}}{\frac{1}{r^{2}}\left(\partial_{v}\phi\right)^{2}+\frac{1}{r^{2}}\left|\nabb\phi\right|^{2}}\leq C\int_{\tilde{\Sigma}_{\tau_{1}}}{J_{\mu}^{T}\left[\psi\right]n^{\mu}_{\tilde{\Sigma}_{\tau_{1}}}}+C\int_{\tilde{N}_{\tau_{1}}}{r^{-1}\left(\partial_{v}\phi\right)^{2}}.
\label{corp-1}
\end{equation}
Note now that since $\left|\nabb\phi\right|^{2}=r^{2}\left|\nabb\psi\right|^{2}$, \eqref{corp-1} yields
\begin{equation*}
\int_{\tilde{\mathcal{D}}_{\tau_{1}}^{\tau_{2}}}{\left|\nabb\psi\right|^{2}}\ \,\leq C\int_{\tilde{\Sigma}_{\tau_{1}}}{J_{\mu}^{T}\left[\psi\right]n^{\mu}_{\tilde{\Sigma}_{\tau_{1}}}}+C\int_{\tilde{N}_{\tau_{1}}}{r^{-1}\left(\partial_{v}\phi\right)^{2}}.
\end{equation*}
Furthermore, for sufficiently large $R_{0}$ we have
\begin{equation*}
\begin{split}
\int_{\tilde{\mathcal{D}}_{\tau_{1}}^{\tau_{2}}}{\frac{1}{r^{2}}\left(\partial_{v}\phi\right)^{2}}\,\geq & \int_{\tilde{\mathcal{D}}_{\tau_{1}}^{\tau_{2}}}{\frac{1}{2D^{2}r^{2}}\left(\partial_{v}\phi\right)^{2}}\,=\int_{\tilde{\mathcal{D}}_{\tau_{1}}^{\tau_{2}}}{\frac{1}{2D^{2}}\left(\partial_{v}\psi\right)^{2}}+\int_{\tilde{\mathcal{D}}_{\tau_{1}}^{\tau_{2}}}{\frac{1}{4Dr^{2}}\partial_{v}(r\psi^{2})}.
\end{split}
\end{equation*}
However, if $\zeta$ is the cut-off function introduced in the proof of Proposition \ref{rweigestprop}, then
\begin{equation*}
\begin{split}
\int_{\tilde{\mathcal{D}}_{\tau_{1}}^{\tau_{2}}}{\frac{1}{4Dr^{2}}\partial_{v}\left(r(\zeta\psi)^{2}\right)}=
\int_{\mathcal{I}^{+}}{\frac{1}{8r}(\zeta\psi)^{2}}.
\end{split}
\end{equation*}
Therefore, the above integral is of the right sign modulo some error terms in the region $\mathcal{W}$ coming from the cut-off $\zeta$. These terms are quadratic in the 1-jet of $\psi$ and so can be controlled by the T-flux. Finally, since $n^{\mu}_{\tilde{N}_{\tau}}$ is null we have $J^{T}_{\mu}[\psi]n^{\mu}_{\tilde{N}_{\tau}}\sim \left(\partial_{v}\psi\right)^{2}+\left|\nabb\psi\right|^{2}$ and thus by \eqref{corp-1} and the coarea formula we have the required result.
\end{proof}

This is a spacetime estimate which does \textbf{not} degenerate at infinity. Note the importance of the fact that the region $\tilde{\mathcal{D}}_{\tau_{1}}^{\tau_{2}}$ does not contain $i^{0}$! If we are to obtain the full decay for the energy, then we need to prove decay for the boundary terms in Proposition \ref{nondegx}. The first step is to derive a spacetime estimate of the $r$-weighted quantity $r^{-1}(\partial_{v}\phi)^{2}$.
\begin{proposition}
There exists a constant $C$ which depends  on $M$  and $\tilde{\Sigma}_{0}$  such that
\begin{equation*}
\begin{split}
\int_{\tau_{1}}^{\tau_{2}}{\left(\int_{\tilde{N}_{\tau}}{r^{-1}(\partial_{v}\phi)^{2}}\right)d\tau}\ \leq C\int_{\tilde{\Sigma}_{\tau_{1}}}{J_{\mu}^{T}\left[\psi\right]n^{\mu}_{\tilde{\Sigma}_{\tau_{1}}}}+C\int_{\tilde{N}_{\tau_{1}}}{\left(\partial_{v}\phi\right)^{2}}.
\end{split}
\end{equation*}
\label{rwe1}
\end{proposition}
\begin{proof}
Appying the $r$-weighted energy estimate for $p=2$ we obtain
\begin{equation*}
\begin{split}
\int_{\tilde{\mathcal{D}}_{\tau_{1}}^{\tau_{2}}}{r^{-1}(\partial_{v}\phi)^{2}}\, &\leq\!\!\int_{\tilde{\mathcal{D}}_{\tau_{1}}^{\tau_{2}}}{\frac{M\sqrt{D}}{4r^{2}}\left|\nabb\phi\right|^{2}}+C\int_{\tilde{\Sigma}_{\tau_{1}}}{J_{\mu}^{T}\left[\psi\right]n^{\mu}_{\tilde{\Sigma}_{\tau_{1}}}}\!+\! C\!\!\int_{\tilde{N}_{\tau_{1}}}{\left(\partial_{v}\phi\right)^{2}}\\
&\leq C\int_{\tilde{\mathcal{D}}_{\tau_{1}}^{\tau_{2}}}{\frac{1}{r^{2}}\left|\nabb\phi\right|^{2}}+C\int_{\tilde{\Sigma}_{\tau_{1}}}{J_{\mu}^{T}\left[\psi\right]n^{\mu}_{\tilde{\Sigma}_{\tau_{1}}}}\!+C\!\int_{\tilde{N}_{\tau_{1}}}{\left(\partial_{v}\phi\right)^{2}}.
\end{split}
\end{equation*}
The result now follows from \eqref{corp-1} and the coarea formula.
\end{proof}

\subsection{Integrated Decay of Local (Higher Order) Energy}
\label{sec:DecayOfLocalEnergy}

We have shown in \cite{aretakis1} that in order to obtain a non-degenerate spacetime estimate near $\mathcal{H}^{+}$ we need to commute the wave equation with the transversal to the horizon vector field $\partial_{r}$ and assume that the zeroth spherical harmonic vanishes. Indeed, if $\mathcal{A}$ is a spatially compact neighbourhood of $\mathcal{H}^{+}$ (which may contain the photon sphere) then we have:
\begin{proposition}
There exists a constant $C$ that depends  on $M$ and $\tilde{\Sigma}_{0}$   such that  if $\psi$ satisfies the wave equation and is supported on $l\geq 1$, then 
\begin{equation*}
\begin{split}
\int_{\tau_{1}}^{\tau_{2}}{\left(\int_{\mathcal{A}\cap\tilde{\Sigma}_{\tau}}{J^{N}_{\mu}[\psi]n^{\mu}_{\tilde{\Sigma}_{\tau}}}\right)d\tau}\ \leq  \  &C\int_{\tilde{\Sigma}_{\tau_{1}}}{J^{N}_{\mu}[\psi]n^{\mu}_{\tilde{\Sigma}_{\tau}}}+C\int_{\tilde{\Sigma}_{\tau_{1}}}{J^{N}_{\mu}[T\psi]n^{\mu}_{\tilde{\Sigma}_{\tau}}}+C\int_{\mathcal{A}\cap\tilde{\Sigma}_{\tau_{1}}}{J^{N}_{\mu}[\partial_{r}\psi]n^{\mu}_{\tilde{\Sigma}_{\tau}}}.
\end{split}
\end{equation*}
\label{inte1}
\end{proposition}
Regarding the above boundary terms we have
\begin{proposition}
There exists a constant $C$ that depends  on $M$ and $\tilde{\Sigma}_{0}$   such that  if $\psi$ satisfies the wave equation and is supported on $l\geq 2$, then 
\begin{equation*}
\int_{\tau_{1}}^{\tau_{2}}{\left(\int_{\mathcal{A}\cap\tilde{\Sigma}_{\tau}}{J^{N}_{\mu}[\partial_{r}\psi]n^{\mu}_{\tilde{\Sigma}_{\tau}}}\right)d\tau}\ \leq
C\sum_{i=0}^{2}{\int_{\tilde{\Sigma}_{\tau_{1}}}{J^{N}_{\mu}[T^{i}\psi]n^{\mu}_{\tilde{\Sigma}_{\tau}}}}+C\sum_{k=1}^{2}{\int_{\mathcal{A}\cap\tilde{\Sigma}_{\tau_{1}}}{J^{N}_{\mu}[\partial_{r}^{k}\psi]n^{\mu}_{\tilde{\Sigma}_{\tau}}}}.
\end{equation*}
\label{inte2}
\end{proposition}
\begin{proof}
Immediate from Theorem \ref{theorem3} of Section \ref{sec:TheMainTheorems} and the coarea formula.
\end{proof}

\subsection{Weighted Energy Estimates  in a Neighbourhood of $\mathcal{H}^{+}$}
\label{sec:RWeightedEnergyEstimatesH}
Since for $l=0$ the above estimates do not hold for generic initial data, we are left proving decay for the degenerate energy. For this we derive a hierarchy of (degenerate) energy estimates in a neighbourhood of $\mathcal{H}^{+}$, the crucial ingredient of which is the existence of the vector field $P$. This vector field is  timelike in the domain of outer communications and becomes null on the horizon  ``linearly". This linearity allows $P$ to capture  the degenerate redshift in $\mathcal{A}$  in a weaker way than $N$ but in stronger way than $T$.
 In this subsection, we use the $(v,r)$ coordinates.

\begin{proposition}
There exists a $\varphi_{\tau}^{T}$-invariant causal vector field $P$ and a constant $C$ which depends only on $M$ such that  for all $\psi$ we have 
\begin{equation*}
\begin{split}
& J_{\mu}^{T}[\psi]n^{\mu}_{\Sigma}\leq  C K^{P}[\psi],\ \  J_{\mu}^{P}[\psi]n^{\mu}_{\Sigma}\leq C K^{N,\delta,-\frac{1}{2}}[\psi]
\end{split}
\end{equation*}
in an appropriate neighbourhood $\mathcal{A}$ of $\mathcal{H}^{+}$.
\label{rstarHenergy}
\end{proposition}
\begin{proof}
Let our ansatz be $P=P^{v}T+P^{r}\partial_{r}$. Recall that 
\begin{equation*}
\begin{split}
K^{P}\left[\psi\right]=&F_{vv}\left(T\psi\right)^{2}+F_{rr}\left(\partial_{r}\psi\right)^{2}+F_{\scriptsize\nabb}\left|\nabb\psi\right|^{2}+F_{vr}\left(T\psi\right)\left(\partial_{r}\psi\right),\\
\end{split}
\end{equation*}
where the coef{}ficients are given by
\begin{equation*}
\begin{split}
&F_{vv}=\left(\partial_{r}P^{v}\right),\ F_{rr}=D\left[\frac{\left(\partial_{r}P^{r}\right)}{2}-\frac{P^{r}}{r}\right]-\frac{P^{r}D'}{2},  \ F_{\scriptsize\nabb}=-\frac{1}{2}\left(\partial_{r}P^{r}\right),\ F_{vr}=D\left(\partial_{r}P^{v}\right)-\frac{2P^{r}}{r}.\\
\end{split}
\end{equation*}
Let us take $P^{r}(r)=-\sqrt{D}$ for $M\leq r\leq r_{0}<2M$ with $r_{0}$ to be determined later. Then
\begin{equation*}
\begin{split}
F_{rr}&=D\left[-\frac{D'}{4\sqrt{D}}+\frac{\sqrt{D}}{r}\right]+\frac{\sqrt{D}D'}{2}=D\left[\frac{D'}{4\sqrt{D}}+\frac{\sqrt{D}}{r}\right]\sim D.
\end{split}
\end{equation*}
Since $\frac{D'}{4\sqrt{D}}=\frac{M}{2r^{2}}$, the constants in $\sim$ depend on $M$ and the choice for $r_{0}$. Also,
\begin{equation*}
\begin{split}
F_{vr}=\sqrt{D}\left[\sqrt{D}(\partial_{r}P^{v})+\frac{2}{r}\right]\leq \epsilon D+\frac{1}{\epsilon}\left[\sqrt{D}(\partial_{r}P^{v})+\frac{2}{r}\right]^{2}.
\end{split}
\end{equation*}
If we take $\epsilon$ sufficiently small  and $P^{v}$ such that  $\frac{1}{\epsilon}\left[\sqrt{D}(\partial_{r}P^{v})+\frac{2}{r}\right]^{2} < \partial_{r}P^{v}$ (note that this is always possible in view of the degeneracy of $\sqrt{D}$ at $\mathcal{H}^{+}$),
then  there exists $r_{0}>M$ such that
\begin{equation}
\begin{split}
K^{P}[\psi]\sim \left((T\psi)^{2}+D(\partial_{r}\psi)^{2}+\left|\nabb\psi\right|^{2}\right)\sim J_{\mu}^{T}[\psi]n^{\mu}_{\Sigma}
\label{p1} 
\end{split}
\end{equation}
 in $\mathcal{A}=\left\{M\leq r\leq r_{0}\right\}$. Extend now $P$ in $\mathcal{R}$ such that $P^{v}(r)=1$ and $P^{r}(r)=0$ for all $r\geq r_{1}>r_{0}$ for some $r_{1}<2M$. This proves the first part of the proposition.  In region $\mathcal{A}$ we have $-g(P,P)\sim\sqrt{D}$ and so 
\begin{equation}
\begin{split}
J_{\mu}^{P}[\psi]n^{\mu}_{\Sigma}&
\sim (T\psi)^{2}+\sqrt{D}(\partial_{r}\psi)^{2}+\left|\nabb\psi\right|^{2}\sim K^{N,\delta,-\frac{1}{2}}[\psi].
\label{p2}
\end{split}
\end{equation}
\end{proof}

\subsection{Decay of Degenerate Energy}
\label{sec:DecayOfDegenerateEnergy}

\subsubsection{Uniform Boundedness of $P$-Energy}
\label{sec:UniformBoundenessOfPEnergy}

First we need to prove that the $P$-flux is uniformly bounded.

\begin{proposition}
There exists a constant $C$ that depends  on $M$  and $\tilde{\Sigma}_{0}$  such that for all solutions $\psi$ of the wave equation we have
\begin{equation}
\int_{\tilde{\Sigma}_{\tau}}{J_{\mu}^{P}[\psi]n^{\mu}_{\tilde{\Sigma}_{\tau}}}\leq C\int_{\tilde{\Sigma}_{0}}{J_{\mu}^{P}[\psi]n^{\mu}_{\tilde{\Sigma}_{0}}}.
\label{pboun}
\end{equation}
\label{pbound}
\end{proposition}
\begin{proof}
Stokes' theorem for the current $J_{\mu}^{P}$ gives us
\begin{equation*}
\int_{\tilde{\Sigma}_{\tau}}{J_{\mu}^{P}n^{\mu}}+\int_{\mathcal{H}^{+}}{J_{\mu}^{P}n^{\mu}}+\int_{\mathcal{I}^{+}}{J_{\mu}^{P}n^{\mu}}+\int_{\tilde{\mathcal{R}}}{K^{P}}=\int_{\tilde{\Sigma}_{0}}{J_{\mu}^{P}n^{\mu}}.
\end{equation*}
Note that since $P$ is a future-directed causal vector field, the boundary integrals over $\mathcal{H}^{+}$ and $\mathcal{I}^{+}$ are non-negative. The same also holds for $K^{P}$ in region $\mathcal{A}$ whereas it vanishes away from the horizon.  In the intermediate region this spacetime integral can be bounded using the degenerate  $X$  estimate of Theorem 1 of \cite{aretakis1}. The result now follows from $J_{\mu}^{T}n^{\mu} \leq CJ_{\mu}^{P}n^{\mu}.$
\end{proof}
We are now in a position to derive local integrated decay for the $T$-energy.
\begin{proposition}
There exists a constant $C$ that depends  on $M$ and $\tilde{\Sigma}_{0}$   such that for all solutions $\psi$ of the wave equation   we have
\begin{equation*}
\int_{\tau_{1}}^{\tau_{2}}{\left(\int_{\mathcal{A}\cap\tilde{\Sigma}_{\tau}}{J_{\mu}^{T}[\psi]n^{\mu}_{\tilde{\Sigma}_{\tau}}}\right)d\tau}\leq C\int_{\tilde{\Sigma}_{\tau_{1}}}{J_{\mu}^{P}[\psi]n^{\mu}_{\tilde{\Sigma}_{\tau_{1}}}}
\end{equation*}
and
\begin{equation*}
\int_{\tau_{1}}^{\tau_{2}}{\left(\int_{\mathcal{A}\cap\tilde{\Sigma}_{\tau}}{J_{\mu}^{P}[\psi]n^{\mu}_{\tilde{\Sigma}_{\tau}}}\right)d\tau}\leq C\int_{\tilde{\Sigma}_{\tau_{1}}}{J_{\mu}^{N}[\psi]n^{\mu}_{\tilde{\Sigma}_{\tau_{1}}}}
\end{equation*}
in an appropriate $\varphi_{\tau}$-invariant neighbourhood $\mathcal{A}$ of $\mathcal{H}^{+}$.
\label{propp1}
\end{proposition}
\begin{proof}
From the divergence identity for the current $J_{\mu}^{P}$ and the boundedness of $P$-energy we have
\begin{equation*}
\int_{\mathcal{A}}{K^{P}}\leq C\int_{\tilde{\Sigma}_{\tau_{1}}}{J_{\mu}^{P}[\psi]n^{\mu}_{\tilde{\Sigma}_{\tau_{1}}}}
\end{equation*}
for a uniform constant $C$. Thus the first estimate follows from \eqref{p1} and the coarea formula. Likewise, the second estimate follows from the divergence identity for the current $J_{\mu}^{N,\delta, -\frac{1}{2}}$, the boundedness of the non-degenerate $N$-energy and \eqref{p2}.
\end{proof}

\subsubsection{The Dyadic Sequence $\rho_{n}$}
\label{sec:TheDyadicSequenceTildeTauN}

In view of Propositions  \ref{nondegx} and \ref{propp1} we have
\begin{equation}
\begin{split}
\int_{\tau_{1}}^{\tau_{2}}{\left(\int_{\tilde{\Sigma}_{\tau}}{J^{T}_{\mu}[\psi]n^{\mu}_{\tilde{\Sigma}_{\tau}}}\right)d\tau}\, \leq\,  CI^{T}_{\tilde{\Sigma}_{\tau_{1}}}[\psi],
\label{tinte1}
\end{split}
\end{equation}
where 
\begin{equation*}
\begin{split}
I^{T}_{\tilde{\Sigma}_{\tau}}[\psi]=&\int_{\tilde{\Sigma}_{\tau}}{J^{P}_{\mu}[\psi]n^{\mu}_{\tilde{\Sigma}_{\tau}}}+
\int_{\tilde{\Sigma}_{\tau}}{J^{T}_{\mu}[T\psi]n^{\mu}_{\tilde{\Sigma}_{\tau}}}+\int_{\tilde{N}_{\tau}}{r^{-1}\left(\partial_{v}\phi\right)^{2}}.
\end{split}
\end{equation*}
Moreover, from Propositions \ref{propp1} and \ref{rwe1} we have
\begin{equation}
\int_{\tau_{1}}^{\tau_{2}}{I_{\tilde{\Sigma}_{\tau}}^{T}[\psi]d\tau}\leq CI_{\tilde{\Sigma}_{\tau_{1}}}^{T}[T\psi]+C\int_{\tilde{\Sigma}_{\tau_{1}}}{J_{\mu}^{N}[\psi]n^{\mu}}+C\int_{\tilde{N}_{\tau_{1}}}{(\partial_{v}\phi)^{2}},
\label{tinte2}
\end{equation}
for a  constant $C$ that  depends  on  $M$  and $\tilde{\Sigma}_{0}$. This implies that there exists a dyadic sequence\footnote{Dyadic sequence is an increasing sequence $\rho_{n}$ such that $\rho_{n}\sim\rho_{n+1}\sim(\rho_{n+1}-\rho_{n}).$} $\rho_{n}$ such that 
\begin{equation*}
I^{T}_{\tilde{\Sigma}_{\rho_{n}}}[\psi]\leq \frac{E_{1}}{\rho_{n}},
\end{equation*}
where $E_{1}$ is equal to the right hand side of \eqref{tinte2} (with $\tau_{1}=0$) and depends only on the initial data of $\psi$. We have now all the tools to derive decay for the degenerate energy.
\begin{proposition}
There exists a constant $C$ that depends on $M$   and $\tilde{\Sigma}_{0}$ such that for all solutions $\psi$ of the wave equation we have 
\begin{equation*}
\int_{\tilde{\Sigma}_{\tau}}J^{T}_{\mu}[\psi]n_{\tilde{\Sigma}_{\tau}}^{\mu}\leq CE_{1}\frac{1}{\tau^{2}},
\end{equation*}
where $E_{1}$ is as defined above.
\label{tdecay}
\end{proposition}
\begin{proof}
We apply \eqref{tinte1} for the dyadic interval $[\rho_{n},\rho_{n+1}]$ to obtain
\begin{equation*}
\int_{\rho_{n}}^{\rho_{n+1}}{\left(\int_{\tilde{\Sigma}_{\tau}}{J^{T}_{\mu}[\psi]n^{\mu}_{\tilde{\Sigma}_{\tau}}}\right)d\tau}\, \leq\,  CE_{1}\frac{1}{\rho_{n}}.
\end{equation*}
In view of the energy estimate
\begin{equation*}
\begin{split}
\int_{\tilde{\Sigma}_{\tau}}{J^{T}_{\mu}[\psi]n^{\mu}_{\tilde{\Sigma}_{\tau}}}\ \leq\, C\int_{\tilde{\Sigma}_{\tau'}}{J^{T}_{\mu}[\psi]n^{\mu}_{\tilde{\Sigma}_{\tau'}}},
\end{split}
\end{equation*}
which holds for all $\tau\geq \tau'$, we have
\begin{equation*}
\begin{split}
(\rho_{n+1}-\rho_{n})\int_{\tilde{\Sigma}_{\rho_{n+1}}}{J^{T}_{\mu}[\psi]n^{\mu}_{\tilde{\Sigma}_{\rho_{n+1}}}}\,\leq\, CE_{1}\frac{1}{\rho_{n+1}}.
\end{split}
\end{equation*}
Since there exists a uniform constant $b>0$ such that $b\tau_{n+1}\leq\tau_{n+1}-\tau_{n}$ we have 
\begin{equation*}
\begin{split}
\int_{\tilde{\Sigma}_{\rho_{n+1}}}{J^{T}_{\mu}[\psi]n^{\mu}_{\tilde{\Sigma}_{\rho_{n+1}}}}\,\leq\, CE_{1}\frac{1}{\rho_{n+1}^{2}}.
\end{split}
\end{equation*}
Now, for $\tau\geq\rho_{1}$ there exists $n\in\mathbb{N}$ such that $\rho_{n}\leq\tau\leq\rho_{n+1}$. Therefore,
\begin{equation*}
\begin{split}
\int_{\tilde{\Sigma}_{\tau}}{J^{T}_{\mu}[\psi]n^{\mu}_{\tilde{\Sigma}_{\tau}}}\,\leq\, C\int_{\tilde{\Sigma}_{\rho_{n+1}}}{J^{T}_{\mu}[\psi]n^{\mu}_{\tilde{\Sigma}_{\rho_{n+1}}}}\, \leq\frac{CE_{1}}{\rho_{n}^{2}}\sim\frac{CE_{1}}{\rho_{n+1}^{2}}\leq CE_{1}\frac{1}{\tau^{2}},
\end{split}
\end{equation*}
which is the required decay result for the $T$-energy. 

\end{proof}

\subsection{Decay of Non-Degenerate Energy}
\label{sec:DecayOfNonDegenerateEnergy}

We  now derive decay for the non-degenerate energy. Note that for obtaining such a result we must use Proposition \ref{inte1} which however holds for solutions to the wave equation supported on the frequencies $l\geq 1$. In this case, in view of the previous estimates we have
\begin{equation}
\begin{split}
\int_{\tau_{1}}^{\tau_{2}}{\left(\int_{\tilde{\Sigma}_{\tau}}{J^{N}_{\mu}[\psi]n^{\mu}_{\tilde{\Sigma}_{\tau}}}\right)d\tau}\, \leq\, & CI^{N}_{\tilde{\Sigma}_{\tau_{1}}}[\psi],
\end{split}
\label{ninte1}
\end{equation}
where 
\begin{equation*}
\begin{split}
I^{N}_{\tilde{\Sigma}_{\tau}}[\psi]=&\int_{\tilde{\Sigma}_{\tau}}{J^{N}_{\mu}[\psi]n^{\mu}_{\tilde{\Sigma}_{\tau}}}+
\int_{\tilde{\Sigma}_{\tau}}{J^{N}_{\mu}[T\psi]n^{\mu}_{\tilde{\Sigma}_{\tau}}}+\int_{\mathcal{A}\cap\tilde{\Sigma}_{\tau}}{J^{N}_{\mu}[\partial_{r}\psi]n^{\mu}_{\tilde{\Sigma}_{\tau}}}+\int_{\tilde{N}_{\tau}}{r^{-1}\left(\partial_{v}\phi\right)^{2}}.
\end{split}
\end{equation*}
\begin{proposition}
There exists a constant $C$ that depends  on $M$   and $\tilde{\Sigma}_{0}$ such that for all solutions $\psi$ to the wave equation which are supported on $l\geq 1$ we have
\begin{equation*}
\begin{split}
\int_{\tilde{\Sigma}_{\tau}}{J^{N}_{\mu}[\psi]n^{\mu}_{\tilde{\Sigma}_{\tau}}}\,\leq\, CE_{2}\frac{1}{\tau},
\end{split}
\end{equation*}
where $E_{2}$ depends only on the initial data of $\psi$ and is equal to the right hand side of \eqref{ninte1} (with $\tau_{1}=0$).
\label{l1decay}
\end{proposition}
\begin{proof}
We apply \eqref{ninte1} for the interval $[0,\tau]$ and use the energy estimate
\begin{equation*}
\begin{split}
\int_{\tilde{\Sigma}_{\tau}}{J^{N}_{\mu}[\psi]n^{\mu}_{\tilde{\Sigma}_{\tau}}}\ \leq\, C\int_{\tilde{\Sigma}_{\tau'}}{J^{N}_{\mu}[\psi]n^{\mu}_{\tilde{\Sigma}_{\tau'}}},
\end{split}
\end{equation*}
which holds for all $\tau\geq \tau'$ and for a uniform constant $C$.
\end{proof}

\subsubsection{The Dyadic Sequence $\tau_{n}$}
\label{sec:APriviledgedDyadicSequence}

If we  consider solutions $\psi$ which are supported on $l\geq 2$ then from Propositions  \ref{rwe1}, \ref{inte2} and by commuting \eqref{ninte1} with $T$ we take
\begin{equation}
\begin{split}
\int_{\tau_{1}}^{\tau_{2}}{I^{N}_{\tilde{\Sigma}_{\tau}}[\psi]d\tau}\,\leq \, &CI^{N}_{\tilde{\Sigma}_{\tau_{1}}}[\psi]+CI^{N}_{\tilde{\Sigma}_{\tau_{1}}}[T\psi]+C\int_{\mathcal{A}\cap\tilde{\Sigma}_{\tau_{1}}}{J^{N}_{\mu}[\partial_{r}\partial_{r}\psi]n^{\mu}_{\tilde{\Sigma}_{\tau}}}+C\int_{\tilde{N}_{\tau_{1}}}{(\partial_{v}\phi)^{2}}
\end{split}
\label{ninte2}
\end{equation}
This implies that there exists a dyadic sequence $\tau_{n}$ such that
\begin{equation*}
\begin{split}
I^{N}_{\tilde{\Sigma}_{\tau_{n}}}[\psi] \, \leq\, \frac{E_{3}}{\tau_{n}},
\end{split}
\end{equation*}
where the constant $E_{3}$ is equal to the right hand side of \eqref{ninte2} (with $\tau_{1}=0)$. We can now derive decay for the non-degenerate energy.
\begin{proposition}
There exists a constant $C$ that depends  on $M$   and $\tilde{\Sigma}_{0}$ such that for all solutions $\psi$ to the wave equation which are supported on $l\geq 2$ we have
\begin{equation*}
\begin{split}
\int_{\tilde{\Sigma}_{\tau}}{J^{N}_{\mu}[\psi]n^{\mu}_{\tilde{\Sigma}_{\tau}}}\,\leq\, CE_{3}\frac{1}{\tau^{2}},
\end{split}
\end{equation*}
where $E_{3}$ is as defined above.
\label{energydecay}
\end{proposition}
\begin{proof}
If we apply \eqref{ninte1} for the dyadic intervals $[\tau_{n},\tau_{n+1}]$ we obtain
\begin{equation*}
\begin{split}
\int_{\tau_{n}}^{\tau_{n+1}}{\left(\int_{\tilde{\Sigma}_{\tau}}{J^{N}_{\mu}[\psi]n^{\mu}_{\tilde{\Sigma}_{\tau}}}\right)d\tau}\, \leq\, \frac{CE_{3}}{\tau_{n}}.
\end{split}
\end{equation*}
In view of the boundedness of the $N$-energy  we have
\begin{equation*}
\begin{split}
(\tau_{n+1}-\tau_{n})\int_{\tilde{\Sigma}_{\tau_{n+1}}}{J^{N}_{\mu}[\psi]n^{\mu}_{\tilde{\Sigma}_{\tau_{n+1}}}}\,\leq\, \frac{CE_{3}}{\tau_{n+1}}.
\end{split}
\end{equation*}
Since there exists a uniform constant $b>0$ such that $b\tau_{n+1}\leq\tau_{n+1}-\tau_{n}$ we obtain
\begin{equation*}
\begin{split}
\int_{\tilde{\Sigma}_{\tau_{n+1}}}{J^{N}_{\mu}[\psi]n^{\mu}_{\tilde{\Sigma}_{\tau_{n+1}}}}\,\leq\, \frac{CE_{3}}{\tau_{n+1}^{2}}.
\end{split}
\end{equation*}
Now, for $\tau\geq\tau_{1}$ there exists $n\in\mathbb{N}$ such that $\tau_{n}\leq\tau\leq\tau_{n+1}$. Therefore,
\begin{equation*}
\begin{split}
\int_{\tilde{\Sigma}_{\tau}}{J^{N}_{\mu}[\psi]n^{\mu}_{\tilde{\Sigma}_{\tau}}}\,\leq\, \tilde{C}\int_{\tilde{\Sigma}_{\tau_{n+1}}}{J^{N}_{\mu}[\psi]n^{\mu}_{\tilde{\Sigma}_{\tau_{n+1}}}}\, \leq\frac{CE_{3}}{\tau_{n}^{2}}\sim\frac{CE_{3}}{\tau_{n+1}^{2}}\leq CE_{3}\frac{1}{\tau^{2}}
\end{split}
\end{equation*}
which is the required decay result for the energy. 

\end{proof}
The above propositions completes the proof of Theorem \ref{t4}.

\section{Pointwise Estimates}
\label{sec:PointwiseEstimates}

\subsection{Retrieving Pointwise Boundedness}
\label{sec:UniformPointwiseBoundedness}

In \cite{aretakis1} we proved that all the solutions to the wave equation $\psi$ remain uniformly bounded in $\mathcal{M}$. We show the same result here by exploiting the spherical symmetry. We work with the foliation $\Sigma_{\tau}$ (or $\tilde{\Sigma}_{\tau}$) and  the induced coordinate system $(\rho,\omega)$. For $r_{0}\geq M$   we have 
\begin{equation*}
\begin{split}
\psi^{2}\left(r_{0},\omega\right)&=\left(\int_{r_{0}}^{+\infty}{\left(\partial_{\rho}\psi\right)d\rho}\right)^{2}\leq\left(\int_{r_{0}}^{+\infty}{\left(\partial_{\rho}\psi\right)^{2}\rho^{2}d\rho}\right)\left(\int_{r_{0}}^{+\infty}{\frac{1}{\rho^{2}}d\rho}\right)
=\frac{1}{r_{0}}\left(\int_{r_{0}}^{+\infty}{\left(\partial_{\rho}\psi\right)^{2}\rho^{2}d\rho}\right).
\end{split}
\end{equation*}
Therefore,
\begin{equation}
\begin{split}
\int_{\mathbb{S}^{2}}{\psi^{2}(r_{0},\omega)d\omega}&\leq\frac{1}{r_{0}}\int_{\mathbb{S}^{2}}{\int_{r_{0}}^{+\infty}{\left(\partial_{\rho}\psi\right)^{2}\rho^{2}d\rho d\omega}}\leq \frac{C}{r_{0}}\int_{\Sigma_{\tau}\cap\left\{r\geq r_{0}\right\}}{J_{\mu}^{N}[\psi]n^{\mu}_{\Sigma_{\tau}}},
\end{split}
\label{1pointwise}
\end{equation}
where $C$ is a constant that depends only on $M$ and $\Sigma_{0}$. 
\begin{theorem}
There exists a constant $C$ which depends on $M$ and $\Sigma_{0}$ such that for all solutions $\psi$ of the wave equation we have
\begin{equation}
\begin{split}
\left|\psi\right|^{2}\leq C\cdot E_{4}\frac{1}{r},
\end{split}
\label{2pointwise}
\end{equation}
where $E_{4}=\sum_{\left|k\right|\leq 2}{\int_{\Sigma_{0}}{J_{\mu}^{N}[\Omega^{k}\psi]n^{\mu}_{\Sigma_{0}}}}.$
\label{unponpsi}
\end{theorem}
\begin{proof}
From the Sobolev inequality on $\mathbb{S}^{2}$ we have $\left|\psi\right|^{2} \leq C\sum_{\left|k\right|\leq 2}{\int_{\mathbb{S}^{2}}{\left(\Omega^{k}\psi\right)^{2}}}$
and the theorem follows from \eqref{1pointwise} and the uniform boundedness of the non-degenerate energy.
\end{proof}

\subsection{Pointwise Decay}
\label{sec:PointwiseDecay}

\subsubsection{Decay away from $\mathcal{H}^{+}$}
\label{sec:DecayAwayMathcalH}

We consider the region $\left\{r\geq R_{1}\right\}$, where $R_{1}>M$. From now on, $C$ will be a constant depending only on $M$, $R_{1}$ and $\tilde{\Sigma}_{0}$.
\begin{figure}[H]
	\centering
		\includegraphics[scale=0.1]{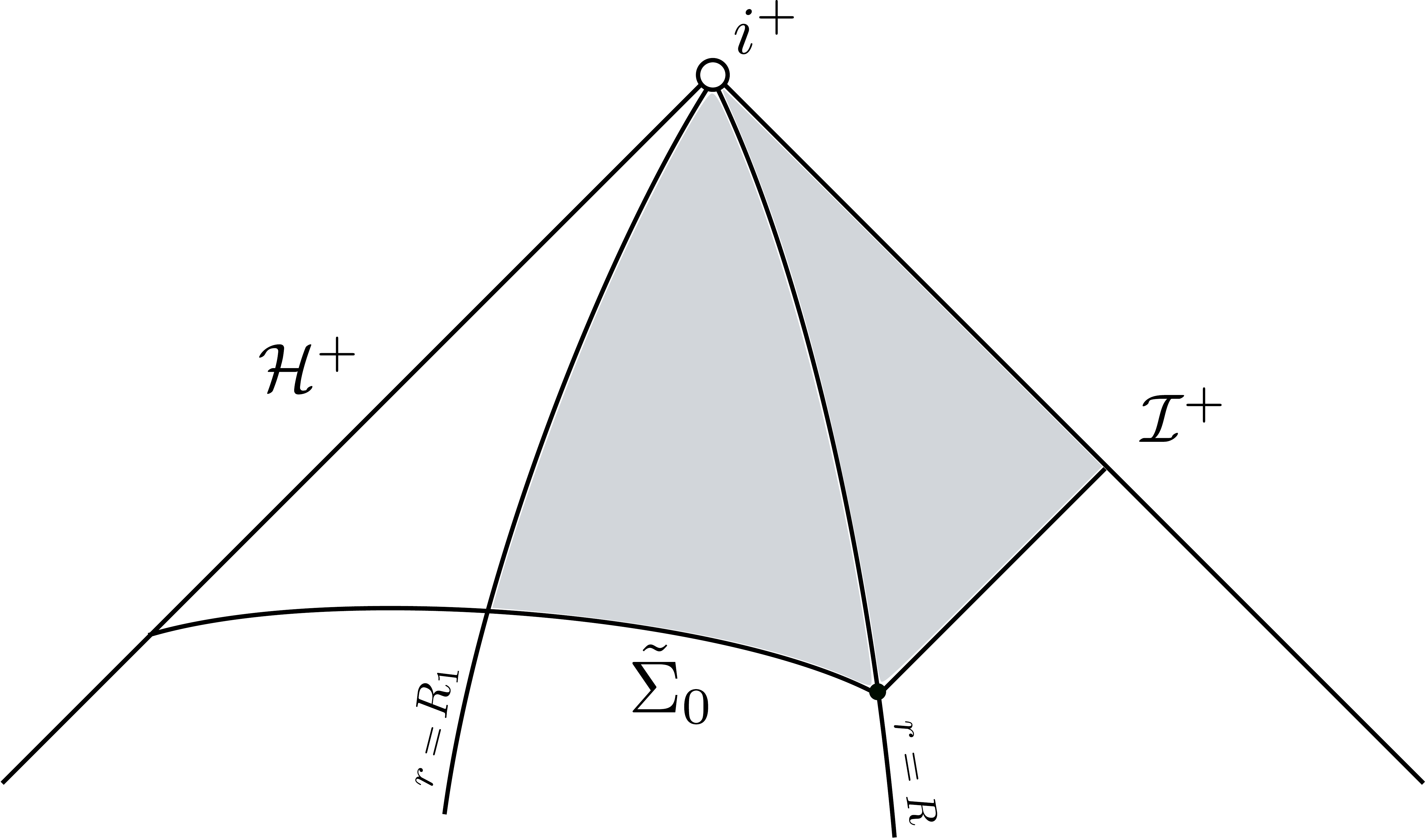}
	\label{fig:1pd}
\end{figure}
Clearly, in this region we have $J_{\mu}^{N}n^{\mu}_{\Sigma}\sim  J_{\mu}^{T}n_{\Sigma}^{\mu}$ and $\sim$ depends on $R_{1}$. Therefore, from \eqref{1pointwise} we have that for any $r\geq R_{1}$  
\begin{equation*}
\int_{\mathbb{S}^{2}}{\psi^{2}(r,\omega)d\omega}\leq \frac{C}{r}\int_{\tilde{\Sigma}_{\tau}}{J_{\mu}^{T}[\psi]n^{\mu}_{\tilde{\Sigma}_{\tau}}}\leq C\cdot E_{1}[\psi]\frac{1}{r\cdot \tau^{2}}.
\end{equation*}
Applying the above for $\psi$, $\Omega_{i}\psi$, $\Omega_{ij}\psi$ and using the  Sobolev inequality for $\mathbb{S}^{2}$ yields
\begin{equation*}
\psi^{2}\leq C E_{5}\frac{1}{r\cdot \tau^{2}},
\end{equation*}
where  $E_{5}=\sum_{\left|k\right|\leq 2}{E_{1}\left[\Omega^{k}\psi\right]}.$
Next we improve the decay with respect to $r$. Observe that for all $r\geq R_{1}$ we have
\begin{equation*}
\begin{split}
\int_{\mathbb{S}^{2}}{(r\psi)^{2}(r,\omega)d\omega}&=\int_{\mathbb{S}^{2}}{(R_{1}\psi)^{2}(R_{1},\omega)d\omega}+2\int_{\mathbb{S}^{2}}\int_{R_{1}}^{r}{\frac{\psi}{\rho}\partial_{\rho}(\rho\psi)\rho^{2}d\rho d\omega}\\
&\leq CE_{1}[\psi]\frac{1}{\tau^{2}}+C\sqrt{\int_{\tilde{\Sigma}_{\tau}\cap\left\{r\geq R_{1}\right\}}\frac{1}{\rho^{2}}\psi^{2}\int_{\tilde{\Sigma}_{\tau}\cap\left\{r\geq R_{1}\right\}}{\left(\partial_{\rho}(\rho \psi)\right)^{2}}}.
\end{split}
\end{equation*}
However, from  the first Hardy inequality (see Section 6 of \cite{aretakis1}) we have
\begin{equation*}
\int_{\tilde{\Sigma}_{\tau}}{\frac{1}{\rho^{2}}\psi^{2}}\leq C\int_{\tilde{\Sigma}_{\tau}}{J_{\mu}^{T}[\psi]n^{\mu}_{\tilde{\Sigma}_{\tau}}}\leq CE_{1}\frac{1}{\tau^{2}}.
\end{equation*}
Moreover, if $R_{0}$ is the constant defined in Section \ref{sec:EnergyDecay} and  recalling that $\rho\psi=\phi$ we have
\begin{equation*}
\begin{split}
\int_{\tilde{\Sigma}_{\tau}\cap\left\{r\geq R_{1}\right\}}{\left(\partial_{\rho}(\rho \psi)\right)^{2}}&=\int_{\tilde{\Sigma}_{\tau}\cap\left\{R_{0}\geq r\geq R_{1}\right\}}{\left(\partial_{\rho}(\rho \psi)\right)^{2}}+\int_{\tilde{N}_{\tau}}{(\partial_{v}\phi)^{2}}\\
&\leq C\int_{\tilde{\Sigma}_{0}}{J_{\mu}^{T}[\psi]n_{\tilde{\Sigma}_{0}}^{\mu}}+\int_{\tilde{N}_{\tau}}{(\partial_{v}\phi)^{2}}\leq  C\int_{\tilde{\Sigma}_{0}}{J_{\mu}^{T}[\psi]n_{\tilde{\Sigma}_{0}}^{\mu}}+\int_{\tilde{N}_{0}}{(\partial_{v}\phi)^{2}},
\end{split}
\end{equation*}
where for the second inequality we used Propositions \ref{rweigestprop} and \ref{nondegx}. Hence for $\tau\geq 1$ we have
\begin{equation*}
\begin{split}
r^{2}\int_{\mathbb{S}^{2}}{\psi^{2}(r,\omega)d\omega}&\leq CE_{1}\frac{1}{\tau^{2}}+C\sqrt{E_{1}}\sqrt{C\int_{\tilde{\Sigma}_{0}}{J_{\mu}^{T}[\psi]n_{\tilde{\Sigma}_{0}}^{\mu}}+\int_{\tilde{N}_{0}}{(\partial_{v}\phi)^{2}}}\frac{1}{\tau}\leq C E_{1}\frac{1}{\tau},
\end{split}
\end{equation*}
since the quantitiy in the square root is dominated by $E_{1}$. Therefore, by the Sobolev inequality on $\mathbb{S}^{2}$ we obtain
\begin{equation*}
\begin{split}
\psi^{2}\leq CE_{5}\frac{1}{r^{2}\cdot\tau}.
\end{split}
\end{equation*}

\subsubsection{Decay near $\mathcal{H}^{+}$}
\label{sec:DecayNearMathcalH}

We are now investigating the behaviour of $\psi$ in the region $\left\{M\leq r\leq R_{1}\right\}$. We first prove the following
\begin{lemma}
There exists a constant $C$ which depends only on $M$ such that for all $r_{1}$ with $M<r_{1}$ and all solutions $\psi$ of the wave equation we have
\begin{equation*}
\begin{split}
\int_{\mathbb{S}^{2}}{\psi^{2}(r_{1},\omega)d\omega}\leq\frac{C}{(r_{1}-M)^{2}}\frac{E_{1}}{\tau^{2}}.
\end{split}
\end{equation*}
\label{1lemmadecay}
\end{lemma}

\begin{proof}
Using \eqref{1pointwise} we obtain 
\begin{equation*}
\begin{split}
\int_{\mathbb{S}^{2}}{\psi^{2}(r_{1},\omega)d\omega}&\leq \frac{C}{r_{1}}\int_{\tilde{\Sigma}_{\tau}\cap\left\{r\geq r_{1}\right\}}{J_{\mu}^{N}[\psi]n^{\mu}_{\tilde{\Sigma}_{\tau}}}=\frac{C}{r_{1}}\int_{\tilde{\Sigma}_{\tau}\cap\left\{r\geq r_{1}\right\}}{
\frac{D(\rho)}{D(\rho)}J_{\mu}^{N}[\psi]n^{\mu}_{\tilde{\Sigma}_{\tau}}}\\ 
&\leq \frac{C}{r_{1}D(r_{1})}\int_{\tilde{\Sigma}_{\tau}\cap\left\{r\geq r_{1}\right\}}{D(\rho)J_{\mu}^{N}[\psi]n^{\mu}_{\tilde{\Sigma}_{\tau}}}\\ &\leq
\frac{C}{(r_{1}-M)^{2}}\int_{\tilde{\Sigma}_{\tau}}{J_{\mu}^{T}[\psi]n^{\mu}_{\tilde{\Sigma}_{\tau}}}\leq\frac{C}{(r_{1}-M)^{2}}\frac{E_{1}}{\tau^{2}}.
\end{split}
\end{equation*}
\end{proof}

\begin{lemma}
There exists a constant $C$ which depends only on $M, R_{1}$ such that for all $r_{0}\in[M, R_{1}]$, $\a >0$ and  solutions $\psi$ of the wave equation, we have
\begin{equation*}
\begin{split}
\int_{\mathbb{S}^{2}}{\psi^{2}(r_{0},\omega)d\omega}&\leq CE_{1}\frac{1}{\tau^{2-2\a}}+C\sqrt{E_{1}}\frac{1}{\tau}\sqrt{\int_{\tilde{\Sigma}_{\tau}\cap\left\{r_{0}\leq r\leq r_{0}+\tau^{-\a}\right\}}{\!\!(\partial_{\rho}\psi)^{2}}}.
\end{split}
\end{equation*}
\label{2lemmadecay}
\end{lemma}
\begin{proof}
We consider the hypersurface $\gamma_{\a}=\left\{r=r_{0}+\tau^{-\a}\right\}$.
\begin{figure}[H]
	\centering
		\includegraphics[scale=0.11]{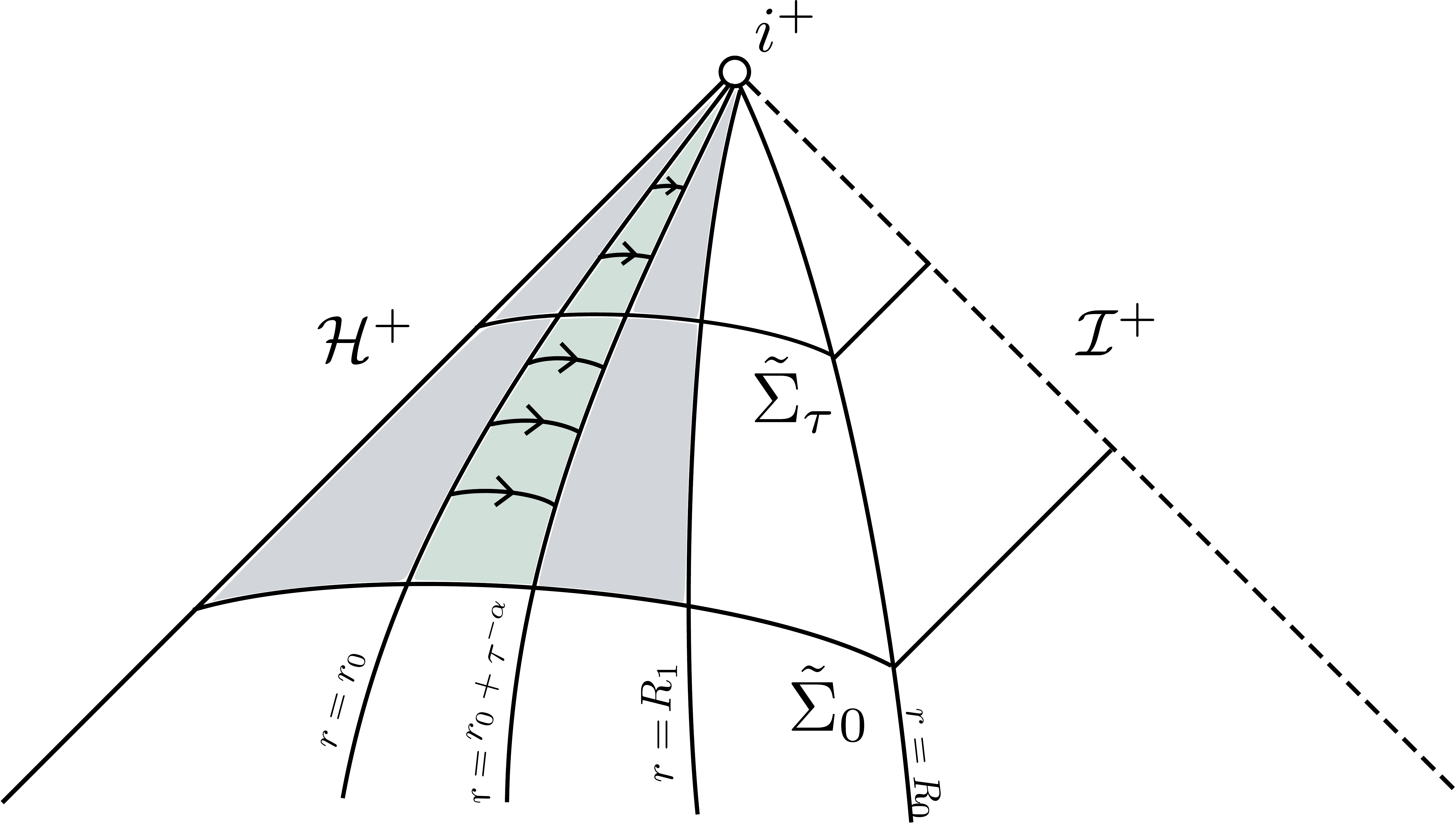}
		\label{fig:2decaypoint}
\end{figure}
Then by applying Stokes' theorem for the hypersurfaces shown in the figure above we obtain
\begin{equation*}
\begin{split}
\int_{\mathbb{S}^{2}}{\psi^{2}(r_{0},\omega)}\leq\int_{\mathbb{S}^{2}}{\psi^{2}(r_{0}+\tau^{-\a},\omega)}+C\int_{\tilde{\Sigma}_{\tau}\cap\left\{r_{0}\leq r\leq r_{0}+\tau^{-\a}\right\}}{\psi(\partial_{\rho}\psi)}.
\end{split}
\end{equation*}
For the first term on the right hand side we apply Lemma \ref{1lemmadecay} (note that $M<r_{0}+\tau^{-\a}$). The lemma now follows from Cauchy-Schwarz applied for the second term, the first Hardy inequality of \cite{aretakis1} and Theorem   \ref{t4}.
\end{proof}

\textbf{The case $l=0$}

We first assume that $\psi$ is spherically symmetric. Then we have the pointwise bound
\begin{equation*}
\begin{split}
\left|\partial_{\rho}\psi\right|\leq C\sqrt{\tilde{E}_{6}},
\end{split}
\end{equation*}
in $\left\{M\leq r\leq R_{1}\right\}$, where
\begin{equation*}
\begin{split}
\tilde{E}_{6}=\left\|\partial_{r}\psi\right\|_{L^{\infty}\left(\tilde{\Sigma}_{0}\right)}^{2}+E_{4}[\psi]+E_{4}[T\psi].
\end{split}
\end{equation*}
This can be easily proved by using the method of characteristics and  integrating along the characteristic $u=c$ the wave equation (expressed in null coordinates).
Hence  Lemma \ref{2lemmadecay} for $\a =\frac{2}{5}$ gives
\begin{equation*}
\begin{split}
\int_{\mathbb{S}^{2}}{\psi^{2}(r_{0},\omega)d\omega}&\leq CE_{1}\frac{1}{\tau^{\frac{6}{5}}}+C\sqrt{E_{1}}\sqrt{\tilde{E}_{6}}\frac{1}{\tau^{\frac{6}{5}}}\leq CE_{6}\frac{1}{\tau^{\frac{6}{5}}},
\end{split}
\end{equation*}
where 
$E_{6}=E_{1}+\tilde{E_{6}}$. Since $\psi$ is spherically symmetric we obtain
\begin{equation}
\begin{split}
\psi^{2}\leq CE_{6}\frac{1}{\tau^{\frac{6}{5}}}.
\end{split}
\label{l=0pointdecay}
\end{equation}

\textbf{The case $l=1$}

Suppose that $\psi$ is supported on $l=1$. Then from Lemma \ref{2lemmadecay} for $\a=\frac{1}{4}$ we obtain
\begin{equation*}
\begin{split}
\int_{\mathbb{S}^{2}}{\psi^{2}(r_{0},\omega)d\omega}&\leq CE_{1}\frac{1}{\tau^{\frac{3}{2}}}+C\sqrt{E_{1}}\frac{1}{\tau}\sqrt{\int_{\tilde{\Sigma}_{\tau}\cap\left\{r_{0}\leq r\leq r_{0}+\tau^{-\a}\right\}}{\!\!(\partial_{\rho}\psi)^{2}}}\\
&\leq CE_{1}\frac{1}{\tau^{\frac{3}{2}}}+C\sqrt{E_{1}}\sqrt{E_{2}}\frac{1}{\tau^{\frac{3}{2}}}\leq C\tilde{E}_{7}\frac{1}{\tau^{\frac{3}{2}}}
\end{split}
\end{equation*}
where we have used Proposition \ref{l1decay} and $\tilde{E}_{7}=E_{1}+E_{2}$. Therefore, by the Sobolev inequality on $\mathbb{S}^{2}$ we have 
\begin{equation*}
\begin{split}
\psi^{2}\leq CE_{7}\frac{1}{\tau^{\frac{3}{2}}}
\end{split}
\end{equation*}
in $\left\{M\leq r\leq R_{1}\right\}$, where $E_{7}=\sum_{\left|k\right|\leq 2}\tilde{E}_{7}{\left[\Omega^{k}\psi\right]}$.

\textbf{The case $l\geq 2$}

Suppose that $\psi$ is supported on $l\geq 2$. Then from \eqref{1pointwise} and Theorem \ref{energydecay} we have that there exists a constant $C$ which depends only on $M$ and $R_{1}$ such that 
\begin{equation*}
\begin{split}
\int_{\mathbb{S}^{2}}\psi^{2}\leq CE_{3}\frac{1}{\tau^{2}}
\end{split}
\end{equation*}
in $\left\{M\leq r\leq R_{1}\right\}$. By Sobolev we finally obtain
\begin{equation*}
\begin{split}
\psi^{2}\leq CE_{8}\frac{1}{\tau^{2}},
\end{split}
\end{equation*}
where $E_{8}=\sum_{\left|k\right|\leq 2}E_{3}{\left[\Omega^{k}\psi\right]}.$ This completes the proof of Theorem \ref{t6}.

\subsection{Applications}
\label{sec:Applications1}

We next derive improved decay for the waves $T\psi$, $TT\psi$.

\begin{proposition}

Fix $R_{1}$ such that $M<R_{1}$ and let $\tau\geq 1$. Let $E_{7},E_{8}$ be the quantities as defined in Section \ref{sec:PointwiseDecay}. Then there exists a constant $C$ that depends on $M, R_{1}$ and $\tilde{\Sigma}_{0}$ such that:
\begin{itemize}
	\item For all solutions $\psi$ to the wave equation we have 
	\begin{equation*}
	\left|T\psi\right|\leq C\sqrt{E_{7}[T\psi]}\frac{1}{\tau^{\frac{3}{4}}}, \ \ \left|TT\psi\right|\leq C\sqrt{E_{8}[TT\psi]}\frac{1}{\tau},
		\end{equation*}
	in $\left\{M\leq r\leq R_{1} \right\}$. 
	\item For all solutions $\psi$ to the wave equation which are supported on the frequencies $l\geq 1$ we have 
	\begin{equation*}
 \left|T\psi\right|\leq C\sqrt{E_{8}[T\psi]}\frac{1}{\tau},
 	\end{equation*}
	in $\left\{M\leq r\leq R_{1} \right\}$.

\end{itemize}

\label{ttpsi}
\end{proposition}

\begin{proof}
We first observe that Proposition \ref{inte1} holds for all frequencies when $\psi$ is replaced with $T\psi$. Indeed, we have shown that we can commute in this case with $\partial_{r}$ (see Theorem \ref{theorem3}, statement (2)). Therefore, using the third Hardy inequality of \cite{aretakis1} implies that \eqref{ninte1} and Proposition \ref{l1decay} hold for all frequencies when $\psi$ is replaced with $T\psi$. Hence, the  argument of Section \ref{sec:DecayNearMathcalH} for the case $l=1$ works for all frequencies when restricted for $T\psi$. 

Similarly, Proposition \ref{inte2} holds for all frequencies when $\psi$ is replaced with $TT\psi$ and so we can argue as in the case $l\geq 2$ above. The second part of the proposition follows in a similar way.

\end{proof}

\section{Higher Order  Estimates}
\label{sec:HigherOrderEstimates}

We finish this paper by obtaining energy and pointwise results for all the derivatives of $\psi$. We first derive decay for the local higher order (non-degenerate) energy of high frequencies and then pointwise decay, non-decay and blow-up results for generic solutions. We finally use a contradiction argument to obtain  blow-up results for the local higher order energy of low frequencies. The general form  of Theorems \ref{theo8}, \ref{theo9} of Section \ref{sec:TheMainTheorems} is proved at the end of this section.

\subsection{Decay of Higher Order Energy}
\label{sec:DecayOfHigherOrderEnergy}

\begin{theorem}
Fix $R_{1}$ such that $R_{1}>M$ and let $\tau\geq 1$. Let also $k,l\in\mathbb{N}$. Then there exists a constant $C$ which depend on $M,l, R_{1}$ and $\tilde{\Sigma}_{0}$ such that the following holds: For all solutions $\psi$ of the wave equation which are supported on the 
 angular frequencies greater or equal to $l$, there exist norms $\tilde{E}_{k,l}$ of the initial data of $\psi$ such that
\begin{itemize}
	\item $\displaystyle\int_{\tilde{\Sigma}_{\tau}\cap\left\{M\leq r\leq R_{1}\right\}}{J_{\mu}^{N}[\partial_{r}^{k}\psi]n_{\tilde{\Sigma}_{\tau}}^{\mu}}\leq C\tilde{E}_{k,l}^{2}\frac{1}{\tau^{2}}$ for all $k\leq l-2$,
	\item $\displaystyle\int_{\tilde{\Sigma}_{\tau}\cap\left\{M\leq r\leq R_{1}\right\}}{J_{\mu}^{N}[\partial_{r}^{l-1}\psi]n_{\tilde{\Sigma}_{\tau}}^{\mu}}\leq C\tilde{E}_{l-1,l}^{2}\frac{1}{\tau}$.
\end{itemize}
\label{hoe}
\end{theorem}
\begin{proof}
By commuting with T and applying local elliptic estimates and previous decay results, the above  integrals decay  on $\tilde{\Sigma}_{\tau}\cap\left\{r_{0}\leq r\leq R_{1}\right\}$ where $r_{0}>M$. So it suffices to prove the above result for $\tilde{\Sigma}_{\tau}\cap\mathcal{A}$, where $\mathcal{A}$ is a $\varphi^{T}_{\tau}$-invariant neighbourhood of $\mathcal{H}^{+}$. For we use the spacetime bound given by  Theorem \ref{theorem3} which implies that there exists a dyadic sequence $\tau_{n}$ such that for all $k\leq l-1$ we have
\begin{equation}
\begin{split}
\int_{\tilde{\Sigma}_{\tau_{n}}\cap\mathcal{A}}{J_{\mu}^{N}[\partial_{r}^{k}\psi]n^{\mu}_{\tilde{\Sigma}_{\tau_{n}}}}\leq CK_{l}\frac{1}{\tau_{n}},
\end{split}
\label{hoed1}
\end{equation}
where
\begin{equation*}
\begin{split}
K_{l}=\sum_{i=0}^{l}\displaystyle\int_{\tilde{\Sigma}_{0}}{J_{\mu}^{N}\left[T^{i}\psi\right]n^{\mu}_{\tilde{\Sigma}_{0}}}+\sum_{i=1}^{l}\int_{\tilde{\Sigma}_{0}\cap\mathcal{A}}{J_{\mu}^{N}\left[\partial^{i}_{r}\psi\right]n^{\mu}_{\tilde{\Sigma}_{0}}}.
\end{split}
\end{equation*}
Then, by Theorem \ref{theorem3} again we have for any $\tau$ such that $\tau_{n}\leq \tau\leq \tau_{n+1}$ 
\begin{equation*}
\begin{split}
\int_{\tilde{\Sigma}_{\tau}\cap\mathcal{A}}{J_{\mu}^{N}[\partial_{r}^{k}\psi]n^{\mu}_{\tilde{\Sigma}_{\tau}}}&\leq C\sum_{i=0}^{k}\displaystyle\int_{\tilde{\Sigma}_{\tau_{n}}}{\!\!J_{\mu}^{N}\left[T^{i}\psi\right]n^{\mu}_{\tilde{\Sigma}_{\tau_{n}}}}\!\!+C\sum_{i=1}^{k}\int_{\tilde{\Sigma}_{\tau_{n}}\cap\mathcal{A}}{\!J_{\mu}^{N}\left[\partial^{i}_{r}\psi\right]n^{\mu}_{\tilde{\Sigma}_{\tau_{n}}}}\leq C E\frac{1}{\tau_{n}}\lesssim CE\frac{1}{\tau},
\end{split}
\end{equation*}
where $E$ depends only on the initial data.  Suppose now that $k\leq l-2$. We apply Theorem \ref{theorem3} for the dyadic intervals $[\tau_{n}, \tau_{n-1}]$ and we obtain
\begin{equation*}
\begin{split}
\int_{\mathcal{A}}{J_{\mu}^{N}[\partial_{r}^{k}\psi]n_{\Sigma}^{\mu}}\leq C\sum_{i=0}^{l-1}\displaystyle\int_{\tilde{\Sigma}_{\tau_{n-1}}}{\!\!\!J_{\mu}^{N}\left[T^{i}\psi\right]n^{\mu}_{\tilde{\Sigma}_{\tau_{n-1}}}}\!\!+\sum_{i=1}^{l-1}\int_{\tilde{\Sigma}_{\tau_{n-1}}\cap\mathcal{A}}{\!\!\!J_{\mu}^{N}\left[\partial^{i}_{r}\psi\right]n^{\mu}_{\tilde{\Sigma}_{\tau_{n-1}}}}.
\end{split}
\end{equation*}
However, the right hand side has been shown to decay like $\tau^{-1}$ and thus a similar argument as above gives us the improved decay for all $k\leq l-2$.
\end{proof}

\subsection{Higher Order Pointwise Estimates}
\label{sec:HigherOrderPointwiseEstimates}

The next theorem provides pointwise results for the derivatives transversal to $\hh$ of $\psi$.
\begin{theorem}
Fix $R_{1}$ such that $R_{1}>M$ and let $\tau\geq 1$. Let also $k,l,m\in\mathbb{N}$. Then, there exist constants $C$ which depend on $M,l, R_{1}$ and $\tilde{\Sigma}_{0}$ such that the following holds: For all solutions $\psi$ of the wave equation which are supported on  angular frequencies greater or equal to $l$, there exist norms $E_{k,l}$ of the initial data  $\psi$ such that
\begin{itemize}
	\item $\displaystyle\left|\partial_{r}^{k}\psi\right|\leq CE_{k,l}\displaystyle\frac{1}{\tau}$ in $\left\{M\leq r\leq R_{1}\right\}$ for all $k\leq l-2$,
	\item $\displaystyle\left|\partial_{r}^{l-1}\psi\right|\leq CE_{l-1,l}\displaystyle\frac{1}{\tau^{\frac{3}{4}}}$ in $\left\{M\leq r\leq R_{1}\right\}$,
		\item $\displaystyle\left|\partial_{r}^{l}\psi\right|\leq CE_{l,l}\displaystyle\frac{1}{\tau^{\frac{1}{4}}}$ in $\left\{M\leq r\leq R_{1}\right\}$.
		\end{itemize}
\label{hp3}
\end{theorem}
\begin{proof}
 Let $r_{0}$ such that $M\leq r_{0}\leq R_{1}$. We consider the cut-off $\delta:[M,R_{1}+1]\rightarrow [0,1]$ such that $\delta (r)=1,\text{ for } r\leq R_{1}+\frac{1}{4}$ and $\delta(r)=0, \text{ for } R_{1}+1/2\leq r \leq R_{1}+1$.
\begin{figure}[H]
	\centering
		\includegraphics[scale=0.15]{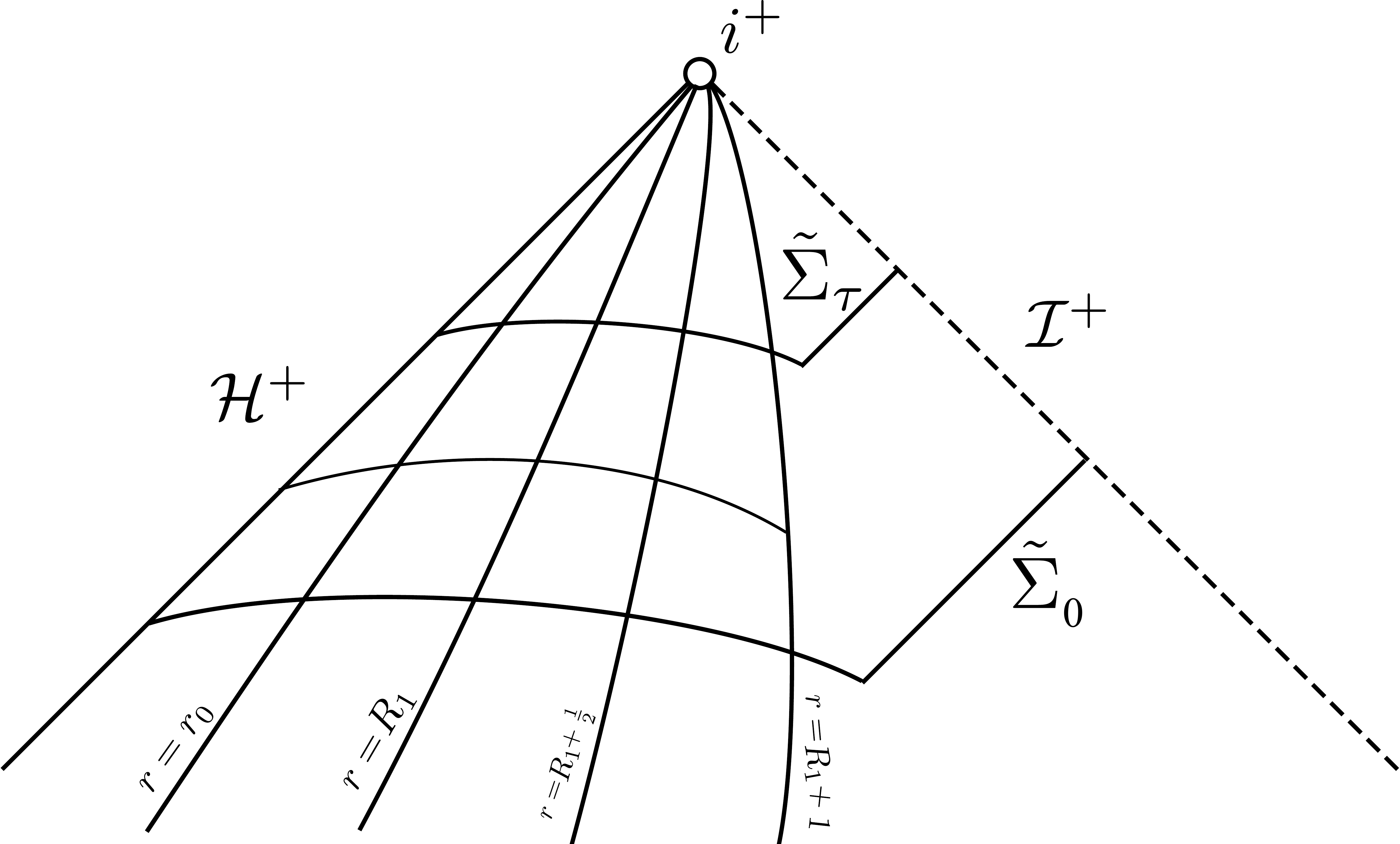}
	\label{fig:hp1}
\end{figure} 
Then,
\begin{equation*}
\begin{split}
\int_{\mathbb{S}^{2}}{\left(\partial_{r}^{k}\psi\right)^{2}(r_{0},\omega)d\omega}&=-2\int_{\tilde{\Sigma}_{\tau}\cap\left\{r_{0}\leq r\leq R_{1}+1\right\}}{\left(\partial_{r}^{k}(\delta\psi)\right)\left(\partial_{\rho}\partial_{r}^{k}(\delta\psi)\right)}\\
&\leq 2\!\left(\int_{\tilde{\Sigma}_{\tau}\cap\left\{r\leq R_{1}+1\right\}}{\!\!\!\left(\partial_{r}^{k}(\delta\psi)\right)^{2}}\right)^{\frac{1}{2}}\!\!\left(\int_{\tilde{\Sigma}_{\tau}\cap\left\{r\leq R_{1}+1\right\}}{\!\!\!\left(\partial_{\rho}\partial_{r}^{k}(\delta\psi)\right)^{2}}\right)^{\frac{1}{2}}\!\!\!.
\end{split}
\end{equation*}
In view of Theorem \ref{hoe} if $k\leq l-2$ then both integrals on the right hand side decay like $\tau^{-2}$. If $k=l-1$ then the first integral decays like $\tau^{-2}$ and the second like $\tau^{-1}$ and if $k=l$ the first integral decays like $\tau^{-1}$ and the second is bounded (Theorem \ref{theorem3}). Commuting with the angular momentum operators and using the Sobolev inequality yield the required pointwise estimates for $\partial_{r}^{k}\psi$ for $k\leq l$.

\end{proof}
One can in fact show that $\partial_{r}^{l}\psi$ decays like $\tau^{-\frac{1}{4}-\delta_{l}},$ where $\delta_{l}>0$ by using the argument of Section \ref{sec:PointwiseEstimates} for the case $l=0$, i.e. by proving that $\partial_{r}^{l+1}\psi$ is uniformly bounded (note that we can not obtain $\frac{3}{5}$ decay in view of the fact that we can not use the first Hardy inequality, which allowed us to obtain further decay for the zeroth order term $\psi$ in the previous section). We leave the details to the reader. 

Let now $H_{l}[\psi]$ be the function on $\mathcal{H}^{+}$ as defined in Theorem \ref{t3}. Since $H_{l}[\psi]$ is conserved along the null geodesics of $\hh$ whenever $\psi$ is supported on the angular frequency $l$, we can simply think of $H_{l}[\psi]$ as a function on $\mathbb{S}^{2}_{\scriptsize 0}=\tilde{\Sigma}_{0}\cap\hh$. We then have the following non-decay result.
\begin{proposition}
For all solutions $\psi$ supported on the angular frequency $l$ we  have  
\begin{equation*}
\partial_{r}^{l+1}\psi(\tau, \theta, \phi)\rightarrow H_{l}[\psi](\theta, \phi)
\end{equation*}
along $\hh$ and generically $H_{l}[\psi](\theta,\phi)\neq 0$ almost everywhere on $\mathbb{S}^{2}_{0}$.
\label{nondecay}
\end{proposition}
\begin{proof}
Since 
\begin{equation*}
\partial_{r}^{l+1}\psi(\tau, \theta, \phi) +\sum_{i=0}^{l}\beta_{i}\partial_{r}^{i}\psi(\tau, \theta, \phi)=H_{l}[\psi](\theta,\phi)
\end{equation*}
on $\hh$ and since all the terms in the sum on the left hand side decay (see Theorem \ref{hp3}) we take  $\partial_{r}^{l+1}\psi(\tau, \theta, \phi)\rightarrow H_{l}[\psi](\theta, \phi)$ on $\hh$. It suffices to show that generically $H_{l}[\psi](\theta,\phi)\neq 0$ almost everywhere on $\mathbb{S}^{2}_{0}$. We will in fact show that for generic solutions $\psi$ of the wave equation the function $H_{l}[\psi]$ is a generic eigenfunction of order $l$ of $\lapp$ on $\mathbb{S}^{2}_{0}$.

Note that the initial data prescribed on $\tilde{\Sigma}_{0}$ do not a priori  determine the function $H_{l}[\psi]$ on $\mathbb{S}^{2}_{0}$ unless $l=0$. Indeed,  $H_{l}[\psi]$ involves derivatives of order $k\leq l+1$ which are not tangential to $\tilde{\Sigma}_{0}$. For this reason we consider another Cauchy problem of the wave equation with initial data prescribed on $\tilde{\Sigma}_{0}^{p}$, where the hypersurface $\tilde{\Sigma}_{0}^{p}$ is as depicted below:
\begin{figure}[H]
	\centering
		\includegraphics[scale=0.11]{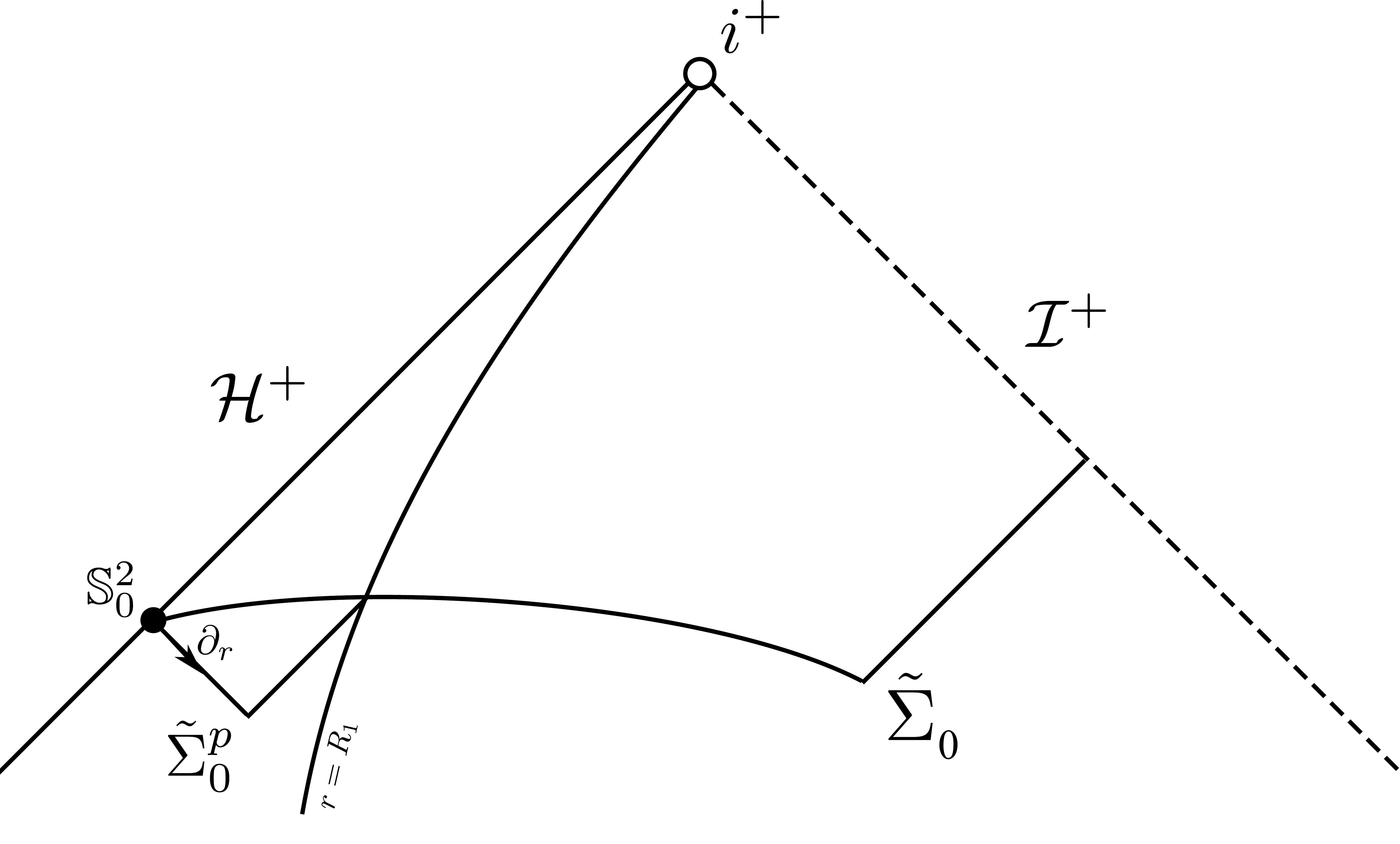}
	\label{fig:s0}
\end{figure} 
Note that the hypersurfaces $\tilde{\Sigma}_{0}$ and $\tilde{\Sigma}_{0}^{p}$ coincide for $r\geq R_{1}$. Any initial data set prescribed on $\tilde{\Sigma}_{0}$ gives rise to a unique initial data set of $\tilde{\Sigma}_{0}^{p}$ and vice versa. The Sobolev norms of the initial data on $\tilde{\Sigma}_{0}$ and $\tilde{\Sigma}_{0}^{p}$ can be compared using the pointwise and energy boundedness  of \cite{aretakis1}. Observe now that given initial data on $\tilde{\Sigma}_{0}^{p}$ the function $H_{l}[\psi]$ is completely determined on $\mathbb{S}^{2}_{0}$, since $H_{l}[\psi]$ involves only tangential to $\tilde{\Sigma}_{0}^{p}$ derivatives at $\mathbb{S}_{0}$. Therefore, generic initial data on $\tilde{\Sigma}_{0}^{p}$ give rise to generic eigenfunctions $H_{l}[\psi]$ of order $l$ of $\lapp$ on $\mathbb{S}_{0}$. Hence, for generic solutions $\psi$ of the wave equation the functions $H_{l}[\psi]$  do not vanish almost everywhere on $\mathbb{S}_{0}$.

\end{proof}

We next show that the above non-decay results imply that  higher order derivatives of generic solutions $\psi$ blow-up along $\mathcal{H}^{+}$. To make our argument clear we first consider the spherically symmetric case where $l=0$. 

\begin{proposition}
Let $k\in\mathbb{N}$ with $k\geq 2$. Then there exists a positive constant $c$ which depends only on $M$ such that for all spherically symmetric solutions $\psi$ to the wave equation  we have 
\begin{equation*}
\left|\partial_{r}^{k}\psi\right|\geq c \left|H_{0}[\psi]\right|\tau^{k-1}
\end{equation*}
asymptotically on $\mathcal{H}^{+}$.
\label{rrl=0}
\end{proposition}
\begin{proof}
We work inductively. Consider the case $k=2$. By differentiating the wave equation (see for instance \eqref{kcom}) we take
\begin{equation}
2T\partial_{r}\partial_{r}\psi+\frac{2}{M}T\partial_{r}\psi-\frac{2}{M^{2}}T\psi+\frac{2}{M^{2}}\partial_{r}\psi=0
\label{1edw}
\end{equation}
on $\mathcal{H}^{+}$. Note that  $T\partial_{r}^{2}\psi$ and $\partial_{r}\psi$ appear with the same sign. If $H_{0}[\psi]=0$ then there is nothing to prove.  Let's suppose that $H_{0}>0$. Then
\begin{equation*}
\begin{split}
\int_{0}^{\tau}\partial_{r}\psi=\int_{0}^{\tau}H_{0}[\psi]-\frac{1}{M}\psi=H_{0}[\psi]\tau-\frac{1}{M}\int_{0}^{\tau}\psi.
\end{split}
\end{equation*}
We observe 
\begin{equation*}
\begin{split}
\left|\int_{0}^{\tau}\psi\right|\leq \int_{0}^{\tau}\left|\psi\right|\leq CE_{6}\int_{0}^{\tau}\frac{1}{\tau^{\frac{3}{5}}}=CE_{6}\tau^{\frac{2}{5}}.
\end{split}
\end{equation*}
Therefore, 
\begin{equation*}
\begin{split}
\int_{0}^{\tau}\partial_{r}\psi\geq H_{0}[\psi]\tau-CE_{6}\tau^{\frac{2}{5}}\geq cH_{0}[\psi]\tau
\end{split}
\end{equation*}
asymptotically on $\hh$. By integrating \eqref{1edw} along $\mathcal{H}^{+}$ we obtain
\begin{equation*}
\begin{split}
\partial_{r}^{2}\psi(\tau)&=\partial_{r}^{2}\psi(0)+\frac{1}{M}\partial_{r}\psi(0)-\frac{1}{M}\partial_{r}\psi(\tau)-\frac{1}{M^{2}}\psi(0)+\frac{1}{M^{2}}\psi(\tau)-\frac{1}{2M^{2}}\int_{0}^{\tau}\partial_{r}\psi\\
&\leq \partial_{r}^{2}\psi(0)+\frac{1}{M}\partial_{r}\psi(0)-\frac{1}{M}\left(H_{0}[\psi]+\frac{1}{M}\psi\right)-\frac{1}{M^{2}}\psi(0)+\frac{1}{M^{2}}\psi(\tau)-\frac{1}{2M^{2}}\int_{0}^{\tau}\partial_{r}\psi\\
&\leq \partial_{r}^{2}\psi(0)+\frac{1}{M}\partial_{r}\psi(0)-\frac{1}{M}H_{0}[\psi]+CE_{6}\frac{1}{\tau^{\frac{3}{5}}}-cH_{0}[\psi]\tau\\
&\leq -cH_{0}[\psi]\tau
\end{split}
\end{equation*}
asymptotically on $\hh$. A similar argument works for any $k\geq 2$. Indeed, we integrate \eqref{kcom} for $k\geq 1$ (and $l=0$) along $\hh$ and note that  $T\partial_{r}^{k+1}\psi$ and $\partial_{r}^{k}\psi$ appear with the same sign. Therefore, by induction on $k$, the integral $\int_{0}^{\tau}\partial_{r}^{k}\psi$ dominates asymptotically all the remaining terms which yields the required blow-up rates on $\mathcal{H}^{+}$. Note that the sign of $\partial_{r}^{k}\psi$ depends on $k$ and $H_{0}[\psi]$.

\end{proof}

\begin{corollary}
Let $k\geq 2$. For generic initial data which give rise to solutions $\psi$ of the wave equation we have
\begin{equation*}
\begin{split}
\left|\partial_{r}^{k}\psi\right|\rightarrow +\infty
\end{split}
\end{equation*}
along $\mathcal{H}^{+}$.
\end{corollary}
\begin{proof}
Decompose $\psi=\psi_{0}+\psi_{\geq 1}$ and thus
\begin{equation*}
\begin{split}
\int_{\mathbb{S}^{2}}{\left|\partial_{r}^{k}\psi\right|^{2}(M,\omega)d\omega}\geq 4\pi\left|\partial_{r}^{k}\psi_{0}\right|^{2}(M,\omega).
\end{split}
\end{equation*}
Hence the result follows by commuting with $\Omega_{i}$, the Sobolev inequality and  the fact that the right hand side blows up as $\tau\rightarrow +\infty$ as $H_{0}[\psi]\neq 0$ generically. 
\end{proof}
Let us consider the case of a general angular frequency $l$. 
\begin{proposition}
Let $k,l\in\mathbb{N}$ with $k\geq 2$. Then there exists a positive constant $c$ which depends only on $M,l,k$ such that for all  solutions $\psi$ to the wave equation which are supported on the frequency $l$ we have 
\begin{equation*}
\left|\partial_{r}^{l+k}\psi\right|(\tau,\theta,\phi)\geq c \left|H_{l}[\psi](\theta,\phi)\right|\tau^{k-1}
\end{equation*}
asymptotically on $\mathcal{H}^{+}$.
\label{rrl}
\end{proposition}
\begin{proof}
We first consider $k=2$. If $H_{l}[\psi](\theta,\phi)=0$ then there is nothing to prove. Suppose that $H_{l}[\psi](\theta,\phi)>0$. Note 
 \begin{equation*}
 \begin{split}
 \int_{0}^{\tau}\partial_{r}^{l+1}\psi&=H_{l}[\psi]\tau-\int_{0}^{\tau}\sum_{i=0}^{l}\beta_{i}\partial_{r}^{i}\psi\\
 &\geq cH_{l}[\psi]\tau
 \end{split}
 \end{equation*}
 asymptotically on $\hh$, since the integral on the right hand side is eventually dominated by $H_{l}[\psi]\tau$ in view of Theorem \ref{hp3}. If we integrate \eqref{kcom} (applied for $k=l+1)$ along the null geodesic of $\mathcal{H}^{+}$ whose projection on the sphere is $(\theta,\phi)$ we will obtain
  \begin{equation*}
 \begin{split}
 \partial_{r}^{l+2}\psi(\tau,\theta,\phi)\leq -cH_{l}[\psi](\theta,\phi)\tau,
  \end{split}
 \end{equation*}
 since the integral $ \int_{0}^{\tau}\partial_{r}^{l+1}\psi$ eventually dominates all the remaining terms (again in view of the previous decay results). The proposition follows inductively by integrating \eqref{kcom} as in Proposition \ref{rrl=0}. Recall finally that for generic solutions $\psi$ we have $H_{l}[\psi]\neq 0$ almost everywhere on $\mathbb{S}^{2}_{0}$.

\end{proof}

\subsection{Blow-up of Higher Order Energy}
\label{sec:BlowUpOfHigherOrderEnergy}

The next theorem provides  blow-up results for the higher order non-degenerate energy. It also shows that our estimates in Section \ref{sec:HigherOrderEstimates} are in fact sharp (regarding at least the restriction on the angular frequencies). 
\begin{theorem}
Fix $R_{1}$ such that $R_{1}>M$. Let also $k,l\in\mathbb{N}$. Then  for generic solutions $\psi$ of the wave equation which are supported on the (fixed) angular frequency $l$ we have
\begin{equation*}
\displaystyle\int_{\tilde{\Sigma}_{\tau}\cap\left\{M\leq r\leq R_{1}\right\}}{J_{\mu}^{N}[\partial_{r}^{k}\psi]n_{\tilde{\Sigma}_{\tau}}^{\mu}}\longrightarrow +\infty 
\end{equation*}
as $\tau\rightarrow +\infty$ for all $k\geq l+1$.
\label{hoeblowup}
\end{theorem}
\begin{proof}

 Consider $M<r_{0}<R_{1}$ and let $\delta$ be the cut-off introduced in the proof of Theorem \ref{hp3}.  Then,
\begin{equation*}
\begin{split}
\int_{\mathbb{S}^{2}}{\left(\partial_{r}^{k}\psi\right)^{2}(r_{0},\omega)d\omega}&=-2\int_{\mathbb{S}^{2}}{\int_{r_{0}}^{R_{1}+1}{(\partial_{r}^{k}(\delta\psi))(\partial_{\rho}\partial_{r}^{k}(\delta\psi))d\rho}d\omega}\\
&\leq C\int_{\tilde{\Sigma}_{\tau}\cap\left\{r_{0}\leq r\leq R_{1}+1\right\}}{\sum_{i=0}^{k}(T\partial_{r}^{i}\psi)^{2}+\sum_{i=0}^{k+1}(\partial_{r}^{i}\psi)^{2}},
\end{split}
\end{equation*}
where $C$ depends on $M$,  $R_{1}$ and $\tilde{\Sigma}_{0}$. Then,
\begin{equation*}
\begin{split}
\int_{\mathbb{S}^{2}}{\left(\partial_{r}^{k}\psi\right)^{2}(r_{0},\omega)d\omega}&\leq \frac{C}{D^{m_{k}}(r_{0})}\int_{\tilde{\Sigma}_{\tau}\cap\left\{r_{0}\leq r\leq R_{1}+1\right\}}{\sum_{i=0}^{k}J_{\mu}^{T}[T^{i}\psi]n^{\mu}_{\tilde{\Sigma}_{\tau}}}\\
&\leq \frac{C}{(r_{0}-M)^{2m_{k}}}\left(\sum_{i=0}^{k}E_{1}(T^{i}\psi)\right)\frac{1}{\tau^{2}},
\end{split}
\end{equation*}
where $m_{k}\in\mathbb{N}$. Note that for the above inequality we used local elliptic estimates (or a more pedestrian way is to use the wave equation and solve with respect to $\partial_{r}^{k}\psi$; this is something we can do since $D(r_{0})>0$). Then using \eqref{kcom} we can inductively replace the $\partial_{r}$ derivatives with the $T$ derivatives.  Therefore, commuting with $\Omega_{i}$ and applying the Sobolev inequality imply that  for any $r_{0}>M$ we have $\left|\partial_{r}^{k}\psi\right|\rightarrow 0$ as $\tau\rightarrow +\infty$ along $r=r_{0}$. Let us assume now that the energy of $\partial_{r}^{k}\psi$ on $\tilde{\Sigma}_{\tau_{j}}\cap\left\{M\leq r\leq R_{1}\right\}$ is uniformly bounded by $B$ (as $\tau_{j}\rightarrow +\infty$). Given $\epsilon >0$ take $r_{0}$ such that $r_{0}-M=\frac{\epsilon^{2}}{4Br_{0}^{2}}$ and let $\tau_{\epsilon}$ be such that for all $\tau\geq \tau_{\epsilon}$ we have $\left|\partial_{r}^{k}\psi(\tau,r_{0})\right|\leq \frac{\epsilon}{8\pi}$. Then,
\begin{equation*}
\begin{split}
\int_{\mathbb{S}^{2}}\left|\partial_{r}^{k}\psi(\tau_{j},M)\right|&\leq \int_{\mathbb{S}^{2}}\left|\partial_{r}^{k}\psi(\tau_{j}, r_{0})\right|+\int_{\tilde{\Sigma}_{\tau_{j}}\cap\left\{M\leq r\leq r_{0}\right\}}\left|\partial_{\rho}\partial_{r}^{k}\psi\right|\\
&\leq \frac{\epsilon}{2}+r_{0}(r_{0}-M)^{\frac{1}{2}}\left(\int_{\tilde{\Sigma}_{\tau_{j}}\cap\left\{M\leq r\leq R_{1}\right\}}{J_{\mu}^{N}[\partial_{r}^{k}\psi]n_{\tilde{\Sigma}_{\tau_{j}}}^{\mu}}\right)^{\frac{1}{2}}\leq \epsilon,
\end{split}
\end{equation*}
for all $\tau\geq \tau_{\epsilon}$.  This proves that $\int_{\mathbb{S}^{2}}\left|\partial_{r}^{k}\psi(\tau_{j},M)\right|\rightarrow 0$ as $t_{j}\rightarrow +\infty$ along $\hh$. However, in view of Propositions \ref{nondecay} and \ref{rrl} we have
\begin{equation*}
\int_{\mathbb{S}^{2}(M)}\left|\partial_{r}^{k}\psi\right|(\tau_{j})\geq c\tau_{j}^{k-1}\int_{\mathbb{S}^{2}(M)}\left|H_{l}[\psi]\right|.
\end{equation*}
We have seen that for generic $\psi$ the function $H_{l}[\psi]$ is non-zero almost everywhere and since it is smooth we have $\int_{\mathbb{S}^{2}(M)}\left|H_{l}[\psi]\right|>0$. This shows that the integral $\int_{\mathbb{S}^{2}}\left|\partial_{r}^{k}\psi(\tau_{j},M)\right|$ can not decay, contradiction.
\end{proof}
\subsection{Applications}
\label{sec:Applications2}
We conclude this paper by proving Theorem \ref{theo8} and \ref{theo9} of Section \ref{sec:TheMainTheorems} which provide the complete picture for the derivatives of $\psi$.
\begin{proof}[Proof of Theorem \ref{theo8}]
The first two statements can be proved as the Proposition \ref{hoe} by observing that in view of statement (2) of Theorem \ref{theorem3} we can find a dyadic sequence $\tau_{n}$ such that \eqref{hoed1} holds for all $k\leq l+m-1$. Having proved these two statements, the remaining ones can be proved by repeating the argument of Theorem \ref{hp3}. 
\end{proof}

\begin{proof}[Proof of Theorem \ref{theo9}]
The first statement follows from Proposition \ref{tmpsi} and the previous decay results. For the second statement we integrate \eqref{kcom} for $k=l+m+1$ and $\psi$ replaced with $T^{m-1}\psi$ and observe that $\int_{0}^{\tau}\partial_{r}^{l+m+1}T^{m}\psi$ dominates eventually all the remaining terms. However, in view of statement (1) we have 
$\int_{0}^{\tau}\partial_{r}^{l+m+1}T^{m}\psi\sim cH_{j}[\psi]\tau$ which completes the proof for $k=l+m+2$.  The general case can be proved by induction on $k$. Finally, the above two statements imply that generically the integral $\int_{\mathbb{S}^{2}}\left|\partial_{r}^{k}T^{m}\psi(\tau,\theta,\phi)\right|$ can not decay as $\tau\rightarrow +\infty$ along $\hh$. Therefore, the last statement can be proved by repeating the argument of Theorem \ref{hoeblowup}.

\end{proof}

\section{Acknowledgements}
\label{sec:Acknowledgements}

I would like to thank Mihalis Dafermos for introducing to me  the problem and for his teaching and advice. I also thank Igor Rodnianski for sharing useful insights. I am very grateful to an anonymous referee who carefully read  this work and suggested many improvements and corrections. I am supported by a Bodossaki Grant.

\appendix

\section{Elliptic Estimates on Lorentzian Manifolds}
\label{sec:EllipticEstimates}

Let us suppose that $\left(\mathcal{M},g\right)$ is a globally hyperbolic time-orientable Lorentzian manifold which admits a Killing vector field $T$. We also suppose that $\mathcal{M}$ is   foliated by spacelike hypersurfaces $\Sigma_{\tau}$, where $\Sigma_{\tau}=\phi_{\tau}\left(\Sigma_{0}\right)$. Here, $\Sigma_{0}$ is a Cauchy hypersurface and $\phi_{\tau}$ is the flow of $T$.

Let $N$ be a $\phi_{\tau}$-invariant \textbf{timelike} vector field and constants $B_{1},B_{2}$ such that
\begin{equation*}
0<B_{1}<-g(N,N)<B_{2}.
\end{equation*}
 
We will first derive the required estimate in $\Sigma_{0}$ which for simplicity we denote by $\Sigma$. For each point $p\in\Sigma$ the orthogonal complement in $T_{p}\mathcal{M}$ of the line that contains $N$ is 3-dimensional and contains a 2-dimensional subspace of the tangent space $T_{p}\Sigma$. Let $X_{2},X_{3}$ be an orthonormal basis of this subspace. Let now $X_{1}$ be a vector tangent to $\Sigma$ which is perpendicular to the plane that is spanned by $X_{2},X_{3}$. Note that the line that passes through $X_{1}$ is uniquely determined by $N$ and $\Sigma$. Then, the metric $g$ can be written as 
\begin{equation*}
\begin{split}
g=\begin{pmatrix}
g_{NN}&g_{NX_{1}}&0&0\\
g_{NX_{1}}&g_{X_{1}X_{1}}&0&0\\
0&0&1&0\\
0&0&0&1\\
\end{pmatrix},\ \ \  g^{-1}=\begin{pmatrix}
\frac{1}{\left|g\right|}g_{X_{1}X_{1}}&-\frac{1}{\left|g\right|}g_{NX_{1}}&0&0\\
-\frac{1}{\left|g\right|}g_{NX_{1}}&\frac{1}{\left|g\right|}g_{NN}&0&0\\
0&0&1&0\\
0&0&0&1\\
\end{pmatrix}.
\end{split}
\end{equation*}
with respect to the frame $\left(N,X_{1},X_{2},X_{3}\right)$ and $\left|g\right|=g_{NN}\cdot g_{X_{1}X_{1}}-g_{NX_{1}}^{2}$. Let $h_{\Sigma}$ be the induced Riemannian metric on the spacelike hypersurface $\Sigma$. Clearly, in general we do \textbf{not} have $h^{ij}_{\Sigma}= g^{ij}$. Indeed
\begin{equation*}
\begin{split}
h_{\Sigma}=\begin{pmatrix}
g_{X_{1}X_{1}}&0&0\\
0&1&0\\
0&0&1\\
\end{pmatrix}, \ \ \ h_{\Sigma}^{-1}=\begin{pmatrix}
\frac{1}{g_{X_{1}X_{1}}}&0&0\\
0&1&0\\
0&0&1\\
\end{pmatrix}.
\end{split}
\end{equation*}
Let $\psi :\mathcal{M}\rightarrow\mathbb{R}$ satisfy the wave equation. Then,
\begin{equation*}
\begin{split}
\Box_{g}\psi &=\text{tr}_{g}\left(\text{Hess}\,\psi\right)=g^{\a\b}\left(\nabla^{2}\psi\right)_{\a\b}=g^{0\b}\left(\left(\nabla^{2}\psi\right)_{0\b}+\left(\nabla^{2}\psi\right)_{\b 0}\right)+g^{ij}\left(\nabla^{2}\psi\right)_{ij}.
\end{split}
\end{equation*}
We will prove that the operator 
\begin{equation*}
P\psi=g^{ij}\left(\nabla^{2}\psi\right)_{ij}
\end{equation*}
is strictly elliptic. Indeed, in view of the formula $\left(\nabla^{2}\psi\right)_{ij}=X_{i}X_{j}\psi-\left(\nabla_{X_{i}}X_{j}\right)\psi,$
the principal part $\sigma$ of $P$ is 
\begin{equation*}
\begin{split}
\sigma\psi=g^{ij}X_{i}X_{j}\psi.
\end{split}
\end{equation*}
If $\xi\in T^{*}\Sigma$, then 
\begin{equation*}
\begin{split}
\sigma\xi&=g^{ij}\xi_{i}\xi_{j}=\frac{1}{\left|g\right|}g_{NN}\xi_{1}^{2}+\xi_{2}^{2}+\xi_{3}^{3}>b\left(\frac{1}{g_{X_{1}X_{1}}}\xi_{1}^{2}+\xi_{2}^{2}+\xi_{3}^{3}\right)=b\left\|\xi\right\|,
\end{split}
\end{equation*}
where the ellipticity constant $b>0$  depends only on $\Sigma$. Moreover, if $\psi$ satisfies $\Box_{g}\psi=0$ then
\begin{equation*}
\begin{split}
\left\|P\psi\right\|^{2}_{L^{2}(\Sigma)}= &\left\|g^{0\b}\left(\left(\nabla^{2}\psi\right)_{0\b}+\left(\nabla^{2}\psi\right)_{\b 0}\right)\right\|_{L^{2}(\Sigma)}^{2}\\
\leq &C\int_{\Sigma}{ \left(\left\|NN\psi\right\|_{L^{2}\left(\Sigma\right)}^{2}+\sum_{i=1}^{3}{\left\|X_{i}N\psi\right\|_{L^{2}\left(\Sigma\right)}^{2}}+\sum_{i=1}^{3}{\left\|X_{i}\psi\right\|_{L^{2}\left(\Sigma\right)}^{2}}+\left\|N\psi\right\|_{L^{2}\left(\Sigma\right)}^{2}\right)}\\
\leq & C\int_{\Sigma}{J_{\mu}^{N}[\psi]n^{\mu}_{\Sigma}+J_{\mu}^{N}[N\psi]n^{\mu}_{\Sigma}},
\end{split}
\end{equation*}
where $C$ is a uniform constant that depends only on the geometry of $\Sigma$ and the precise choice of $N$. Therefore, if $\psi$ can be  shown to appropriately decay at infinity then by a global elliptic estimate on $\Sigma$ we obtain
\begin{equation*}
\begin{split}
\left\|\psi\right\|_{\overset{\!\!.}{H^{1}}\left(\Sigma\right)}^{2}+\left\|\psi\right\|_{\overset{\!\!.}{H^{2}}\left(\Sigma\right)}^{2}&\leq C\cdot\left\|P\psi\right\|_{L^{2}\left(\Sigma\right)}^{2}\leq\int_{\Sigma}{ CJ_{\mu}^{N}[\psi]n^{\mu}_{\Sigma}+CJ_{\mu}^{N}[N\psi]n^{\mu}_{\Sigma}}.
\end{split}
\end{equation*}
for some uniform positive constant $C$ (here ${\overset{\!\!.}{H^{k}}\left(\Sigma\right)}$ denotes the homogeneous Sobolev space where the zeroth order term is omitted).

In case our analysis is local and thus we want to confine ourselves in a compact submanifold $\overline{\Sigma}$ of $\Sigma$ then by a local elliptic estimate on $\overline{\Sigma}$ we have
\begin{equation*}
\begin{split}
\left\|\psi\right\|_{{H}^{2}\left(\overline{\Sigma}\right)}^{2}&\leq C\cdot\left\|P\psi\right\|_{L^{2}\left(\overline{\Sigma}\right)}^{2}+\left\|\psi\right\|_{{H}^{1}\left(\overline{\Sigma}\right)}^{2}\leq \int_{\overline{\Sigma}}{\left(CJ_{\mu}^{N}[\psi]n^{\mu}_{\overline{\Sigma}}+CJ_{\mu}^{N}[N\psi]n^{\mu}_{\overline{\Sigma}}+\psi^{2}\right)}.
\end{split}
\end{equation*}

One can also estimate spacetime integrals by using elliptic estimates. Indeed, if $\overline{\mathcal{R}}\left(0,\tau\right)$ is the spacetime region as defined before, then  
\begin{equation*}
\int_{\overline{\mathcal{R}}\left(0,\tau\right)}{f\left|\nabla u\right|dg_{\overline{\mathcal{R}}}}=\int_{0}^{\tau}{\left(\int_{\overline{\Sigma}_{\tau}}{f}dg_{\overline{\Sigma}_{\tau}}\right)dt},
\end{equation*}
where the integrals are with respect to the induced volume form and $u:\mathcal{M}\rightarrow\mathbb{R}$ is such that $u\left(p\right)=\tau$ iff $p\in\overline{\Sigma}_{\tau}$. Then $\nabla u$ is proportional to $n_{\overline{\Sigma}_{\tau}}$ and since $T(u)=1$, $\nabla u$ is $\phi_{\tau}$-invariant. Therefore, $\left|\nabla u\right|$ is uniformly bounded. If now $f$ is quadratic on the 2-jet of $\psi$ then
\begin{equation*}
\begin{split}
\left|\int_{\overline{\mathcal{R}}\left(0,\tau\right)}{fdg_{\overline{\mathcal{R}}}}\right|&\leq C\int_{0}^{\tau}{\left\|\psi\right\|_{{H}^{2}\left(\overline{\Sigma}_{\tilde{\tau}}\right)}^{2}d\tilde{\tau}}\leq C \int_{0}^{\tau}{\left(\int_{\overline{\Sigma}_{\tilde{\tau}}}{ J_{\mu}^{N}[\psi]n^{\mu}_{\overline{\Sigma}_{\tilde{\tau}}}+J_{\mu}^{N}[N\psi]n^{\mu}_{\overline{\Sigma}_{\tilde{\tau}}}}+\psi^{2}\right)d\tilde{\tau}}\\
&\leq C\int_{\overline{\mathcal{R}}\left(0,\tau\right)}{J_{\mu}^{N}[\psi]n^{\mu}_{\overline{\Sigma}}+CJ_{\mu}^{N}[N\psi]n^{\mu}_{\overline{\Sigma}}+\psi^{2}}.
\end{split}
\end{equation*}
In applications  we usually use these results away from $\mathcal{H}^{+}$ where we commute with $T$ and we use the degenerate $X$ and Morawetz estimates of \cite{aretakis1}. We can also use this estimate even if $\Sigma$ (and $\mathcal{R}$) crosses $\mathcal{H}^{+}$, provided we have commuted the wave equation with $N$ and $NN$ (recall that we  need commutation with $NN$ only for degenerate black holes).


\begin{thebibliography}{99}

\bibitem{aretakis1} S. Aretakis, \textit{The Wave Equation on  Extreme Reissner-Nordstr\"{o}m Black Hole Spacetimes I: Stability and Instability Results}, to appear in  Comm. Math. Phys.

\bibitem{aretakis3} S. Aretakis, \textit{The Price Law for Self-Gravitating Scalar Fields On Extreme Black Hole Spacetimes}, in preparation

\bibitem{aretakis4} S. Aretakis, \textit{Decay of axisymmetric solutions to the wave equation on extreme Kerr}, preprint

\bibitem{blu1} P. Blue and A. Soffer, \textit{Phase space analysis on some black hole manifolds},  Journal of Functional Analysis (2009), Volume \textbf{256}, Issue 1, 1--90



\bibitem{christab} D. Christodoulou and S. Klainerman, \textit{The Global Nonlinear Stability of the Minkowski Space}, Princeton University Press 1994

\bibitem{schris} D. Christodoulou, \textit{On the global initial value problem and the issue
of singularities}, Classical Quantum Gravity \textbf{16} (1999), No. 12A, A23--A35


\bibitem{chrin} D. Christodoulou, \textit{The instability of naked singularities in the gravitational
collapse of a scalar field}, Ann. of Math. \textbf{149} (1999), No. 1, 183--217
 
\bibitem{christodoulou_actionprinciple}
  D. Christodoulou,
  \emph{The Action Principle and Partial Differential Equations},
  Princeton University Press, New Jersey,  2000.
  
 
  
  \bibitem{formblackholes} D. Christodoulou, \textit{The Formation of Black Holes in General Relativity}, Zurich: European Mathematical Society Publishing House (2009)
  
  \bibitem{chru} P. Chru\'{s}ciel and L. Nguyen, \textit{A uniqueness theorem for degenerate Kerr-Newman black holes},  arXiv:1002.1737
  
\bibitem{d1} M. Dafermos, \textit{Stability and instability of the Cauchy horizon for the spherically
symmetric Einstein-Maxwell-scalar field equations}, Ann. of Math.\textbf{ 158}
(2003), No. 3, 875--928

  
\bibitem{price} M. Dafermos and I. Rodnianski, \textit{A proof of Price' s law for the collapse of a
self-gravitating scalar field}, Invent. Math. \textbf{162} (2005), 381--457


\bibitem{dr3} M. Dafermos and I. Rodnianski,
\emph{The redshift effect and radiation decay on black hole
spacetimes}, Comm. Pure  Appl. Math. {\bf 62} (2009), 859--919

\bibitem{dr7}M. Dafermos and I. Rodnianski, \textit{A proof of the uniform boundedness of solutions to the wave equation on slowly rotating Kerr backgrounds},  Invent. Math.  (2011)

\bibitem{md}M. Dafermos and I. Rodnianski,
\emph{Lectures on Black Holes and Linear Waves}, arXiv:0811.0354



\bibitem{new}M. Dafermos and I. Rodnianski, \textit{A new physical-space approach to decay for the wave equation with applications to black hole spacetimes}, arXiv:0910.4957

\bibitem{mikraa}M. Dafermos and I. Rodnianski, \textit{Decay for solutions of the wave equation on Kerr exterior spacetimes $I-II$: The cases $\left|a\right|\ll M$} or axisymmetry, arXiv:1010.5132 

\bibitem{megalaa} M. Dafermos and I. Rodnianski, \textit{The black holes stability problem for linear scalar perturbations},  arXiv:1010.5137

\bibitem{drimos} R. M. Wald, \emph{Note on the stability of the 
Schwarzschild metric} J. Math. Phys. {\bf 20} (1979), 1056--1058


\bibitem{wald} R. M. Wald, \textit{General Relativity}, The University of Chicago Press, 1984

\bibitem{wa1}  R. Wald and B. Kay, \textit{Linear stability of Schwarzschild under perturbations
which are nonvanishing on the bifurcation 2-sphere}, Classical Quantum
Gravity \textbf{4} (1987), No. 4, 893--898



 
\end{thebibliography}
\end{document}